\documentclass[10pt,final,letterpaper,twocolumn,twoside,romanappendices]{IEEEtran}

\usepackage{mathdefs}
\usepackage{amsmath}
\usepackage{amssymb}
\usepackage{graphicx} 
\usepackage{float}
\usepackage{algpseudocode}
\usepackage{algorithm}

\newtheorem{remark}{Remark}
\newtheorem{corollary}{Corollary}

\newtheorem{lemma}{Lemma}
\newtheorem{proposition}{Proposition}

\usepackage{hyperref}
\hypersetup{colorlinks=false}
\hypersetup{colorlinks,citecolor=black,filecolor=black,linkcolor=black,urlcolor=black}

\newtheorem{assumption}{Assumption}

\newcommand{\sizecorr}[1]{\makebox[0cm]{\phantom{$\displaystyle #1$}}}

\title{Delay on broadcast erasure channels under \\ random linear combinations}

\author{
Nan~Xie,~\IEEEmembership{Member,~IEEE,} and 
Steven Weber,~\IEEEmembership{Senior Member,~IEEE}%
\thanks{
Manuscript received October 1, 2013; revised July 25, 2016; accepted November 03, 2016. Date of current version November 18, 2016.  This work was supported by the National Science Foundation under award CCF-1016588. A preliminary version of this work was presented at the Information Theory and its Applications (ITA) Workshop in San Diego, CA in February 2013 \cite{XieWeb2013}.}%
\thanks{N.~Xie and S.~Weber are with the Department of Electrical and Computer Engineering of Drexel University, Philadelphia, PA, 19104 USA (e-mail: nx23@drexel.edu; sweber@coe.drexel.edu). }
\thanks{Communicated by C.\ Nair, Associate Editor for Shannon Theory.}
\thanks{Color versions of one or more of the figures in this paper are available online at http://ieeexplore.ieee.org.}
\thanks{This paper has been accepted for publication by IEEE Transactions on Information Theory.  The DOI is 10.1109/TIT.2016.2634007, and copyright has been transferred to IEEE.  An (early access) version of this article is available from IEEE at http://ieeexplore.ieee.org/document/7762932/.}
}

\begin{document}
\maketitle

\begin{abstract}
We consider a transmitter broadcasting random linear combinations (over a field of size $d$) formed from a block of $c$ packets to a collection of $n$ receivers, where the channels between the transmitter and each receiver are independent erasure channels with reception probabilities $\qbf = (q_1,\ldots,q_n)$.  We establish several properties of the random delay until all $n$ receivers have recovered all $c$ packets, denoted $Y_{n:n}^{(c)}$.  First, we provide lower and upper bounds, exact expressions, and a recurrence for the moments of $Y_{n:n}^{(c)}$.  Second, we study the delay per packet $Y_{n:n}^{(c)}/c$ as a function of $c$, including the asymptotic delay (as $c \to \infty$), and monotonicity (in $c$) properties of the delay per packet.  Third, we employ extreme value theory to investigate $Y_{n:n}^{(c)}$ as a function of $n$ (as $n \to \infty$).  Several results are new, some results are extensions of existing results, and some results are proofs of known results using new (probabilistic) proof techniques.  
\end{abstract}


\section{Introduction} 
\label{sec:introduction}

The focus of this paper is on the number of time slots required to broadcast a collection of $c$ packets to $n$ receivers, where the channels between the transmitter and each receiver are independent erasure channels with reception probabilities $\qbf = (q_1,\ldots,q_n)$ (Fig.\ \ref{fig:1}).  In particular, the random delay associated with the transmission is the number of elapsed time slots until each receiver has all $c$ packets, when the transmitter forms random linear combinations of the $c$ packets over a (finite or infinite) field of size $d$, denoted $Y_{n:n}^{(c)}$.  The focus of our investigation is primarily, although not exclusively, on deriving properties (e.g., exact expressions, lower and upper bounds, asymptotics) of the $r^{\rm th}$ moment of $Y_{n:n}^{(c)}$, denoted $\Ebb\left[\left(Y_{n:n}^{(c)}\right)^r\right]$.  The proof techniques we employ are almost entirely probabilistic.  

\begin{figure}[!ht]
\centering
\includegraphics[width=0.4\textwidth]{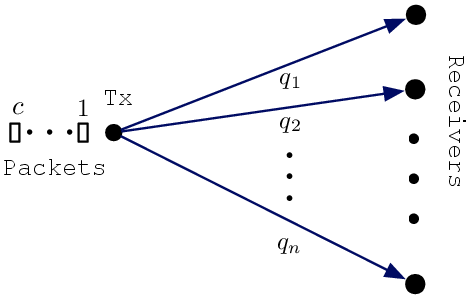}
\caption{The heterogeneous broadcast erasure channel consists of a transmitter and $n$ receivers connected over independent erasure channels with reception probabilities $\qbf = (q_1,\ldots,q_n)$.  Of interest is the time required for all $n$ receivers to receive all $c$ packets, denoted $Y_{n:n}^{(c)}$.  The homogeneous channel has $q_j =q$ for $j \in \{1, \ldots, n\}$.}
\label{fig:1}
\end{figure}

\subsection{Motivation and related work}

The broad motivation and context for this work is the fact that the use of random linear combinations of packets as a coding paradigm has received a great deal of attention within both the network coding (see, e.g., \cite{HoMed2006} and subsequent work) and fountain coding  (see, e.g., \cite{ByeLub2002} and subsequent work) communities.  More specifically, there have been a number of recent (since 2006) works focused on the broadcast delay of random linear combinations over erasure channels, including \cite{EryOzd2008,CogShr2011a,CogShr2011b,CogShr2011c,YanShr2012,SwaEry2013}.  We now relate the contributions of this paper within the context of this body of work.  Several of our results are new, some results are extensions of existing results, and some results are new proofs of known results using new (probabilistic) proof techniques.  

To our knowledge the earliest work on broadcast delay using random linear combinations over erasure channels is that of Eryilmaz, Ozdaglar, M\'{e}dard, and Ahmed \cite{EryOzd2008} (with a preliminary conference version in \cite{EryOzd2006}).  They compare the delays under scheduling with (or without) channel state information vs.\ random linear coding and establish the superiority of the latter, and also establish explicit expressions for the delays under random linear coding, along with asymptotic expressions in the number of receivers.  We extend and reprove via alternate techniques a couple of their results.

The work of Cogill, Shrader, and Ephremides \cite{CogShr2011a}, and Cogill and Shrader \cite{CogShr2011b,CogShr2011c} address (variously) throughput, delay, and stability of multicast queueing systems over erasure channels using random linear packet coding.  Of these, the work closest to ours is \cite{CogShr2011a}.  Their focus in this work is primarily on the stability region of arrival rates for multicast queueing systems, and in establishing their results they obtain several results on the broadcast delay. Again, we extend and reprove via alternate techniques a couple of their results.  

Recent work by Yang and Shroff \cite{YanShr2012} has extended the channel model to the Markov-modulated erasure channel (allowing correlations in time).  Additionally, Swapna, Eryilmaz, and Shroff \cite{SwaEry2013} have studied extensions of this framework to address the case where the blocklength $c$ scales with $n$, the number of receivers, and in particular they show the existence of a phase-transition at $c(n) = \Theta(\log n)$.  Our work does not address either of these extensions --- we restrict our attention to erasure channel realizations that are independent in time (and across users), and we do not address the asymptotic regime when the blocklength $c$ grows with the number of users $n$.  

Besides our contributions extending certain results in the above work, we have also investigated the following additional topics that, to our knowledge, have not been previously addressed.  First, a common theme throughout our work is on lower and upper bounds on each of the moments of the delay, which hold for all finite $c,n$, and which we additionally show to be (almost) asymptotically tight for fixed $n$ as $c$ grows large, and for fixed $c$ as $n$ grows large.  Second, we establish the intuitive fact that the expected delay per packet is decreasing in the blocklength $c$.  Although this fact is intuitive, the proof is non-trivial, and in fact we show some (perhaps) counter-intuitive results giving necessary and sufficient conditions on the sample-path realizations of the delay per packet to be decreasing as the blocklength is increased.  Third, we employ extreme value theory and stochastic ordering in studying the asymptotic behavior of the delay as the number of receivers $n$ is increased.  As will become evident, it is natural to use these two tools together.  

\subsection{Outline}

The paper is structured as follows.  The model and common notation are introduced in \S\ref{sec:modelnotation}, and \S\ref{sec:delayUT} analyzes the delay without coding, i.e., when each packet in the block is repeatedly broadcast until received by all receivers.  The next three sections form the heart of the paper.  First, in \S\ref{sec:delayRLNC} we address the delay under random linear combinations, with subsections for $i)$ lower and upper bounds, $ii)$ exact expressions, $iii)$ a recurrence for the delay, and $iv)$ a characterization of the channels that minimize the delay.  Second, in \S\ref{sec:deponbl} we address the delay per packet as a function of the blocklength $c$, with subsections for $i)$ the asymptotic (in $c$) delay per packet, $ii)$ monotonicity (in $c$) properties of the delay per packet, and $iii)$ bounds on the expected delay per packet that are (almost) asymptotically tight in $c$. Third, in \S\ref{sec:deponnumrx} we address the delay as a function of the number of receivers $n$, where our approach couples (extreme) order statistic inequalities with stochastic ordering and extreme value theory to establish that the bounds on delay are asymptotically tight in $n$.  A brief conclusion is offered in \S\ref{sec:conclusion}. Several appendices follow the references, holding long proofs from \S\ref{sec:delayUT}, \S\ref{sec:delayRLNC}, \S\ref{sec:deponbl}, and \S\ref{sec:deponnumrx} respectively.  Table \ref{tab:summaryofresults} contains a summary of the results in the paper (refer to \S\ref{sec:modelnotation} for notation), where the governing assumption for each result (see Assumption \ref{assum:GoverningAssumptions} in \S\ref{ssec:discdisbn}), if applicable, is indicated.

\begin{table}
\centering
\caption{Summary of results}
\begin{tabular}{cll}
\S$\textbf{\#}$/Result & {\bf Title}/Description \\ \hline \hline
{\bf \S\ref{sec:modelnotation}} & {\bf Model and notation} \\ 
Prop.\ \ref{prop:stochorder} & Stochastic orderings under different assumptions \\ \hline \hline
{\bf \S\ref{sec:delayUT}} & {\bf Delay under uncoded transmission} \\ 
Prop.\ \ref{prop:boundsOnEZ} & Bounds on $\Ebb[X_{n:n}^r]$ (A2) \\
Cor.\ \ref{cor:1} & Bounds on $\Ebb[X_{n:n}^r]$ (A3) \\ \hline 
Prop.\ \ref{prop:kthMomentGeo} & Exact $\Ebb[X^r]$ \\
Prop.\ \ref{prop:kthMomentMaxGeo} & Exact $\Ebb[X_{n:n}^r]$ (A2) \\
Cor.\ \ref{cor:firstTwoMomentsMaxGeo} & Exact $\Ebb[X_{n:n}]$ and $\Ebb[X_{n:n}^2]$ (A2) \\ \hline \hline 
{\bf \S\ref{sec:delayRLNC}} & {\bf Delay under random linear combinations} \\  
Prop.\ \ref{prop:boundsOnEY1} & Bounds on $\Ebb\left[ \left(Y_{n:n}^{(c)}\right)^r \right]$ (A1) \\
Prop.\ \ref{prop:boundsOnEY2}& Bounds on $\Ebb\left[ \left(Y_{n:n}^{(c)}\right)^r \right]$ (A2) \\
Prop.\ \ref{prop:boundsOnEY3} & Bounds on $\Ebb\left[ \left(Y_{n:n}^{(c)}\right)^r \right]$ (A3) \\ \hline
Prop.\ \ref{prop:exacty} & Exact $\Ebb\left[ \left(Y_{n:n}^{(c)}\right)^r \right]$ (A1,A2) \\ \hline
Prop.\ \ref{pro:recexp} & Recurrence for $Y_{n:n}^{(\cbf)}$ (A1,A2) \\
Cor.\ \ref{cor:exprec} & Recurrence for $\Ebb[Y_{n:n}^{(\cbf)}]$ (A1,A2) \\ \hline
Prop.\ \ref{prop:extremaldelayUT} & Minimization of $\Ebb[X_{n:n}^r]$ (A2)\\
Prop.\ \ref{prop:extremaldelayRLNC} & Minimization of $\Ebb\left[ \left(Y_{n:n}^{(c)}\right)^r \right]$ (A2) \\ \hline \hline
{\bf \S\ref{sec:deponbl}} & {\bf Dependence of RLC delay on the blocklength $c$} \\
Prop.\ \ref{prop:convergence} & Convergence in $c$ of $\{Y_{n:n}^{(c)}/c\}$ (A1) \\
Prop.\ \ref{prop:asymcondsat} & Satisfying conditions for Prop.\ \ref{prop:convergence} (A1) \\
Prop. \ref{prop:comparingAsymptotic} & Inequality $\lim_{c \to \infty} \Ebb[Y_{n:n}^{(c)}]/c \leq \Ebb[X_{n:n}]$ (A1) \\ \hline
Prop.\ \ref{prop:monotonicityHomogeneous} & $\Ebb[Y_{n:n}^{(c)}]/c$ monotone decreasing in $c$ (A2) \\
Prop.\ \ref{prop:monotonicityHeterogeneous} & $\Ebb[Y_{n:n}^{(c)}]/c$ monotone decreasing in $c$ (A1) \\
Cor.\ \ref{cor:blockpartitionopt} & Optimal sum block delay over block partitions (A1) \\
Prop.\ \ref{prop:nostochasticordering} & $\frac{1}{c}Y^{(c)}_{j}$ and $\frac{1}{c+1}Y^{(c+1)}_{j}$ not stochastically ordered (A2) \\
Prop.\ \ref{prop:monoSamplePathCombined} & Sample path monotonicity in $c$ of $T_{n:n}^{(c,m)}(\omega)$ (A2) \\ \hline
Prop.\ \ref{prop:LBandUBinCHomoRxInftyFieldsize} & Asymptotically (in $c$) tight bounds on $\Ebb[Y_{n:n}^{(c)}]/c$ (A3) \\ \hline \hline
{\bf \S\ref{sec:deponnumrx}} & {\bf Dependence of RLC delay on the num.\ of receivers $n$} \\
Cor.\ \ref{cor:lowupboundsnegbinevt} & Bounds on $\lim_{n \to \infty}  \Ebb \left[ \left(\phi(q) Y_{n:n}^{(c)} - b_n\right)^r \right]$ (A3) \\
Lem. \ref{lem:momentConvergence} & Standardized and non-standardized asymptotic moments \\
Prop.\ \ref{prop:scalingOfYnnWRTn} & Asymptotic (in $n$) scaling of $\Ebb\left[ \left( Y_{n:n}^{(c)}\right)^{r} \right]$ (A3) \\ \hline
Prop.\ \ref{prop:dlcasymptoticinn} & Asym.\ (in $n$) tight lower bound on $\Ebb\left[\left(\tilde{Y}_{n:n}^{(c)}\right)^{r}\right]$ (A3) \\
Prop.\ \ref{prop:rossasymptoticinn} & Asym.\ (in $n$) tight upper bound on $\Ebb\left[\left(\tilde{Y}_{n:n}^{(c)}\right)^{r}\right]$ (A3) \\ \hline \hline
{\bf Appendices} & \\ \hline
Lem.\ \ref{lem:regbetamono} & $I_x(a,b)$ increasing in $b$ \\
Lem.\ \ref{lem:regbetalogconcave} & $\log I_{x}(a,b)$ concave on $x \in (0,1)$ \\
Prop.\ \ref{prop:LogConcaveSuffConditions} & Conditions ensuring log-concavity of $1 - F_{Z}^{n-1}(t)$ \\
Prop.\ \ref{prop:GammaLogConcave} & $1 - F_{Z}^{n-1}(t)$ log-concave for $Z \sim \mathrm{Gamma}(c,1)$ \\
Prop.\ \ref{prop:delaCalOptimizationGammaRV} & Optimization of de la Cal's bound \eqref{eq:delaCallb} \\ \hline
Lem.\ \ref{lem:randind} & Column invariant row-index selection rule (A2) \\ 
Lem.\ \ref{lem:tuwerty} & $\Ebb\left[ \max_{j} \sum_{k=1}^{c} Z_{j,k}^{(c+1)} \right] > \Ebb\left[ \sum_{k=1}^{c} Z_{\hat{J}^{(c+1)},k}^{(c+1)} \right]$ (A2) \\
Lem.\ \ref{lem:markovchainselectionrule} & $(\hat{J},X_{j,k}\!+\!X_{j,k'},X_{j,k},X_{j,k'})$ form Markov Chains (A1) \\
Lem.\ \ref{lem:StochasticOrderingOfGeoRVs} & Stochastic ordering of geometric RVs (A1) \\ \hline
Prop.\ \ref{prop:rossdlcexpexample} & Lower and upper moment bounds for iid exponentials \\ \hline \hline
\end{tabular}
\label{tab:summaryofresults}
\end{table} 

\section{Model and notation} 
\label{sec:modelnotation}

In this section we introduce the notation and the model.  Relevant discrete distributions are covered in \S\ref{ssec:discdisbn} and continuous distributions in \S\ref{ssec:contdisbn}.  In many situations (e.g., bounding a moment of the random delay) it is more convenient to work with the ``associated'' continuous random variable (RV), where the association is through the existence of a stochastic ordering, as will be described in \S\ref{ssec:contdisbn}.

Our notational convention is to denote the set of natural numbers (i.e., positive integers) by $\Nbb \equiv \{1, 2, \ldots \}$. For any $n \in \Nbb$, we write $[n]$ to denote the set $\{1, \ldots, n\}$. $\mathbf{1}_{S}$ is either an indicator function if $S$ is a boolean function (event), or a binary vector with $S$ indicating the set of indices taking $1$.  The natural logarithm is denoted $\log(\cdot)$. Our convention for the geometric RV is such that it denotes the number of independent Bernoulli trials needed to get the first ``success'' (hence the support $\{1, 2, \ldots \}$) and is parameterized by $q$, the success probability. $\Pbb(\cdot)$ and $\Ebb[\cdot]$ denote the probability and expectation respectively. Superscripts enclosed in parentheses of a RV indicate the corresponding blocklength. $Z \sim F_Z$ indicates the RV $Z$ has a cumulative distribution function (CDF) $F_Z$. $Z_1 \sim Z_2$ and $Z_1 \stackrel{d}{=} Z_2$ for RVs $Z_1,Z_2$ both mean they are equal in distribution. Table \ref{tab:notationgeneral} lists frequently used notation, while Table \ref{tab:notationdisbn} lists notation for the specific distributions in the paper; additional notation will be explained at first use.

\begin{table}[!ht]
\centering
\caption{General notation}
\begin{tabular}{ll}
Symbol & Meaning \\ \hline\hline
$n$, $[n]$ & number and set of receivers \\
$[n]_s$ & set of all subsets of $[n]$ of size $s \in [n]$ \\
$c$, $[c]$ & blocklength, and set of packets per block \\ \hline
$\qbf = (q_{j \in [n]})$ & reception probabilities (A1, A2) \\
$q_{j \in [n]} = q$ & reception probabilities (A3) \\
$\Qbf = (q_{j \in [n], k \in [c]})$ & success probability matrix (A1) \\ 
$q_{j \in [n],k \in [c]} = q_j$ & success probabilities (A2) \\
$q_{j \in [n],k \in [c]} = q$ & success probabilities (A3) \\ \hline
$\phi(q)$ & rate parameter (\S\ref{sec:modelnotation}) of the continuous analog RV \\
$d$ & field size for coeff.\ used in linear combinations \\ \hline
$H_n$ & $n^{\rm th}$ harmonic number \\
$\binom{m}{k}$ & binomial coefficient \\
${m \brace k}$ & Stirling number of the second kind \\
$Z,\tilde{Z}$ & a generic discrete/continuous RV \\
$Z_{n:n},\tilde{Z}_{n:n}$ & maximum order statistic of $(Z_{j \in [n]})$, $(\tilde{Z}_{j \in [n]})$ \\
$F_Z(z),F_{\tilde{Z}}(z) $ & cumulative distribution function (CDF) for $Z,\tilde{Z}$ \\
$p_Z(z),f_{\tilde{Z}}(z)$ & prob.\ mass/density function (PMF/PDF) for $Z,\tilde{Z}$ \\
$\Ebb[Z^r],\Ebb[\tilde{Z}^r]$ & $r^{\rm th}$ ($r \in \Nbb$) moment of $Z,\tilde{Z}$ \\ \hline\hline
\end{tabular}
\label{tab:notationgeneral}
\end{table}

\begin{table}[!ht]
\centering
\caption{Probability distributions}
\begin{tabular}{rl}
Symbol & Meaning \\ \hline\hline
$X \sim \mathrm{Geo}(q)$ & geometric RV with success prob.\ $q$ \\
$\tilde{X} \sim \mathrm{Exp}(1)$ & exponential RV with rate $1$ \\
$\frac{1}{\phi}\tilde{X} \sim \mathrm{Exp}(\phi)$ & exponential RV with rate $\phi$ \\ \hline
$Y_j^{(c)} = \sum_{k=1}^{c} X_{j,k}$ & time required for $c$ successes at receiver $j$ \\
$Y_j^{(c)} \sim \qquad \qquad \quad \; $ & \\
$\mathrm{GenNegBin}(c,\qbf_j)$ & generalized negative binomial RV (A1) \\
$\mathrm{NegBin}(c,q_j)$ & negative binomial RV (A2) \\ 
$\mathrm{NegBin}(c,q)$ & negative binomial RV (A3) \\ \hline
$\tilde{Y}^{(c)} \sim \mathrm{Gamma}(c,1)$ & sum of $c$ iid unit rate exponential RVs \\
$\frac{1}{\phi} \tilde{Y}^{(c)} \sim \mathrm{Gamma}(c,\phi)$ & sum of $c$ iid exponential RVs with rate $\phi$ \\
$\tilde{Y}_j^{(c)} = \sum_{k=1}^{c} \tilde{X}_{j,k}$ & continuous analog of $Y_j^{(c)}$ \\
$\tilde{Y}_j^{(c)} \sim \qquad \qquad \quad \;$ & \\
$\mathrm{HypoExp}(c,\phi(\qbf_j))$ & hypo-exponential RV (A1) \\
$\mathrm{Gamma}(c,\phi(q_j))$ & Gamma RV (A2) \\ 
$\mathrm{Gamma}(c,\phi(q))$ & Gamma RV (A3) \\ \hline
$W \sim \mathrm{Bin}(m,q)$ & binomial RV with $m$ trials \& succ.\ prob.\ $q$ \\
$I_x(\alpha,\beta)$ & regularized incomplete beta function \\
$I_q(l,m-l+1)$ & CCDF of $W$ evaluated at $l$: $\Pbb(W \geq l)$\\
\hline\hline
\end{tabular}
\label{tab:notationdisbn}
\end{table} 

\subsection{Discrete distributions}
\label{ssec:discdisbn}

Index the transmission slots by $\Nbb$, and define the (random) block delay to be the total number of elapsed slots until each of the $n$ receivers has successfully decoded the block of $c$ packets.  We consider two separate transmission schemes: $i)$ uncoded transmission (\S\ref{sec:delayUT}), abbreviated UT, and $ii)$ random linear combinations (\S\ref{sec:delayRLNC} through \S\ref{sec:deponnumrx}), abbreviated RLC.

Under UT, the transmitter in effect treats each of the $c$ packets as a separate block, and repeatedly (re-)transmits each packet until all receivers have that packet, at which time it moves on to the next packet in the block of $c$ packets, if any.  The presumption is that a feedback channel exists which alerts the transmitter that all $n$ receivers have the packet.  Due to the standing independence assumptions, it is clear that the overall random delay to transmit all $c$ packets is the sum of $c$ independent and identically distributed (iid) RVs, each representing the time (in slots) for the transmitter to complete one of the $c$ packets.  In particular, define the random delay per packet under UT as the maximum of $n$ geometric RVs, $X_{n:n} \equiv \max(X_1,\ldots,X_n)$, with $(X_1,\ldots,X_n)$ independent and $X_j \sim \mathrm{Geo}(q_j)$ representing the delay per packet under UT for receiver $j$, i.e., the number of transmission attempts until receiver $j$ is successful.  

Under RLC, the transmitter forms in each time slot a new random linear combination of the $c$ (original, information) packets, with coefficients generated uniformly at random from a finite or infinite field of size $d$, namely $\{0,\ldots,d-1\}$.  This combination of packets (i.e., the encoded packet) is then broadcast to the receivers; the block delay under RLC is the number of time slots until all receivers have decoded all $c$ packets. 
In particular, each receiver is able to recover the packets after receiving $c$ linearly independent combinations.  This is because $c$ receptions are required for the matrix of coding vectors, formed by stacking the vectors of coefficients used in each combination, to have full rank, and therefore to be invertible \cite{HoMed2006,HoLun2008,Fra2011}.  To be sure, these coefficients constitute a source of overhead not found in UT, but this overhead can be amortized by scaling the packet size.  Alternately, so-called ``non-coherent network coding'', using the concept of vector space or subspace coding \cite{KoeKsc2008}, ameliorates the need for the coefficients to be incorporated in the packet header.  At any rate, our interest lies in the delay and not in the packet encoding overhead, and as such we ignore the RLC overhead relative to UT.

More precisely, the block delay for a block of $c$ packets under RLC is denoted by $Y_{n:n}^{(c)} \equiv \max \left(Y_1^{(c)},\ldots,Y_n^{(c)} \right)$.  The expected block delay is $\Ebb\left[Y_{n:n}^{(c)}\right]$ and the expected delay per packet is $\Ebb\left[Y_{n:n}^{(c)}\right]/c$.  Each $Y_{j}^{(c)} = \sum_{k = 1}^{c} X_{j,k}$ is a generalized negative binomial RV, denoted $Y_j^{(c)}\sim \mathrm{GenNegBin}(c,\qbf_j)$ with parameter $\qbf_j = (q_{j,1},\ldots,q_{j,c})$, for independent geometric RVs $(X_{j,1},\ldots,X_{j,c})$ with $X_{j,k} \sim \mathrm{Geo}(q_{j,k})$.  We say that a receiver $j$ is in \textit{state} $k \in [c]$ if it has already received a maximum of $k-1$ linearly independent combinations from the transmitter.  Then $X_{j,k}$ represents the duration for which receiver $j$ stays in state $k$, and is given by the elapsed time slots between receiver $j$ obtaining the $(k-1)^{\rm th}$ and $k^{\rm th}$ linearly independent combinations.  We reiterate that $k$ is not a time index, per se, although $k = k(t)$ is nondecreasing in $t$.  That $\left(X_{j,k}\right)$ are independent in $j$ ($k$) is due to the assumption that the erasure channels are independent in space (time), respectively.  The $n \times c$ matrix $\Qbf$ with entries $q_{j,k}$ for $j \in [n]$ and $k \in [c]$ holds the success probability indexed by receiver $j$ at state $k$:
\begin{equation}
\label{eq:tyt}
q_{j,k} = (1-d^{k-1-c}) q_j, ~~ j \in [n], ~ k \in [c].
\end{equation}
Here, $q_j$ is the time-invariant channel reception probability for receiver $j$, $d$ is the field size, and $1-d^{k-1-c}$ is the probability that the linear combination sent by the transmitter is in fact independent of the $k-1$ linear combinations already received by receiver $j$ (e.g., \cite[Lemma 1]{CogShr2011a}).  Eq.\ \eqref{eq:tyt} is the most general case, which we occasionally specialize for tractability, as indicated below.
\begin{assumption}
\label{assum:GoverningAssumptions}
Throughout, we assume one of the three cases listed below, in order of decreasing generality.
\begin{itemize}
\itemsep=-2pt
\item[A1:] {\bf State-dependent receptions, heterogeneous receivers.} The field size $d$ is finite and the reception probabilities $\qbf = (q_1,\ldots,q_n)$ are unrestricted (heterogeneous).  The success probabilities are given by the $n \times c$ matrix $\Qbf$ with entries $q_{j,k}$ in \eqref{eq:tyt}.  The random block delay $Y_{n:n}^{(c)}$ is the maximum of $n$ independent generalized negative binomial RVs $Y_j^{(c)} \sim \mathrm{GenNegBin}(c,\qbf_j)$, with $\qbf_j = (q_{j,1},\ldots,q_{j,c})$, for $j \in [n]$.  
 
\item[A2:] {\bf State-independent receptions, heterogeneous receivers.} The field size $d$ is infinite, but the reception probabilities $\qbf = (q_1,\ldots,q_n)$ are unrestricted (heterogeneous).  Due to the infinite field size, received linear combinations are linearly independent of all previous combinations, and so the success probabilities equal the reception probabilities, which are independent of $k$, i.e., $q_{j,k} = q_j$ for $j \in [n]$ and all $k \in [c]$.  The success probabilities are given by $\qbf = (q_j, j \in [n])$.  The random block delay $Y_{n:n}^{(c)}$ is the maximum of $n$ independent negative binomial RVs $Y_j^{(c)} \sim \mathrm{NegBin}(c,q_j)$, for $j \in [n]$.

\item[A3:] {\bf State-independent receptions, homogeneous receivers.} The field size $d$ is infinite, and the reception probabilities are homogeneous, i.e., $q_j = q$ for all $j \in [n]$.  Due to the infinite field size, received linear combinations are linearly independent of all previous combinations, and so  the success probabilities equal the reception probability, which is  independent of $k$ and $j$, i.e., $q_{j,k} = q$ for all $j \in [n]$ and all $k \in [c]$.  The random block delay $Y_{n:n}^{(c)}$ is the maximum of $n$ iid negative binomial RVs $Y_j^{(c)} \sim \mathrm{NegBin}(c,q)$, for $j \in [n]$.
\end{itemize}
\end{assumption}
Throughout, when necessary we will indicate the governing assumption.  To reiterate, A1 is the most general assumption, A2 is a special case of A1 for $d = \infty$, and A3 is a special case of A2 for $q_{j \in [n]} = q$.  The governing assumption (if applicable) for each result in the paper is also listed in Table \ref{tab:summaryofresults}.  

It is worth noting that, in general, setting $c=1$ in RLC does not recover UT as a special case.  In particular, for $d < \infty$, setting $c=1$ gives $q_{j,1} = (1-d^{-1}) q_j \neq q_j$, since we are forming linear combinations over a block of one packet, for which there is a probability of selecting the $0$ coefficient, which is (trivially) linearly dependent upon previous receptions.  Naturally, we {\em do} recover UT as a special case of RLC for $c=1$ and $d=\infty$.  

\subsection{Continuous distributions}
\label{ssec:contdisbn}

Although the (discrete) geometric and negative binomial distributions directly capture the discrete delay of interest to us, nonetheless we will often find it useful to consider what we call the {\em continuous analogs} of these distributions.  The lynchpin connecting the discrete RVs to their continuous RV analogs is the notion of \textit{stochastic ordering}, the basic concepts of which we briefly review below, drawing directly from Ross \cite[Chapter 9]{Ros1996}.  As described below, stochastic ordering is preserved under positive scaling, translation, and component-wise nondecreasing functions (with independent components) and more importantly implies moment ordering.  Collectively, these properties allow us to establish inequalities on the $r^{\rm th}$ moment of a discrete RV in terms of the $r^{\rm th}$ moment of its (often more tractable) continuous analog.

Generic scalar (continuous or discrete) RVs $Z_1,Z_2$ are said to be stochastically ordered, denoted $Z_1 \leq_{\rm st} Z_2$, if for all $z$: $\bar{F}_{Z_1}(z) \leq \bar{F}_{Z_2}(z)$ for $\bar{F} = 1 - F$ the complementary CDF (CCDF) of the RV.  Stochastic ordering implies moment ordering, i.e., $Z_1 \leq_{\rm st} Z_2$ implies $\Ebb[Z_1] \leq \Ebb[Z_2]$ (\cite[Lemma 9.1.1]{Ros1996}).  If we instead let $Z_1,Z_2$ denote random $n$-vectors with {\em independent} stochastically ordered components, i.e., $Z_1 = (Z_{1,1},\ldots,Z_{1,n})$ and $Z_2 = (Z_{2,1},\ldots,Z_{2,n})$ each have independent components and $Z_{1,j} \leq_{\rm st} Z_{2,j}$ for all $j \in [n]$, then the stochastic ordering is preserved under any multivariate nondecreasing function $f$, i.e., $f(Z_1) \leq_{\rm st} f(Z_2)$ (\cite[Example 9.2(A)]{Ros1996}).  Since the functions $f(z) = z^r$ ($r > 0$) for nonnegative $z$ and $f(z_1,\ldots,z_n) = \max(z_1,\ldots,z_n)$ are both nondecreasing, it follows that $\max(Z_{1,1},\ldots,Z_{1,n}) \leq_{\rm st} \max(Z_{2,1},\ldots,Z_{2,n})$, $\max(Z_{1,1},\ldots,Z_{1,n})^{r} \leq_{\rm st} \max(Z_{2,1},\ldots,Z_{2,n})^{r}$, and thus $\Ebb[\max(Z_{1,1},\ldots,Z_{1,n})^r] \leq \Ebb[\max(Z_{2,1},\ldots,Z_{2,n})^r]$, for any nonnegative integer $r$.

Throughout the paper we indicate the continuous RV matched to a discrete RV, say $Z$, by $\tilde{Z}$.  As summarized in Table \ref{tab:notationdisbn}, let $\tilde{X} \sim \mathrm{Exp}(1)$ denote a unit-rate exponential RV, and $\tilde{Y}^{(c)} \sim \mathrm{Gamma}(c,1) = \tilde{X}_1 + \cdots + \tilde{X}_c$ denote a Gamma RV, i.e., the sum of $c$ independent unit-rate exponential RVs.  For $\phi > 0$ it is easily seen that $\frac{1}{\phi} \tilde{X} \sim \mathrm{Exp}(\phi)$ and $\frac{1}{\phi} \tilde{Y}^{(c)} \sim \mathrm{Gamma}(c, \phi)$.  Thus there is no loss of generality in restricting our attention to unit rate exponentials and Gamma RVs, as the general case is obtained by scaling.  

Recall the discrete definitions of $X_{n:n} \equiv \max(X_1,\ldots,X_n)$ (for independent $X_j \sim \mathrm{Geo}(q_j)$) and $Y_{n:n}^{(c)} \equiv \max(Y_1^{(c)},\ldots,Y_n^{(c)})$ (for independent $Y_j^{(c)} = X_{j,1} + \cdots + X_{j,c}$ with $X_{j,k} \sim \mathrm{Geo}(q_{j,k})$).  Set $\phi(q) \equiv - \log (1-q)$, which may be viewed as the ``rate'' parameter of the exponential RV matched to $\mathrm{Geo}(q)$.  Define the continuous analog $\tilde{X}_{n:n} \equiv \max(\tilde{X}_1,\ldots,\tilde{X}_n)$ (for independent $\tilde{X}_j$), where $i)$ under Assumption A2 $\tilde{X}_j \sim \frac{1}{\phi(q_j)} \tilde{X} \sim \mathrm{Exp}(\phi(q_j))$, while $ii)$ under Assumption A3 $\tilde{X}_j \sim \frac{1}{\phi(q)} \tilde{X} \sim \mathrm{Exp}(\phi(q))$.  Similarly, we define the continuous analog $\tilde{Y}_{n:n}^{(c)} \equiv \max(\tilde{Y}_1^{(c)},\ldots,\tilde{Y}_n^{(c)})$ (for independent $\tilde{Y}_j^{(c)} = \tilde{X}_{j,1} + \cdots + \tilde{X}_{j,c}$ the sum of $c$ independent RVs $(\tilde{X}_{j,k},k \in [c])$).  We have: 
\begin{enumerate}
\item[$i)$] Under Assumption A1 $\tilde{X}_{j,k} \sim \frac{1}{\phi(q_{j,k})} \tilde{X} \sim \mathrm{Exp}(\phi(q_{j,k}))$, and $\tilde{Y}_j^{(c)} \sim \mathrm{HypoExp}(c,\phi(\qbf_j))$ for $\phi(\qbf_j) = (\phi(q_{j,k}), k \in [c])$; 
\item[$ii)$] Under Assumption A2 $\tilde{X}_{j,k} \sim \frac{1}{\phi(q_{j})} \tilde{X} \sim \mathrm{Exp}(\phi(q_{j}))$, and $\tilde{Y}_j^{(c)} = \frac{1}{\phi(q_j)} \tilde{Y}^{(c)} \sim \mathrm{Gamma}(c,\phi(q_j))$;
\item[$iii)$] Under Assumption A3 $\tilde{X}_{j,k} \sim \frac{1}{\phi(q)} \tilde{X} \sim \mathrm{Exp}(\phi(q))$, and $\tilde{Y}_j^{(c)} = \frac{1}{\phi(q)} \tilde{Y}^{(c)} \sim \mathrm{Gamma}(c,\phi(q))$.
\end{enumerate}
We now establish several stochastic orderings between discrete RVs and their continuous analogs.

\begin{proposition}
\label{prop:stochorder}
The following stochastic orderings hold for each $j \in [n]$:
\begin{enumerate}
\item[$i)$] Under Assumption A1, for RLC: $\tilde{Y}_{j}^{(c)} \sim \mathrm{HypoExp}(c, \phi(\qbf_{j}))$ and $Y_{j}^{(c)} \sim \mathrm{GenNegBin}(c, \qbf_{j})$:
\begin{equation}
\tilde{Y}_{j}^{(c)} \leq_{\rm st} Y_{j}^{(c)} \leq_{\rm st} \tilde{Y}_{j}^{(c)} + c.
\end{equation}

\item[$ii)$] Under Assumption A2, for UT: $\tilde{X}_j \sim \frac{1}{\phi(q_j)} \tilde{X} \sim \mathrm{Exp}(\phi(q_j))$ and $X_j \sim \mathrm{Geo}(q_j)$:
\begin{equation}
\frac{1}{\phi(q_{j})} \tilde{X} \leq_{\rm st} X_{j} \leq_{\rm st} \frac{1}{\phi(q_{j})} \tilde{X}+1.
\end{equation}

\item[$iii)$] Under Assumption A2, for RLC: $\tilde{Y}_j^{(c)} \sim \frac{1}{\phi(q_j)} \tilde{Y}^{(c)} \sim \mathrm{Gamma}(c,\phi(q_j))$ and $Y_{j}^{(c)} \sim \mathrm{NegBin}(c,q_{j})$:
\begin{equation}
\frac{1}{\phi(q_{j})} \tilde{Y}^{(c)} \leq_{\rm st} Y_{j}^{(c)} \leq_{\rm st}  \frac{1}{\phi(q_{j})} \tilde{Y}^{(c)} + c.
\end{equation}

\item[$iv)$] Under Assumption A3, for UT: $\tilde{X}_j \sim \frac{1}{\phi(q)} \tilde{X} \sim \mathrm{Exp}(\phi(q))$ and $X \sim \mathrm{Geo}(q)$:
\begin{equation}
\frac{1}{\phi(q)} \tilde{X} \leq_{\rm st} X \leq_{\rm st} \frac{1}{\phi(q)} \tilde{X}+1.
\end{equation}

\item[$v)$] Under Assumption A3, for RLC: $\tilde{Y}_j^{(c)} \sim \frac{1}{\phi(q)} \tilde{Y}^{(c)} \sim \mathrm{Gamma}(c,\phi(q))$ and $Y^{(c)} \sim \mathrm{NegBin}(c,q)$:
\begin{equation}
\frac{1}{\phi(q)} \tilde{Y}^{(c)} \leq_{\rm st} Y^{(c)} \leq_{\rm st}  \frac{1}{\phi(q)} \tilde{Y}^{(c)} + c.
\end{equation}
\end{enumerate}
The above stochastic orderings omit the dependence upon $j$ when the distributions are not dependent upon $j$.  Furthermore, all these stochastic orderings are preserved when the RVs are raised to the $r^{\rm th}$ power, and, after taking expectations of these powers, the ordering is preserved for the $r^{\rm th}$ moments. 
\end{proposition}

\begin{IEEEproof}
It suffices to prove cases $i)$ and $ii)$, since $iii)$ and $v)$ are special cases of $i)$, and $iv)$ is a special case of $ii)$.  First, we prove case $ii)$.  Let $X_j \sim \mathrm{Geo}(q_j)$ and recall $\tilde{X} \sim \mathrm{Exp}(1)$.  We first show $\frac{1}{\phi(q_j)} \tilde{X} \leq_{\rm st} X_{j}$.  Observe for any $x \geq 0$:
\begin{eqnarray}
\bar{F}_{\frac{1}{\phi(q_j)} \tilde{X}}(x) & = &\Pbb\left(\tilde{X} > \phi(q_j) x\right) = \erm^{-\phi(q_j)x} = (1-q_j)^{x} \nonumber \\
& \leq & (1-q_j)^{\lfloor x \rfloor} = \bar{F}_{X_j}(x),
\end{eqnarray}
where $\lfloor x \rfloor$ is the largest integer not exceeding $x$.  We next show $X_j \leq_{\rm st} \frac{1}{\phi(q_j)} \tilde{X}+1$:
\begin{eqnarray}
\bar{F}_{X_j}(x) & = & \Pbb(X_j > x) = (1-q_j)^{\lfloor x \rfloor} < (1-q_j)^{(x-1)} \nonumber \\
& = & \erm^{-\phi(q_j)(x-1)} = \bar{F}_{\frac{1}{\phi(q_j)} \tilde{X}+1}(x).
\end{eqnarray}
Second, we prove case $i)$.  The fact $\tilde{Y}_{j}^{(c)} \leq_{\rm st} Y_{j}^{(c)} \leq_{\rm st} \tilde{Y}_{j}^{(c)}+c$ follows from $1)$ the stochastic ordering of $\frac{1}{\phi(q_{j,k})} \tilde{X} \leq_{\rm st} X_{j,k} \leq_{\rm st} \frac{1}{\phi(q_{j,k})} \tilde{X}+1$, $2)$ the fact that the sum of independent nonnegative RVs is a nondecreasing function of a random vector composed of those random variables, and $3)$ the fact that stochastic ordering is preserved under nondecreasing functions of random vectors with independent stochastically ordered components (\cite[Example 9.2(A)]{Ros1996}).  
\end{IEEEproof}

\section{Delay under uncoded transmission} 
\label{sec:delayUT}

The $r^{\rm th}$ ($r \in \Nbb$) moment of delay under UT, $\Ebb\left[ X_{n:n}^{r} \right]$, is investigated in this section. We first give lower and upper bounds (\S\ref{ssec:utbounds}), then derive an exact closed-form expression using the (moment) generating function (\S\ref{ssec:utexact}), after which this section closes with further remarks on related work.  Throughout this section we assume A2: state-independent receptions, heterogeneous receivers.  State-independence follows from the fact that under UT there is no coding.

\subsection{Lower and upper bounds on the moments of delay under UT}
\label{ssec:utbounds}

Prop.\ \ref{prop:boundsOnEZ} below hinges upon the stochastic ordering between exponential and geometric RVs, and the following ``min-max identity''.  Define $[n]_s$ to be the set of all subsets of $[n]$ of size $s$.

\begin{proposition} [min-max identity, \cite{RosPek2007}, pp.\ 128] 
\label{prop:minmax}
For nonnegative (not necessarily independent) RVs $Z_{1}, \ldots, Z_{n}$:
\begin{equation}
\Ebb\left[ \max_{j \in [n]} Z_{j} \right] = \sum_{s=1}^{n}\left(-1\right)^{s+1} \sum_{A \in [n]_s} \Ebb \left[ \min_{j \in A} Z_j \right].
\end{equation}
\end{proposition}

\begin{proposition} 
\label{prop:boundsOnEZ}
Assume A2: State-independent receptions, heterogeneous receivers.  The $r^{\rm th}$ moment of the maximum of $n$ independent geometric RVs $\Ebb[X_{n:n}^{r}]$ has lower and upper bounds
\begin{equation}
\label{eq:boundsOnEZ1}
\psi_{r}(\qbf) \leq \Ebb[X_{n:n}^{r}] \leq \sum_{s=0}^{r} \binom{r}{s} \psi_{s}(\qbf),
\end{equation}
where 
\begin{equation}
\label{eq:cd}
\psi_{r}(\qbf) \equiv r! \sum_{s=1}^{n}(-1)^{s+1} \sum_{A \in [n]_s} \left(\sum_{j \in A} \phi(q_{j})\right)^{-r}.
\end{equation}
\end{proposition}
\begin{IEEEproof} 
It follows from the stochastic ordering presented in \S \ref{sec:modelnotation} that
\begin{equation}
\Ebb[\tilde{X}_{n:n}^r] \leq \Ebb[X_{n:n}^r] \leq \Ebb[(\tilde{X}_{n:n}+1)^r].
\end{equation}
First, applying the min-max identity (Prop.\ \ref{prop:minmax}) to $(\tilde{X}_1^r,\ldots,\tilde{X}_n^r)$ yields
\begin{equation}
\Ebb[\max_{j \in [n]} \tilde{X}_j^r] = \sum_{s=1}^n (-1)^{s+1} \sum_{A \in [n]_s} \Ebb[\min_{j \in A} \tilde{X}_j^r].
\end{equation}
Since $i)$ the minimum of independent exponential RVs is an exponential RV whose rate is the sum of the rates of these individual exponential RVs, and $ii)$ the $r^{\rm th}$ moment of an exponential RV with rate $\lambda$ (say) is $r!/\lambda^r$, it follows that $\Ebb\left[\min_{j \in A} \tilde{X}_j^r\right]  =  \Ebb\left[\left(\min_{j \in A} \tilde{X}_j\right)^r\right]   = r! \left( \sum_{j \in A} \phi(q_j) \right)^{-r}$.  This gives the lower bound $\psi_r(\qbf)$ in \eqref{eq:boundsOnEZ1}.  Next, the binomial theorem gives
\begin{equation}
\left( \max_{j \in [n]} \tilde{X}_j+1 \right)^r = \sum_{s=0}^r \binom{r}{s} \max_{j \in [n]} \tilde{X}_j^s,
\end{equation}
and taking expectations gives the upper bound in \eqref{eq:boundsOnEZ1}.
\end{IEEEproof}

\begin{remark}
Notice the gap is given by: $\sum_{s=0}^{r-1} \binom{r}{s} \psi_{s}(\qbf)$,
and when $r=1$ the gap equals $1$, since $\psi_{0}(\qbf) = \sum_{s=1}^{n}\left(-1\right)^{s+1} \binom{n}{s} = - \left[ \sum_{s=0}^{n}\left(-1\right)^{s} \binom{n}{s} - 1\right] = 1$.
\end{remark}

\begin{remark}
Lemma 4 of \cite{CogShr2011a} derives a lower bound on the expected time slots required to broadcast a packet to all $n$ homogeneous receivers (i.e., $\Ebb[X_{n:n}]$ when $q_j = q$ for $j \in [n]$), in order to construct an outer bound on the stability region of scheduling policies (Theorem 6).  Our result is more general in that it $i)$ provides lower and upper bounds, $ii)$ applies to heterogeneous channels $\qbf = (q_1,\ldots,q_n)$, and $iii)$ works for arbitrary ($r^{\rm th}$) moments.  Further, the proof technique we employ is probabilistic instead of analytic.  
\end{remark}

\begin{remark} 
\label{rem:1}
The classic (sequential) coupon-collector problem (\cite[\S3.6]{MotRag1995}) asks for the expected time to receive all of $n$ coupons, when in each time slot a new coupon is selected uniformly at random from $[n]$.  The broadcast delay problem can be considered as a ``parallel'' coupon-collector problem in that in each time slot each of the $n$ coupons is received independently with probabilities $\qbf$, i.e., a reception by receiver $j$ corresponds to collecting a coupon of type $j$.  
\end{remark}

Specializing Prop.\ \ref{prop:boundsOnEZ} to the homogeneous channel case with $r=1, q_j = q$ for all $j \in [n]$ yields simpler expressions for the lower and upper bounds after using the combinatorial identity $\sum_{s=1}^{n} (-1)^{s}\binom{n}{s} ( - \frac{1}{s} )  = \sum_{j=1}^{n} \frac{1}{j}$.  We observe the lower bound in Cor.\ \ref{cor:1} given below is the same as the one from Lemma 4 in \cite{CogShr2011a}.
\begin{corollary}
\label{cor:1}
Assume A3: State-independent receptions, homogeneous receivers.  The expectation of the maximum of $n$ iid geometric RVs with parameter $q$ is bounded as:
\begin{equation}
\frac{H_n}{\phi(q)}  ~ \leq \Ebb[X_{n:n}] ~ \leq \frac{H_n}{\phi(q)} + 1,
\end{equation}
for $H_n \equiv 1 + 1/2 + \cdots + 1/n$ the $n^{\rm th}$ harmonic number.  
\end{corollary}

\subsection{An exact moment expression for delay under UT}
\label{ssec:utexact}
The following proposition, whose proof can be found in Appendix \ref{app:delayUT}, derives an expression for the $r^{\rm th}$ moment of a geometric RV.
\begin{proposition} 
\label{prop:kthMomentGeo}
The $r^{\rm th}$ ($r \in \Nbb$) moment of a geometric RV $X \sim \mathrm{Geo}(q)$ is
\begin{equation}
\Ebb\left[ X^{r} \right] 
= \frac{1}{q} \sum_{l = 1}^{r} {r \brace l} \left( \frac{1}{q} -1 \right)^{l-1} l!, ~ q \in (0,1),
\end{equation}
where ${r \brace l}$ is the \textit{Stirling number of the second kind}.
\end{proposition}

It is interesting to observe Stirling numbers of the second kind (often defined as the number of partitions of $[r]$ into $l$ sets) often show up in expressions for the moments of discrete RVs. As another example, the $r^{\rm th}$ moment of a Poisson RV with parameter $\lambda$, say $Z \sim \mathrm{Poi}(\lambda)$, is given by $\Ebb[Z^r] = \sum_{l=1}^{r} { r \brace l } \lambda^{l}$.  Interpretations of the Stirling numbers of the second kind including their affinity to differential operators are discussed in \cite{website:StirlingNum2ndKind-planetmath}.

With Prop.\ \ref{prop:kthMomentGeo}, we can now present an exact moment expression by further leveraging the min-max identity (Prop.\ \ref{prop:minmax}) and the property that the minimum of independent geometric RVs is also a geometric RV.

\begin{proposition} 
\label{prop:kthMomentMaxGeo}
Assume A2: State-independent receptions, heterogeneous receivers.  The $r^{\rm th}$ moment of the maximum of $n$ independent geometric RVs $X_{j} \sim \mathrm{Geo}(q_{j}), j \in [n]$ is given by:
\begin{equation}
\Ebb[X_{n:n}^r] 
= \sum_{s=1}^{n}(-1)^{s+1} \sum_{A \in [n]_s} \frac{1}{q_{A}} \sum_{l = 1}^{r} {r \brace l} \left( \frac{1}{q_{A}} - 1 \right)^{l-1} l! ,
\end{equation}
where $q_{A} \equiv 1 - \prod_{j \in A} (1-q_j)$ and ${r \brace l}$ denotes the Stirling number of the second kind.
\end{proposition}
\begin{IEEEproof} 
Use min-max identity:
\begin{eqnarray}
\Ebb[X_{n:n}^r] & = & \Ebb [ \left(\max\left( X_1,\ldots,X_n\right)\right)^r ] \nonumber \\
& = & \Ebb [ \max\left( X_1^r,\ldots,X_n^r\right) ] \nonumber \\ 
& = & \sum_{s=1}^{n}\left(-1\right)^{s+1}\sum_{A \in [n]_s} \Ebb \left[ \min_{j \in A} X_j^r \right] \nonumber \\
& = & \sum_{s=1}^{n}\left(-1\right)^{s+1}\sum_{A \in [n]_s} \Ebb \left[ \left(\min_{j \in A} X_j\right)^r \right].
\end{eqnarray}
Now recognize $\min_{j \in A} X_j \sim \mathrm{Geo}(q_{A})$ (due to independence) and substitute the expression for the $r^{\rm th}$ moment of $\mathrm{Geo}(q)$ derived in Prop.\ \ref{prop:kthMomentGeo}.
\end{IEEEproof}
The following corollary illustrates the expressions from Prop.\ \ref{prop:kthMomentMaxGeo} for the first two moments ($r=1,2$).
\begin{corollary} 
\label{cor:firstTwoMomentsMaxGeo}
Assume A2: State-independent receptions, heterogeneous receivers.  The first two moments of maximum of $n$ independent geometric RVs with parameters $\qbf$ are:
\begin{IEEEeqnarray}{rCl}
\Ebb[X_{n:n}] & = & \sum_{s=1}^{n} \left(-1\right)^{s+1}\sum_{A \in [n]_s} \left(1-\prod_{j \in A}( 1 - q_j) \right)^{-1} \nonumber \\ 
\Ebb[X_{n:n}^2] & = & \sum_{s=1}^{n} \left(-1\right)^{s+1}\sum_{A \in [n]_s} \frac{1 + \prod_{j \in A} (1-q_j)}{\left(1 - \prod_{j \in A} (1-q_j) \right)^2}.\IEEEeqnarraynumspace
\end{IEEEeqnarray}
\end{corollary}

\begin{remark}
Prop.\ \ref{prop:kthMomentMaxGeo} builds upon the min-max identity (Prop.\ \ref{prop:minmax}) from \cite[pp.\ 128]{RosPek2007}.  It is straightforward to use the principle of inclusion and exclusion (PIE) upon which Prop.\ \ref{prop:minmax} is proved to establish a sequence of lower and upper bounds on $\Ebb[Z_{n:n}]$.  Moreover, there is an analogous ``max-min'' identity giving $\Ebb[\min_{j \in [n]} Z_j]$ in terms of $\Ebb[\max_{j \in A} Z_j]$ for $A \subseteq [n]$.  Also, recent work \cite{Zas2012} has generalized the min-max and max-min identities to a more general ``sorting'' identity in a non-probabilistic setting.

Characterizing the expectation of the maximum of $n$ geometric RVs is addressed in \cite{Eis2008}.  Let $(X_1,\ldots,X_n)$ be iid geometric RVs with parameter $q$.  By using Fourier analysis of the distribution of the fractional part of the maximum of corresponding iid exponential RVs, \cite[Corollary 2]{Eis2008} obtains an exact formula:
\begin{equation} 
\Ebb[X_{n:n}] = \frac{1}{2} +\frac{H_n}{\phi(q)} - \sum_{k \neq 0} \frac{1}{2\pi k \irm} \prod_{j=1}^{n} \left( 1 + \frac{2\pi k \irm }{j \phi(q)}\right)^{-1},
\label{eq:eisenberg2008maxgeo}
\end{equation}
where $\irm = \sqrt{-1}$. This result improves an earlier result \cite[Eq.\ $2.8$]{SzpReg1990}.  Using our Corollary \ref{cor:1}, the infinite sum in \eqref{eq:eisenberg2008maxgeo} can be shown to have absolute value no larger than $1/2$.
\end{remark}

\section{Delay under random linear combinations} 
\label{sec:delayRLNC}

We now turn our attention from delay under UT to delay under RLC, with a focus on the $r^{\rm th}$ ($r \in \Nbb$) moment, $\Ebb\left[ \left(Y_{n:n}^{(c)}\right)^r \right]$. We provide $i)$ lower and upper bounds for $\Ebb\left[ \left(Y_{n:n}^{(c)}\right)^r \right]$ in \S\ref{ssec:rlncbounds} (Props.\ \ref{prop:boundsOnEY1}, \ref{prop:boundsOnEY2}, and \ref{prop:boundsOnEY3}), $ii)$ exact expressions for $\Ebb\left[ \left(Y_{n:n}^{(c)}\right)^r \right]$ in \S\ref{ssec:exact} (Prop.\ \ref{prop:exacty}) which involves a sum over an infinite number of terms, $iii)$ a recurrence for $Y_{n:n}$ in \S\ref{ssec:recurrence} (Prop.\ \ref{pro:recexp}) which allows $\Ebb\left[ \left(Y_{n:n}^{(c)}\right)^r \right]$ to be computed in a finite number of steps, and $iv)$ in \S\ref{ssec:extremalchannels} (Props.\ \ref{prop:extremaldelayUT} and \ref{prop:extremaldelayRLNC}) a characterization of the channel reception probabilities $\qbf = (q_1,\ldots,q_n)$ for which $\Ebb\left[ \left(Y_{n:n}^{(c)}\right)^r \right]$ is minimized.

\subsection{Lower and upper bounds on the moments of delay under RLC}
\label{ssec:rlncbounds}

Our first result, Prop.\ \ref{prop:boundsOnEY1}, holds for the most general case (Assumption A1), and follows immediately from the stochastic ordering between generalized negative binomial (sum of independent geometric) and hypo-exponential (sum of independent exponential) RVs. The drawback of this result is that the upper bound is loose, i.e., $Y_{j}^{(c)} \leq_{\rm st} \tilde{Y}_{j}^{(c)}+c$.  To address this, we present in Prop.\ \ref{prop:boundsOnEY2} a continuous RV $\tilde{W}^{(c)}$ such that $\tilde{W}^{(c)} \leq_{\rm st} Y^{(c)} \leq_{\rm st} \tilde{W}^{(c)}+1$, with the caveat that the RV $\tilde{W}^{(c)}$ is constructed only under Assumption A2,  in which case the hypo-exponential and the generalized negative binomial RVs reduce to Gamma and negative binomial RVs respectively.  Finally, we present in Prop.\ \ref{prop:boundsOnEY3} explicit lower and upper bounds on the delay moments under Assumption A3.

\begin{proposition}
\label{prop:boundsOnEY1}
Assume A1: State-dependent receptions, heterogeneous receivers.  The $r^{\rm th}$ moment of the maximum of $n$ independent generalized negative binomial RVs $\Ebb\left[ \left(Y_{n:n}^{(c)}\right)^r \right]$ has lower and upper bounds
\begin{equation}
\label{eq:boundsOnRLNC1}
\Psi_{r}(\Qbf) \leq \Ebb\left[ \left(Y_{n:n}^{(c)}\right)^r \right] \leq \sum_{s=0}^{r} \binom{r}{s} \Psi_{s}(\Qbf) c^{r-s},
\end{equation}
where 
\begin{equation}
\label{eq:cdtytyyt}
\Psi_{r}(\Qbf) \equiv \Ebb\left[ \left(\tilde{Y}_{n:n}^{(c)}\right)^r \right] = \int_0^{\infty} \left(1 - \prod_{j \in [n]} F_{\tilde{Y}_j^{(c)}} \left(y^{\frac{1}{r}} \right)\right) \drm y,
\end{equation}
and $F_{\tilde{Y}_j^{(c)}}(y)$ is the CDF for the hypo-exponential RV $\tilde{Y}_j^{(c)} \sim \mathrm{HypoExp}(c,\phi(\qbf_j))$, with parameter vector $\phi(\qbf_j) = (\phi(q_{j,1}),\ldots,\phi(q_{j,c}))$.
\end{proposition}

\begin{IEEEproof} 
Recall the stochastic ordering discussed in \S\ref{sec:modelnotation}: $\left(\tilde{Y}_{n:n}^{(c)}\right)^r \leq_{\rm st} \left(Y_{n:n}^{(c)}\right)^r \leq_{\rm st} \left(\tilde{Y}_{n:n}^{(c)}+c\right)^r$. Taking expectations of these stochastically ordered RVs (and using the binomial theorem on the upper bound) yields \eqref{eq:boundsOnRLNC1}.  To obtain \eqref{eq:cdtytyyt} observe
\begin{eqnarray}
\Ebb\left[ \left(\tilde{Y}_{n:n}^{(c)}\right)^r \right]
\!\!&=&\!\! \int_0^{\infty} \Pbb\left(\left(\tilde{Y}_{n:n}^{(c)}\right)^r > y\right) \drm y \nonumber \\
\!\!&=&\!\! \int_0^{\infty} \left(1 - \Pbb\left(\left(\tilde{Y}_{n:n}^{(c)}\right)^r \leq y\right)\right) \drm y \nonumber \\
\!\!&=&\!\! \int_0^{\infty} \left(1 - \prod_{j \in [n]} \Pbb\left(\left(\tilde{Y}_j^{(c)}\right)^r \leq y\right) \right) \drm y \nonumber \\
\!\!&=&\!\! \int_0^{\infty} \left(1 - \prod_{j \in [n]} \Pbb \left(\tilde{Y}_j^{(c)} \leq y^{\frac{1}{r}} \right) \right) \drm y.
\end{eqnarray}
\end{IEEEproof}

\begin{proposition}
\label{prop:boundsOnEY2}
Assume A2: State-independent receptions, heterogeneous receivers.  The $r^{\rm th}$ moment of the maximum of $n$ independent generalized negative binomial RVs $\Ebb\left[ \left(Y_{n:n}^{(c)}\right)^r \right]$ has lower and upper bounds
\begin{equation}
\label{eq:boundsOnRLNC2}
\tilde{\Psi}_{r}(\qbf) \leq \Ebb\left[ \left(Y_{n:n}^{(c)}\right)^r \right] \leq \sum_{s=0}^{r} \binom{r}{s} \tilde{\Psi}_{s}(\qbf),
\end{equation}
where 
\begin{IEEEeqnarray}{rCl}
\label{eq:cd2}
\tilde{\Psi}_{r}(\qbf) & \equiv & \Ebb\left[ \left(\tilde{W}_{n:n}^{(c)}\right)^r \right] \nonumber \\
& = & c^{r} +\int_{c^{r}}^{\infty} \left( 1 - \prod_{j \in \left[ n \right]} I_{q_{j}}(c,w^{\frac{1}{r}}-c+1) \right) \drm w.\IEEEeqnarraynumspace
\end{IEEEeqnarray}
\end{proposition}
\begin{IEEEproof}
For integer blocklength $c$ and success probability $q \in (0,1)$, define the continuous RV $\tilde{W}^{(c)}$ with support $[c,\infty)$ and CDF $F_{\tilde{W}^{(c)}}(w) = I_q(c,w-c+1) \mathbf{1}_{\left\{ w \geq c \right\}}$.  We first show that $\tilde{W}^{(c)}$ is a valid continuous RV.  Clearly $F_{\tilde{W}^{(c)}}(w) \in [0,1]$, and by construction $\lim_{w \to -\infty} F_{\tilde{W}^{(c)}}(w) = 0$. By Lem.\ \ref{lem:regbetamono} in Appendix \ref{app:delayRLNC} we have: $\frac{\drm}{\drm w} F_{\tilde{W}^{(c)}}(w) > 0$, and it also follows that  $\lim_{w \to \infty} F_{\tilde{W}^{(c)}}(w) = 1$. 

We now show that $\tilde{W}^{(c)} \leq_{\rm st} Y^{(c)} \leq_{\rm st} \tilde{W}^{(c)}+1$ for $Y^{(c)} \sim \mathrm{NegBin}(c,q)$.  Observe that the CDF for $\tilde{W}^{(c)}$ and $Y^{(c)}$ are equal at every integer $m \geq c$:
\begin{eqnarray}
F_{\tilde{W}^{(c)}}(m) & = & I_q(c,m-c+1) = \Pbb(\mathrm{Bin}(m,q) \geq c) \nonumber \\
& = & \Pbb(Y^{(c)} \leq m) = F_{Y^{(c)}}(m).
\end{eqnarray}
This, along with the facts that $\frac{\drm}{\drm w} F_{\tilde{W}^{(c)}}(w) > 0$ and $F_{Y^{(c)}}(w)$ is piecewise constant in $w$, guarantees $F_{\tilde{W}^{(c)}}(w) \geq F_{Y^{(c)}}(w)$, and thus $\tilde{W}^{(c)} \leq_{\rm st} Y^{(c)}$.  Similarly, 
\begin{IEEEeqnarray}{rCl}
\IEEEeqnarraymulticol{3}{l}
{F_{\tilde{W}^{(c)}+1}(m)
}\nonumber \\* \quad
& = & I_q(c,(m-1)-c+1) = \Pbb(\mathrm{Bin}(m-1,q) \geq c) \nonumber \\
& = & \Pbb(Y^{(c)} \leq m-1) = F_{Y^{(c)}}(m-1),
\end{IEEEeqnarray}
and thus $F_{\tilde{W}^{(c)}+1}(m) = F_{Y^{(c)}}(m-1) \leq F_{Y^{(c)}}(m)$, which establishes $Y^{(c)} \leq_{\rm st} \tilde{W}^{(c)}+1$.  It follows that $\left( \tilde{W}_{n:n}^{(c)} \right)^r \leq_{\rm st} \left(Y_{n:n}^{(c)}\right)^r \leq_{\rm st} \left(\tilde{W}_{n:n}^{(c)}+1\right)^r$.  Taking expectations of these stochastically ordered RVs (and using the binomial theorem on the upper bound) yields \eqref{eq:boundsOnRLNC2}.  To obtain \eqref{eq:cd2} observe
\begin{eqnarray}
\Ebb\left[\left( \tilde{W}_{n:n}^{(c)} \right)^r\right] \!\!& = &\!\! \int_0^{\infty} \Pbb\left(\left( \tilde{W}_{n:n}^{(c)} \right)^r  > w\right) \drm w \nonumber \\
\!\!& = &\!\! c^{r} + \int_{c^{r}}^{\infty} \Pbb\left(\left( \tilde{W}_{n:n}^{(c)} \right)^r > w\right) \drm w,
\end{eqnarray}
and apply the same steps used at the end of the proof of Prop.\ \ref{prop:boundsOnEY1}.
\end{IEEEproof}

Observe the bounds given in the previous two propositions may not be easy to compute since they themselves involve calculation of moments of the maximum order statistic of a continuous RV, the support of which may be infinite. The following proposition shows that by further bounding the continuous RV we obtain bounds that also have simple structure. More precisely, we will be working with the Gamma distribution, for which the bounds will be used in investigating the dependence on $c$ when $n$ is fixed (\S \ref{ssec:rossdlctightinc}), and on $n$ when $c$ is fixed (\S \ref{sec:deponnumrx}). For simplicity, the following proposition is stated under Assumption A3, yet generalizing to the heterogeneous receivers case is not hard.

\begin{proposition}
\label{prop:boundsOnEY3}
Assume A3: State-independent receptions, homogeneous receivers.  The $r^{\rm th}$ moment of the maximum of $n$ iid negative binomial RVs $\Ebb\left[\left(Y_{n:n}^{(c)}\right)^{r}\right]$ has lower and upper bounds
\begin{eqnarray}
\label{eq:boundsOnRLNC3}
\Ebb\left[\left(Y_{n:n}^{(c)}\right)^{r}\right] & \geq & \frac{1}{\phi(q)^{r}} l_{\tilde{Y}}(t,n,c,r) \nonumber \\
\Ebb\left[\left(Y_{n:n}^{(c)}\right)^{r}\right]  & \leq  & \sum_{p=0}^{r} \binom{r}{p} \frac{1}{\phi(q)^{p}} u_{\tilde{Y}}(s,n,c,p)  c^{r-p},
\end{eqnarray}
where 
\begin{IEEEeqnarray}{rCl}
\label{eq:delaCalandRossSecFourA}
& & \!\!\!\!\!\!\! l_{\tilde{Y}}(t,n,c,r)  \equiv t^{r} - \left( t^{r} - \Ebb \left[ \left( \tilde{Y}^{(c)} \right)^{r} \right] \right) \left(1 - Q(c, t) \right)^{n-1} \nonumber \\
\IEEEeqnarraymulticol{3}{l}
{u_{\tilde{Y}}(s,n,c,r) 
}\nonumber \\* 
& \equiv &  s + n \left( \sizecorr{\frac{\Gamma(c+r)}{\Gamma(c)} s^{\frac{c+r-1}{r}} \erm^{-s^{\frac{1}{r}}} \sum_{m=0}^{r-1} \frac{1}{s^{\frac{m}{r}}\Gamma(c+r-m)}}   \left( \frac{\Gamma(c+r)}{\Gamma(c)} - s \right) Q(c,s^{\frac{1}{r}}) \right. \nonumber \\
& & \negmedspace{} + \left. \frac{\Gamma(c+r)}{\Gamma(c)} s^{\frac{c+r-1}{r}} \erm^{-s^{\frac{1}{r}}} \sum_{m=0}^{r-1} \frac{1}{s^{\frac{m}{r}}\Gamma(c+r-m)} \right).
\end{IEEEeqnarray}
For the lower bound, $t \geq \left(\Ebb\left[ \left( \tilde{Y}^{(c)}\right)^{r}\right]\right)^{\frac{1}{r}}$ and $\tilde{Y}^{(c)} \sim \mathrm{Gamma}(c,1)$; it is made the tightest by solving for the unique stationary point $t^{*}$ such that $\frac{\partial}{\partial t} l_{\tilde{Y}}(t,n,c,r) = 0$. For the upper bound, $s \geq 0$, and the optimal $s^{*}$ is such that $n Q\left(c, \left(s^*\right)^{\frac{1}{r}}\right) = 1$ for $Q(c,s) \equiv \frac{\Gamma(c,s)}{\Gamma(c)} = \Pbb\left( \tilde{Y}^{(c)} > s \right)$ the CCDF of the $\mathrm{Gamma}(c,1)$ distribution.  For $r=1$, the optimal upper bound is
\begin{equation} 
\label{eq:rossuboptfornegbin}
\Ebb\left[Y_{n:n}^{(c)}\right] \leq 
c\left(1 + \frac{1}{\phi(q)} n Q\left(c+1,Q^{-1}\left(c,\frac{1}{n}\right)\right) \right).
\end{equation}
\end{proposition}
The proof can be found in Appendix \ref{app:delayRLNC}.

Observe we have used two distinct continuous-RV stochastic orderings in the last two propositions. Prop.\ \ref{prop:boundsOnEY2} relies upon $\tilde{W}_j^{(c)} \leq_{\rm st} Y_j^{(c)} \leq_{\rm st} \tilde{W}_j^{(c)} + 1$ under Assumption A2, while Prop.\ \ref{prop:boundsOnEY3} relies upon $\tilde{Y}_j^{(c)} \leq_{\rm st} Y_j^{(c)} \leq_{\rm st} \tilde{Y}_j^{(c)} + c$ under Assumption A3.  These orderings are illustrated in Fig.\ \ref{fig:stochorder}.  Note the stochastic ordering with $\tilde{W}_j^{(c)}$ is much tighter to $Y_j^{(c)}$ than is that with $\tilde{Y}_j^{(c)}$, although the latter is much easier to compute than the former.  

\begin{figure}[ht!]
\centering
\includegraphics[width=0.49\textwidth]{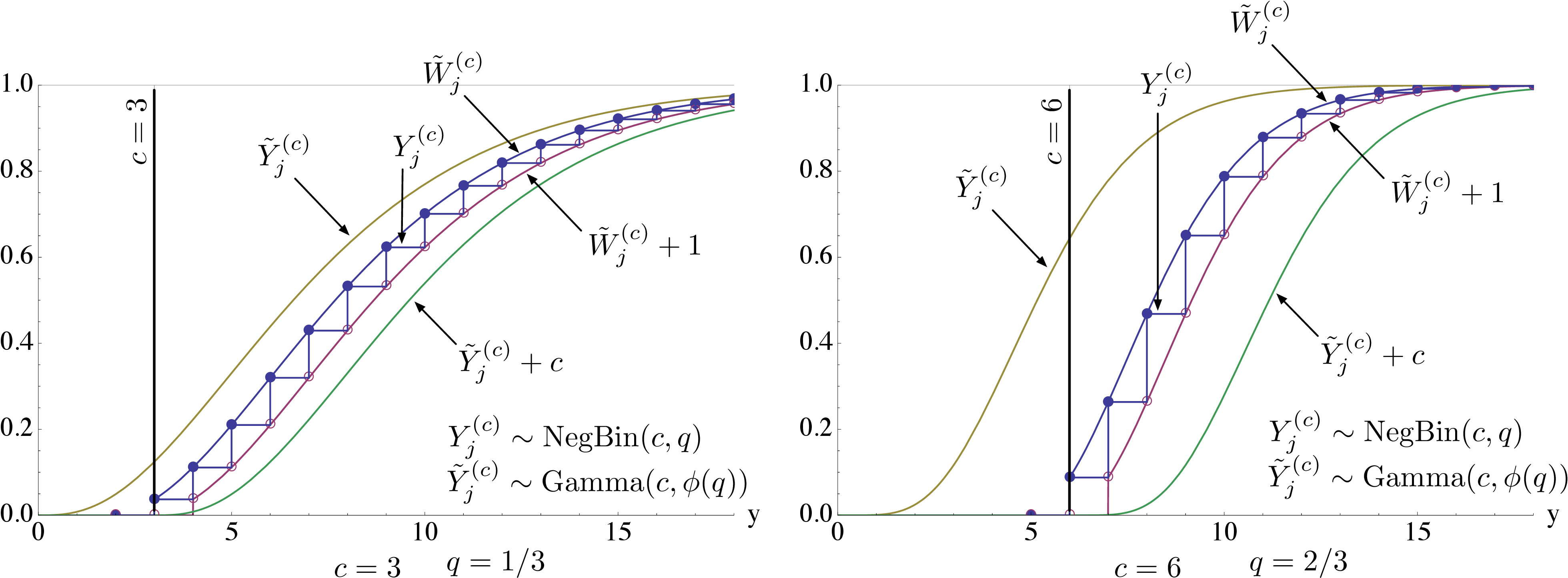}
\caption{Illustration of the CDFs for RVs in the stochastic orderings $\tilde{W}_j^{(c)} \leq_{\rm st} Y_j^{(c)} \leq_{\rm st} \tilde{W}_j^{(c)} + 1$ and $\tilde{Y}_j^{(c)} \leq_{\rm st} Y_j^{(c)} \leq_{\rm st} \tilde{Y}_j^{(c)} + c$ for $c = 3$ and $q = 1/3$ (left), and $c=6$ and $q=2/3$ (right).}
\label{fig:stochorder}
\end{figure}

\subsection{Exact expressions of the moments of delay under RLC}
\label{ssec:exact}

Prop.\ \ref{prop:exacty} below, proved in Appendix \ref{app:delayRLNC}, gives an expression for the exact delay $\Ebb\left[ \left(Y_{n:n}^{(c)}\right)^r \right]$ under RLC for the state-dependent case as well as two equivalent expressions for $\Ebb\left[ \left(Y_{n:n}^{(c)}\right)^r \right]$ under RLC for the state-independent case.  
\begin{proposition}
\label{prop:exacty}
Assume A1: State-dependent receptions, heterogeneous receivers:
\begin{IEEEeqnarray}{rCl}
  \IEEEeqnarraymulticol{3}{l}{\Ebb\left[ \left(Y_{n:n}^{(c)}\right)^r \right] = 
  }\nonumber\\* \quad
  c^{r} + \sum_{m = c^{r}}^{\infty} \left(  1 - \prod_{j=1}^{n} \sum_{t=c}^{\lfloor m^{\frac{1}{r}} \rfloor} \sum_{\boldsymbol\alpha \in \Amc_t} \prod_{k=1}^{c}(1-q_{j,k})^{\alpha_{k} - 1}q_{j,k} \right),
  \nonumber\\*
\label{eq:tu}
\end{IEEEeqnarray}
where $\Amc_t$ is the finite set of all $c$-vectors of positive integers that sum to $t$, i.e., 
\begin{equation}
\Amc_t = \left\{ \boldsymbol\alpha = (\alpha_1,\ldots,\alpha_c) ~:~ \alpha_k \in \Nbb, ~ \sum_{k=1}^c \alpha_k = t \right\}.
\end{equation}
Assume A2: State-independent receptions, heterogeneous receivers.  Eq.\ \eqref{eq:tu} simplifies to 
\begin{IEEEeqnarray}{rCl}
  \IEEEeqnarraymulticol{3}{l}{\Ebb\left[ \left(Y_{n:n}^{(c)}\right)^r \right] = 
  }\nonumber\\* \quad
c^{r}+ \sum_{m = c^{r}}^{\infty} \left(  1 - \prod_{j=1}^{n} \sum_{t=c}^{\lfloor m^{\frac{1}{r}} \rfloor} \binom{t-1}{c-1}(1-q_{j})^{t-c}q_{j}^{c} \right).
  \nonumber\\*
\label{eq:tv}
\end{IEEEeqnarray}
An alternative expression for the state-independent case is
\begin{equation}
\label{eq:tw}
\Ebb\left[ \left(Y_{n:n}^{(c)}\right)^r \right] = c^{r} + \sum_{m=c^{r}}^{\infty} \left( 1 - \prod_{j=1}^{n} I_{q_{j}}\left( c, \lfloor m^{\frac{1}{r}} \rfloor-c+1 \right) \right),
\end{equation}
where in the last equation, $I_{x}(\alpha, \beta)$ is the regularized incomplete beta function, which can be used as both the CDF of the beta distribution and the tail probability of the binomial RV $W \sim \mathrm{Bin}(m,q)$:
\begin{equation}
\label{eq:tx}
\Pbb(W \geq l) = \sum_{r=l}^{m} \binom{m}{r}q^{r}(1-q)^{m-r} = I_{q}(l, m-l+1).
\end{equation}
\end{proposition}

\subsection{Recurrence for the (moments of) delay under RLC}
\label{ssec:recurrence}
A limitation of Prop.\ \ref{prop:exacty} is that the expression involves an infinite summation.  In this subsection, we offer a recurrence equation for the RV $Y_{n:n}$ that permits calculation  of the exact value of $\Ebb\left[ \left(Y_{n:n}^{(c)}\right)^r \right]$ using only a finite number of steps.  In fact it is convenient to generalize our setting slightly, in that we now suppose that receiver $j$ requires reception of $c_j$ packets; prior to this we have assumed $c_j = c$ for all $j \in [n]$.  Let $\cbf^0 = (c_1^0,\ldots,c_n^0)$ be the $n$-vector of required number of successes for each receiver.  The recurrence will be in terms of the generic vector $\cbf \leq \cbf^0$ (component-wise), interpreted as the number of successes left to go for each receiver,  as explained below. We shall also in this subsection write $Y_{n:n}^{(\cbf)}$ to emphasize this.

We introduce some shorthand notation.  First, define $[n(\cbf)] \equiv \{ j \in [n] : c_j > 0\}$ as the set of active receivers, i.e., those still requiring an additional reception to complete.  Second, define the probabilities of success and failure by the active receiver subsets $S$ and $[n(\cbf)]\setminus S$ respectively, to be: 
\begin{eqnarray}
q(S,\cbf) & \equiv  & \prod_{j \in S} q_{j,c_j^0-c_j+1}, \nonumber \\
\bar{q}([n(\cbf)]\setminus S,\cbf)  & \equiv  & \prod_{j \in [n(\cbf)]\setminus S} (1-q_{j,c_j^0-c_j+1}).  
\end{eqnarray}
In words, $q(S,\cbf)$ is the probability of success for active receivers with indices in $S$ where the ``successes to go'' $c_j$ and initial number of successes required $c_j^0$ determine the state index $k = c_j^0 - c_j + 1$ for receiver $j$.  Further, $\bar{q}([n(\cbf)]\setminus S,\cbf)$ is the probability of failure for the active receivers not indexed by $S$, where again the state index for each such receiver $j$ is $k = c_j^0 - c_j + 1$.  

Third, define $\mathbf{1}_S$ to be the $n$-vector with ones in the positions indexed by $S$ and zero elsewhere.  We reiterate that in the state-dependent case with success probability matrix $\Qbf$, we have $Y_{n:n}^{(\cbf)} = \max(Y_1^{(c_{1})},\ldots,Y_n^{(c_{n})})$ where $(Y_1^{(c_{1})},\ldots,Y_n^{(c_{n})})$ are independent with $Y_j^{(c_{j})} = \sum_{k=1}^{c_j} X_{j,k}$ and the $X_{j,k} \sim \mathrm{Geo}(q_{j,k})$ are themselves independent in $j$ and $k$.  In the state-independent case we have $X_{j,k} \sim \mathrm{Geo}(q_j)$ for $k \in [c_j]$.  

\begin{proposition}
\label{pro:recexp}
Assume A1: State-dependent receptions, heterogeneous receivers.  The RV $Y_{n:n}^{(\cbf^{0})}$ defined above admits the recurrence 
\begin{equation} 
\label{eq:rutt}
Y_{n:n}^{(\cbf)} = 1 + \sum_{S \subseteq [n(\cbf)]} q(S,\cbf) \bar{q}([n(\cbf)]\setminus S,\cbf) Y_{n:n}^{(\cbf- \mathbf{1}_S)},
\end{equation}
with boundary condition $Y_{n:n}^{(\mathbf{0})} = 0$.  Assume A2: State-independent receptions, heterogeneous receivers.  The RV $Y_{n:n}^{(\cbf^0)}$ defined above admits the recurrence
\begin{equation}
\label{eq:rvtt}
Y_{n:n}^{(\cbf)} = 1 + \sum_{S \subseteq [n(\cbf)]} \prod_{j \in S} q_j  \prod_{j \in [n(\cbf)] \setminus S} (1-q_j)  Y_{n:n}^{(\cbf - \mathbf{1}_S)},
\end{equation}
again with boundary condition $Y_{n:n}^{(\mathbf{0})} = 0$.  
\end{proposition}

\begin{IEEEproof}
The recurrence is obtained by conditioning on the outcomes of the current trial with $\cbf^0$ indicating the initial target number of successes required by each receiver. The set of outcomes for the trial corresponds to the set of all subsets $S$ of the active receivers $[n(\cbf)]$, i.e., all possible subsets of potential successes.  The probability of success by active receivers in $S$ and failure by active receivers not in $S$ is $q(S,\cbf) \bar{q}([n(\cbf)] \setminus S,\cbf)$.  The effect of successes by receivers in $S$ is to reduce the number of required successes by those receivers by one, i.e., $\cbf \to \cbf - \mathbf{1}_S$, thus advancing the active receivers in $S$ to the next ``state''.  That is, their success probability advances from $q_{j,k_j}$ where $k_j = c_j^0 - c_j + 1$ to $q_{j,k_j + 1}$.  This gives \eqref{eq:rutt}.  To obtain \eqref{eq:rvtt}, observe that in the state-independent case, the probability of success by $S$ and failure by $[n(\cbf)]\setminus S$ is $\prod_{j \in S} q_j$ times $\prod_{j \in [n(\cbf)] \setminus S} (1-q_j)$.
\end{IEEEproof}

\begin{remark}
The empty set is included in the set of subsets of $[n(\cbf)]$ in \eqref{eq:rutt} and \eqref{eq:rvtt}, which must be ``subtracted out'' to express the recurrence for $Y_{n:n}^{(\cbf)}$ in terms of strictly smaller vectors $\cbf' < \cbf$ (in at least one component).  Doing this gives, for \eqref{eq:rutt}, 
\begin{equation}
Y_{n:n}^{(\cbf)} = \frac{1 + \sum_{S \subseteq [n(\cbf)] \setminus \emptyset} q(S,\cbf) \bar{q}([n(\cbf)]\setminus S,\cbf) Y_{n:n}^{(\cbf - \mathbf{1}_S)}}{1 - q(\emptyset,\cbf) \bar{q}([n(\cbf)],\cbf)},
\end{equation}
where $[n(\cbf)] \setminus \emptyset$ denotes all non-empty subsets of the set of active receivers.
\end{remark}

Taking the expectations of \eqref{eq:rutt} and \eqref{eq:rvtt} yields recurrences on $\Ebb[Y_{n:n}^{(c)}]$ for the state dependent and independent cases.
\begin{corollary}
\label{cor:exprec}Assume A1: State-dependent receptions, heterogeneous receivers.  The expectation $\Ebb[Y_{n:n}^{(\cbf^0)}]$ admits the recurrence
\begin{equation} 
\label{eq:rutt2}
\Ebb\left[Y_{n:n}^{(\cbf)}\right] = 1 + \sum_{S \subseteq [n(\cbf)]} q(S,\cbf) \bar{q}([n(\cbf)]\setminus S,\cbf) \Ebb\left[Y_{n:n}^{(\cbf - \mathbf{1}_S)}\right],
\end{equation}
with boundary condition $\Ebb[Y_{n:n}^{(\mathbf{0})}] = 0$.  Assume A2: State-independent receptions, heterogeneous receivers.  The expectation $\Ebb[Y_{n:n}^{(\cbf^0)}]$ admits the recurrence
\begin{equation}
\label{eq:rvtt2}
\Ebb\left[Y_{n:n}^{(\cbf)}\right] = 1 + \sum_{S \subseteq [n(\cbf)]} \prod_{j \in S} q_j   \prod_{j \in [n(\cbf)] \setminus S} (1-q_j)  \Ebb\left[Y_{n:n}^{(\cbf - \mathbf{1}_S)}\right],
\end{equation}
again with boundary condition $\Ebb[Y_{n:n}^{(\mathbf{0})}] = 0$.  
\end{corollary}

It is straightforward to see that for $\cbf = c_j \ebf_j$ (the $n$-vector of all zeros with value $c_j$ in position $j$), in the state-independent case we have the boundary condition $\Ebb[Y_{n:n}^{(c_j \ebf_j)}] = c_j/q_j$.  That is, when there is only one receiver left to receive, the expected duration is the expectation of a negative binomial RV, which of course is $c_j$ successes required times $1/q_j$ time slots on average per success. 

\begin{remark}
A parallel recurrence holds for the RV $Y_{1:n} = Y_{1:n}^{(\cbf^0)}$ where $Y_{1:n} ^{(\cbf^{0})} \equiv \min(Y_1^{(c_{1}^{0})},\ldots,Y_n^{(c_{n}^{0})})$ is the random time until the {\em first} receiver, say $j$, receives its target number of successes $c_j^0$ (as opposed to counting until all receivers reach their targets).  For the state-dependent case the RV $Y_{1:n} = Y_{1:n}^{(\cbf^0)}$ defined above admits the recurrence 
\begin{equation} 
\label{eq:rutt3}
Y_{1:n}^{(\cbf)} = 1 + \sum_{S \subseteq [n]} q(S,\cbf) \bar{q}([n]\setminus S,\cbf) Y_{1:n}^{(\cbf - \mathbf{1}_S)},
\end{equation}
with boundary condition $Y_{1:n}^{(\cbf)} = 0$ for any $\cbf$ containing one or more zeros.  For the state-independent case the RV $Y_{1:n} = Y_{1:n}^{(\cbf^0)}$ defined above admits the recurrence
\begin{equation}
\label{eq:rvtt3}
Y_{1:n}^{(\cbf)} = 1 + \sum_{S \subseteq [n]} \prod_{j \in S} q_j  \prod_{j \in [n] \setminus S} (1-q_j)  Y_{1:n}^{(\cbf - \mathbf{1}_S)},
\end{equation}
again with boundary condition $Y_{1:n}^{(\cbf)} = 0$ for any $\cbf$ containing one or more zeros.  
\end{remark}

\begin{remark}
In fact we can use the \textit{first-step analysis} technique in probability to establish recurrences to compute arbitrary moments of $Y_{n:n}$, i.e., $\Ebb\left[\left(Y_{n:n}^{(\cbf^0)}\right)^r\right]$. To show this, first we observe
\begin{IEEEeqnarray}{rCl}
\IEEEeqnarraymulticol{3}{l}
{\Pbb(Y_{n:n}^{(\cbf)} = y ~|~ S \in [n(\cbf)] \mbox{ succeed in c.t.s.})
}\nonumber \\* \quad
&  = & \Pbb(1 + Y_{n:n}^{(\cbf - \mathbf{1}_S)} = y)
\end{IEEEeqnarray}
where ``c.t.s.'' stands for {\em current time slot}.  That is, the event $\{Y_{n:n}^{(\cbf)} = y\}$ conditioned on the event of successes for active receivers $S$ in the current time slot, is the same as the (unconditional) event $\{ 1 + Y_{n:n}^{(\cbf-\mathbf{1}_S)} = y \}$. In other words, conditioning on the number of successes in the current time slot affects the state (the number of successes to go) but, aside from that, the system probabilistically restarts itself. Next, we can equate arbitrary powers $r$ of both sides, i.e., 
\begin{IEEEeqnarray}{rCl}
\IEEEeqnarraymulticol{3}{l}
{\Pbb\left(\left(Y_{n:n}^{(\cbf)}\right)^r= y^r ~|~ S \in [n(\cbf)] \mbox{ succeed in c.t.s}\right)
}\nonumber \\* \quad
&  = & \Pbb\left(\left(1 + Y_{n:n}^{(\cbf - \mathbf{1}_S)}\right)^r = y^r \right). 
\end{IEEEeqnarray}
On this basis, our recurrence can be generalized to one for higher moments of $Y_{n:n}$, e.g., for the state-dependent case:
\begin{IEEEeqnarray}{rCl}
\IEEEeqnarraymulticol{3}{l}
{\Ebb\left[\left(Y_{n:n}^{(\cbf)}\right)^r\right]
}\nonumber \\* \quad
&  = & \sum_{S \subseteq [n(\cbf)]} q(S,\cbf) \bar{q}([n(\cbf)]\setminus S,\cbf) \Ebb\left[\left(1+Y_{n:n}^{(\cbf - \mathbf{1}_S)}\right)^r\right], \IEEEeqnarraynumspace
\label{eq:rutt4}
\end{IEEEeqnarray}
with boundary condition $\Ebb\left[\left(Y_{n:n}^{(\mathbf{0})}\right)^r\right] = 0$.
\end{remark}

\begin{remark}
Recurrences for the maximum of geometric RVs and in fact the maximum of sums of geometric RVs are found in the literature.  We mention in particular \cite[Eq.\ (2.5)]{SzpReg1990} which provided the inspiration for our recurrence, but in the simpler context of the expected maximum of iid geometric RVs.   Further, recent work \cite{DikDim2010} on the expected delay of network coded packets routed simultaneously over two paths employs a similar recurrence to ours, but with significant differences.  In particular, since in their context there is a single receiver looking for $c$ innovative packets over $n=2$ disjoint routes, their recurrence is univariate in the scalar $c$.  The general problem of expressing the maximum of negative binomial RVs is addressed in \cite{GraPro1997}.
\end{remark}

\subsection{Erasure channels that minimize delay}
\label{ssec:extremalchannels}

In this subsection we are interested in characterizing the channels for which $\Ebb\left[ \left(Y_{n:n}^{(c)}\right)^r \right]$ is minimized, subject to an equality constraint on the average over the $n$ channel reception probabilities.  Throughout this subsection we restrict our attention to the case $q_{j,k} = q_j$ for all $j \in [n]$, or equivalently, we assume the field size $d$ is infinite in \eqref{eq:tyt}.  Denote the corresponding channel reception probabilities as $\qbf = (q_j, j \in [n])$, the average as $\bar{\qbf} = \left(1/n\right)\sum_{j=1}^{n} q_{j}$, and the set of possible channel vectors with average $q$ as $\Qmc_q = \{ \qbf \in (0,1)^n : \bar{\qbf} = q\}$, for $q \in (0,1)$.  Prop.\ \ref{prop:extremaldelayUT} (\ref{prop:extremaldelayRLNC}) establishes that the delay minimizing channel vector is $\qbf = (q,\ldots,q)$ when minimized over $\Qmc_q$, for UT (RLC), respectively.  As RLC includes UT as a special case under Assumption A2 (state-independent receptions, heterogeneous receivers), Prop.\ \ref{prop:extremaldelayRLNC} implies Prop.\ \ref{prop:extremaldelayUT}.  We include both because the proof of Prop.\ \ref{prop:extremaldelayUT} is simpler than that of Prop.\ \ref{prop:extremaldelayRLNC} while still retaining the  key ideas.  
\begin{proposition} 
\label{prop:extremaldelayUT}
For any $q \in (0,1)$, the $r^{\rm th}$ ($r \in \Nbb$) moment of delay under UT, $\Ebb[X_{n:n}^r]$, is minimized over $\qbf \in \Qmc_q$ for $\qbf = (q,\ldots,q)$.
\end{proposition}
\begin{IEEEproof}
Write $q \mathbf{1}$ to denote $(q,\ldots,q)$.  We need to show $\Ebb[X_{n:n}^{r}(\qbf)] \geq \Ebb[X_{n:n}^{r}(q \mathbf{1})]$ for all $\qbf \in \Qmc_q$. Toward this, we write
\begin{IEEEeqnarray}{rCl}
\Ebb[X_{n:n}^r]  & = & \sum_{x=0}^{\infty} \Pbb(X_{n:n}^r > x) \nonumber \\
&=& \sum_{x=0}^{\infty} \Pbb(X_{n:n}> \lfloor x^{\frac{1}{r}} \rfloor) \nonumber \\
&=& \sum_{x=0}^{\infty} \left( 1 - \prod_{j \in [n]} \Pbb(X_{j} \leq \lfloor x^{\frac{1}{r}} \rfloor) \right) \nonumber \\
& = & 1 + \sum_{x=1}^{\infty} \left( 1 - \prod_{j \in [n]} \left( 1 - \left( 1 - q_{j} \right) ^{\lfloor x^{\frac{1}{r}} \rfloor} \right) \right).
\end{IEEEeqnarray}
It suffices to show the desired inequality holds termwise with respect to the summing variable $x \in \Nbb$:
\begin{equation} 
1 - \prod_{j \in [n]} \left( 1 - \left( 1 - q_{j} \right) ^{\lfloor x^{\frac{1}{r}} \rfloor} \right) \geq 1 -  \left( 1 - \left( 1 - q \right) ^{\lfloor x^{\frac{1}{r}} \rfloor} \right)^{n}.
\end{equation}
This is equivalent to showing
\begin{equation} 
\label{eq:termwiseInequalityUnderUT}
\frac{1}{n} \sum_{j \in [n]} \log \left( 1 - (1-q_{j})^{\lfloor x^{\frac{1}{r}} \rfloor} \right) \leq \log \left(1 - (1-q)^{\lfloor x^{\frac{1}{r}} \rfloor} \right).
\end{equation}
Given $x$, define $f(z) \equiv \log(1 - (1-z)^{\lfloor x^{\frac{1}{r}} \rfloor})$. Further, define a discrete uniform random variable $Z$ with support $\{q_{1}, \ldots, q_{n} \}$.  By assumption $\Ebb[Z] = \bar{\qbf} = q$.  Consequently \eqref{eq:termwiseInequalityUnderUT} directly follows from Jensen's inequality and the concavity of $f(z)$ on $z \in (0,1)$, for each $x \in \Nbb$.
\end{IEEEproof}

\begin{proposition} 
\label{prop:extremaldelayRLNC}
Assume A2: State-independent receptions, heterogeneous receivers. For any $c \in \Nbb$ and $q \in (0,1)$, the $r^{\rm th}$ ($r \in \Nbb$) moment of delay under RLC, $\Ebb\left[ \left(Y_{n:n}^{(c)}\right)^r \right]$, is minimized over $\qbf \in \Qmc_q$ for $\qbf = (q,\ldots,q)$.
\end{proposition}
\begin{IEEEproof}
The proof mirrors that of Prop.\ \ref{prop:extremaldelayUT}.  Write $q \mathbf{1}$ to denote $(q,\ldots,q)$.  We need to show $\Ebb\left[\left(Y_{n:n}^{(c)} (\qbf) \right)^{r}\right] \geq \Ebb\left[\left(Y_{n:n}^{(c)} (q \mathbf{1}) \right)^{r}\right]$ for all $\qbf \in \Qmc_q$.  An expression for $\Ebb\left[ \left(Y_{n:n}^{(c)}\right)^r \right]$ in terms of regularized incomplete beta function is given in \S\ref{ssec:exact} Eq.\ \eqref{eq:tw}.  To establish this inequality, it suffices to show this inequality holds for each term in the sum.  That is, for each integer $m \geq c^{r}$, we must show
\begin{equation}
1 - \prod_{j=1}^{n} I_{q_{j}}\left( c,  \lfloor m^{\frac{1}{r}} \rfloor -c+1 \right)  \! \geq \!  1 - \left( I_{q}\left( c,  \lfloor m^{\frac{1}{r}} \rfloor -c+1 \right) \right)^{n}.
\end{equation}
Since the function $I_{x}(a,b)$ is always nonnegative, the above is equivalent to
\begin{equation} 
\label{eq:termwiseInequalityUnderRLNC}
\frac{1}{n} \sum_{j=1}^{n} \log I_{q_{j}} (c,  \lfloor m^{\frac{1}{r}} \rfloor -c+1) \leq \log I_{q} (c,  \lfloor m^{\frac{1}{r}} \rfloor -c+1).
\end{equation}
Given integer $m \geq c^{r}$, define $f(z) \equiv \log I_{z}(c,  \lfloor m^{\frac{1}{r}} \rfloor -c+1)$, which is shown in Appendix \ref{app:delayRLNC} (Lem.\ \ref{lem:regbetalogconcave}) to be concave on $z \in (0,1)$. Then \eqref{eq:termwiseInequalityUnderRLNC} follows from Jensen's inequality where the associated RV is again a discrete uniform RV with support $\{q_{1}, \ldots, q_{n}\}$.
\end{IEEEproof}
\begin{remark}
A similar argument establishes the fact that replacing the success probabilities of any subset of receivers with the average over that subset will yield a lower average delay.  Specifically, for any $\qbf \in (0,1)^n$ and any nonempty $\Smc \subseteq [n]$, form $\qbf'$ with $q_j' = q_j$ for $j \not\in \Smc$ and $q_j' = \sum_{i \in \Smc} q_i/|\Smc|$ for $j \in \Smc$.  Then $\Ebb\left[ \left(Y_{n:n}^{(c)}\right)^r \right]$ is no larger under $\qbf'$ than under $\qbf$.
\end{remark}
\begin{remark}
As will be shown in Prop.\ \ref{prop:convergence} (\S\ref{ssec:convergence}), the asymptotic delay per packet (as $c \to \infty$) converges to $1/\min_{j \in [n]} q_j$, meaning the asymptotic delay depends upon the channel vector $\qbf$ only through its minimum component value (bottleneck channel).  Given this, Prop.\ \ref{prop:extremaldelayRLNC} is not surprising since the maximizer of $\min \qbf$ over $\qbf \in \Qmc_q$ is $\qbf = (q,\ldots,q)$: 
\begin{equation}
\argmin_{\qbf \in \Qmc_q} \frac{1}{\min_{j \in [n]} q_j} 
= \argmax_{\qbf \in \Qmc_q} \min_{j \in [n]} q_j = (q,\ldots,q).
\end{equation}
We emphasize, however, that Prop.\ \ref{prop:extremaldelayRLNC} holds for all $c \in \Nbb$, whereas Prop.\ \ref{prop:convergence} holds as $c \to \infty$.
\end{remark}

\section{Dependence of RLC delay on the blocklength $c$}
\label{sec:deponbl} 

In the previous section we derived exact expressions (involving infinite sums) as well as a recurrence for the moments of the RLC block delay.  Although these expressions have the virtue of being exact and computable, they fail to provide insight on the delay's dependence on the key model parameters $c,n$.  This and the subsequent section address this deficiency by studying the asymptotic (occasionally normalized) block delay as $c$ or $n$ grows to infinity.  In this section we investigate the dependence on $c$ with $n$ fixed, while in the next section (\S\ref{sec:deponnumrx}) we investigate the dependence on $n$ with $c$ fixed. This section contains several subsections.  First, we show in \S\ref{ssec:convergence} the delay per packet converges in probability and in $r^{\rm th}$ mean.  This result allows us to easily establish that the asymptotic (in $c$) delay per packet under RLC is lower than the packet delay under UT.  Second, we show in \S\ref{ssec:mono} a stronger result (implying the previous result) that the expected delay per packet under RLC is monotone decreasing in $c$ for both infinite and finite field sizes. Furthermore, the monotonicity properties are investigated from a sample path perspective. Third, we show in \S\ref{ssec:rossdlctightinc} that the normalized lower and upper bounds in \S\ref{ssec:rlncbounds}, specialized to the first moment, are (almost) asymptotically tight as $c \to \infty$.  Throughout this section we employ the superscript $(c)$ on all quantities (not just $Y_{n:n}^{(c)},Y_j^{(c)}$) to emphasize the dependence upon $c$.

\subsection{Asymptotic delay per packet as $c \to \infty$}
\label{ssec:convergence}

Recall the $n \times c$ matrix $\Qbf^{(c)}$ holds the parameters $q_{j,k}^{(c)}$, where $X_{j,k}^{(c)} \sim \mathrm{Geo}(q_{j,k}^{(c)})$. The first main result (Prop.\ \ref{prop:convergence}) establishes that, under some convergence requirements on $\{X_{j,k}^{(c)}\}$, the asymptotic per packet delay ($Y_{n:n}^{(c)}/c$) converges to $1/\min_j q_j$ as $c \to \infty$ in both probability and $r^{\rm th}$ mean, where $q_j$ is the asymptotic (in $c$) average of row $j$ of $\Qbf^{(c)}$.  That is, the asymptotic packet delay depends solely on the reception probability of the bottleneck channel(s), $\min_{j \in [n]} q_j$.  The conditions required in Prop.\ \ref{prop:convergence} are in fact satisfied for the specific $q_{j,k}^{(c)}$ in \eqref{eq:tyt} capturing the effect of the field size on the probability of linear independence (Prop.\ \ref{prop:asymcondsat}).  The second main result (Prop.\ \ref{prop:comparingAsymptotic}) uses this to conclude the asymptotic per packet delay of RLC is lower than the packet delay under UT.  Normalize $(Y_1^{(c)},\ldots,Y_n^{(c)})$ as $(\hat{Y}_1^{(c)},\ldots,\hat{Y}_n^{(c)})$ with 
\begin{equation}
\hat{Y}^{(c)}_j \equiv \frac{Y_j^{(c)}}{c} = \frac{1}{c} \sum_{k=1}^c X_{j,k}^{(c)}, ~ j \in [n].  
\end{equation}
Similarly, normalize $Y_{n:n}^{(c)}$ as $\hat{Y}_{n:n}^{(c)} \equiv Y_{n:n}^{(c)}/c$.
\begin{proposition} 
\label{prop:convergence}
Assume A1: State-dependent receptions, heterogeneous receivers.  Given the $n \times c$ matrix $\Qbf^{(c)}$ with success probabilities $q_{j,k}^{(c)}$, suppose there exist $n$-vectors $\qbf = (q_1,\ldots,q_n)$ and $\boldsymbol\sigma^2 = (\sigma^2_1,\ldots,\sigma^2_n)$ such that the following conditions hold for all $j \in [n]$:
\begin{eqnarray} 
\lim_{c \to \infty} \frac{1}{c} \sum_{k=1}^c \Ebb[X_{j,k}^{(c)}] &=& \frac{1}{q_j} < \infty \nonumber \\
\lim_{c \to \infty} \frac{1}{c} \sum_{k=1}^c \mathrm{Var}(X_{j,k}^{(c)}) &=& \sigma_j^2 < \infty.
\label{eq:convcond}
\end{eqnarray}
Then, as $c \to \infty$, $\hat{Y}_{n:n}^{(c)}$ converges in probability:
\begin{equation}
\hat{Y}_{n:n}^{(c)} \stackrel{\Pbb}{\longrightarrow} (\min_{j \in [n]} q_j)^{-1}.
\end{equation}
Further, fix $r \in \Nbb$ and suppose that each $X_{j,k}^{(c)}$ has uniformly bounded moments up to order $2r$.  Then \eqref{eq:convcond} along with this assumption ensures that, as $c \to \infty$, $\hat{Y}_{n:n}^{(c)}$ converges in $r^{\rm th}$ mean:
\begin{equation}
\hat{Y}_{n:n}^{(c)} \stackrel{L^r}{\longrightarrow} (\min_{j \in [n]} q_j)^{-1}.
\end{equation}
\end{proposition}
The proof of this proposition, as well as that of the following one, can be found in Appendix \ref{app:convergencePf}.

\begin{proposition}
\label{prop:asymcondsat}
Assume A1: State-dependent receptions, heterogeneous receivers.  Conditions \eqref{eq:convcond} of Prop.\ \ref{prop:convergence} on the mean and variance of $\{X_{j,k}^{(c)}\}$ for the weak law of large numbers to apply are satisfied for $q_{j,k}^{(c)} = (1-d^{k-1-c}) q_j$ as in \eqref{eq:tyt}, for any field size $d \geq 2$. 
\end{proposition}

\begin{remark}
In particular for $r=1$ it easily follows from Prop.\ \ref{prop:convergence} that $\Ebb[\hat{Y}_{n:n}^{(c)}] \to 1 / \min_j q_j $ as $c \to \infty$.  Note  \cite{CogShr2011a} establishes this fact as well, but by quite different means. Namely, they derive lower and upper bounds on $\Ebb[\hat{Y}_{n:n}^{(c)}]$ for each $c$ and show that these bounds converge in $c$ to $1 / \min_j q_j $.
\end{remark}

\begin{remark}
Prop.\ \ref{prop:convergence} reveals something important about RLC over the broadcast erasure channel.  Namely, $\Ebb[\hat{Y}_{n:n}^{(c)}] \to \left(\min_{j \in [n]} q_j \right)^{-1}$ depends upon $\qbf$ (and thus $n$) only through the statistic $\min_{j \in [n]} q_j$.  Thus, in particular, the average delay per packet for $n$ channels with success probabilities $\Qbf^{(c)}$ is asymptotically in $c$ the same as the average delay per packet for a single erasure channel (i.e., $n'=1$) with $q_1 = \min_{j \in [n]} q_j$.  In short, the performance of the broadcast erasure channel is asymptotically (in $c$) limited by the bottleneck channel. 
\end{remark}

\begin{proposition}
\label{prop:comparingAsymptotic}
Assume A1: State-dependent receptions, heterogeneous receivers.  The asymptotic expected delay per packet under RLC is superior to the expected delay per packet under UT:
\begin{equation}
\lim_{c \to \infty} \frac{\Ebb[Y_{n:n}^{(c)}]}{c} = \frac{1}{\min_{j \in [n]} q_j}  \leq \Ebb[X_{n:n}].  
\end{equation}
\end{proposition}

\begin{IEEEproof}
Fix the number of receivers $n$ and the reception probabilities $\qbf = (q_1,\ldots,q_n)$.  Suppose $i \in \arg\min_{j \in [n]} q_j$ and fix an arbitrary other index $k \in [n] \setminus \{i\}$.  Consider two copies of this channel, one running RLC with $c \to \infty$ for all $n$ receivers with reception probabilities $\qbf$, and the other running UT with only the pair of $2$ receivers, $i,k$, selected from $[n]$, with reception probabilities $0 < q_i \leq q_k$.  The asymptotic expected delay per packet under RLC is $1/q_i$, and the expected delay per packet under UT in the pruned systems with only these $2$ receivers is, using Cor.\ \ref{cor:firstTwoMomentsMaxGeo} (\S\ref{ssec:utexact}), 
\begin{eqnarray}
\Ebb[\max\left(X_{i}, X_{k} \right)] &=& \frac{1}{q_{i}} + \frac{1}{q_{k}} - \frac{1}{q_{i} + q_{k} - q_{i} q_{k}}.
\end{eqnarray}
Observe that if we can show $1 / \min_j q_j \leq \Ebb[X_{n:n}]$ for the pruned system, then clearly it also holds for the original system with all $n$ receivers.  Simple algebra establishes that $1/q_i \leq \Ebb[\max\left(X_{i}, X_{k} \right)] $, where the inequality is tight only in the degenerate cases when either $q_i = 0$ or $q_k = 1$. 
\end{IEEEproof}

\begin{remark}
The above proof attempted using the lower bound $\psi_{1}(\qbf)$ from Prop.\ \ref{prop:boundsOnEZ}, instead of the exact expression for $\Ebb[X_{n:n}]$ from Prop.\ \ref{prop:kthMomentMaxGeo} (Cor.\ \ref{cor:firstTwoMomentsMaxGeo} in \S\ref{ssec:utexact}) would fail as there are non-degenerate choices for $\qbf$ for which the lower bound on expected delay per packet under UT could be superior to that of RLC.  This is in fact a key motivation for developing exact expressions for the moment of $X_{n:n}$ (Prop.\ \ref{prop:kthMomentMaxGeo}).  The desired inequality does go through when using the upper bound on $\Ebb[X_{n:n}]$, namely $\psi_{1}(\qbf) + 1$, from Prop.\ \ref{prop:boundsOnEZ}, but this inequality is inconclusive as it only relates an upper bound on UT to the asymptotic performance of RLC.  
\end{remark}

\subsection{Monotonicity properties of delay per packet} 
\label{ssec:mono}
The previous subsection demonstrates the advantage in expected delay per packet of RLC over UT in the asymptotic regime, as the blocklength $c \to \infty$.  In this subsection we establish certain monotonicity properties for the delay per packet under RLC as a function of $c$, with the number of receivers $n$ held fixed.  We provide four results.  Props.\ \ref{prop:monotonicityHomogeneous} and \ref{prop:monotonicityHeterogeneous} establish the expected delay per packet is monotone decreasing in the blocklength when the field size is infinite and finite, respectively.  Thus, given two blocklengths $c_1$ and $c_2$, respectively, the above propositions give that the expected delay per packet for under $c_1$ is less than that under $c_2$ if $c_1 > c_2$.  In contrast, Props.\ \ref{prop:nostochasticordering} and \ref{prop:monoSamplePathCombined} establish ``negative'' results.  Prop.\ \ref{prop:nostochasticordering} establishes the sequence of random delay per packet at each receiver $j$, $\{Y^{(c)}_j/c\}_c$, is not stochastically ordered, while Prop.\ \ref{prop:monoSamplePathCombined} (case $ii)$) establishes that in order for the delay per packet under $c_1$ to be no larger than that under $c_2$ for all sample path realizations, it is necessary and sufficient that $c_1 = k c_2$ for some integer $k \geq 2$.  The proofs of Props.\ \ref{prop:monotonicityHomogeneous}, \ref{prop:monotonicityHeterogeneous}, and \ref{prop:monoSamplePathCombined} are in Appendix \ref{app:monoPf}.
 
Consider a collection of RVs $(X_{j,k}^{(c)}, ~ j \in [n], k \in [c])$ where $X_{j,k}^{(c)} \sim \mathrm{Geo}(q_{j,k}^{(c)})$ represents the elapsed time slots between obtaining the $(k-1)^{\rm th}$ and $k^{\rm th}$ linearly independent packet encodings by receiver $j$.  Recall our assumption that $X_{j,k}$ are independent in $j$ and $k$.  Further recall $Y_j^{(c)} = X_{j,1}^{(c)}+\cdots+X_{j,c}^{(c)}$ for $j \in [n]$ represents the decoding delay for each receiver $j$ to recover all $c$ packets.  Define the normalized expected maximum delay difference as
\begin{equation}
\label{eq:delc}
\Delta^{(c)} \equiv \frac{1}{c} \Ebb\left[ Y_{n:n}^{(c)}\right] - \frac{1}{c+1} \Ebb\left[Y_{n:n}^{(c+1)}\right].
\end{equation}
Note that $\{\Delta^{(c)}\}$ is positive for all $c$ iff $\{\Ebb[Y_{n:n}^{(c)}]/c\}$ is monotone decreasing in $c$, and we prove monotonicity by establishing the sign of $\Delta^{(c)}$ for each $c$. 

In the proofs that follow it will be of use to characterize the random set of indices $\Jmc^{(c)} = \argmax_{j \in [n]} Y_j^{(c)}$ with the largest decoding delay.  Equivalently, $\Jmc^{(c)}$ is the set of row indices in the random matrix $\Xbf^{(c)} \equiv (X_{j,k}^{(c)}, j \in [n], k \in [c])$ with the largest row-sum, where ``row-sum'' is defined as the sum of all the elements in the same row.  More generally, we are interested in matrix row selection rules that are column-invariant, as is the max row-sum selection rule.  Define a row-index selection rule $\chi : \Rbb^{n \times c} \to [n]$ that produces a row index $j = \chi(A)$ for any matrix $A$, and in particular returns a random variable $J = \chi(\Xbf^{(c)})$ for the random matrix $\Xbf^{(c)}$.  Let $\sigma$ be any permutation of $[c]$ and $\sigma(\Xbf^{(c)})$ is the matrix $\Xbf^{(c)}$ with its columns permuted according to $\sigma$.  Say $\chi$ is column invariant if $\chi(\Xbf^{(c)}) = \chi(\sigma(\Xbf^{(c)}))$ for all $\sigma \in \Sigma$, where $\Sigma$ is the set of all permutations of $[c]$. In particular we will be using the column invariant selection rule $\hat{J}^{(c)} = \min \hat{\Jmc}^{(c)}$.  In words, $\hat{J}^{(c)}$ is the smallest index in the (random) set of indices in $[n]$ that achieve the maximum $Y_{j}^{(c)}$.  It would work as well to select any other element from $\hat{\Jmc}^{(c)}$; the important point is that the selection $\hat{J}^{(c)}$ is uniquely determined for each $\hat{\Jmc}^{(c)}$.
\begin{proposition} 
\label{prop:monotonicityHomogeneous}
Assume A2: State-independent receptions, heterogeneous receivers.  Then the expected delay per packet $\Ebb[Y_{n:n}^{(c)}/c]$ is decreasing in the code blocklength $c$. 
\end{proposition}
\begin{proposition} 
\label{prop:monotonicityHeterogeneous}
Assume A1: State-dependent receptions, heterogeneous receivers.  Then the expected delay per packet $\Ebb[Y_{n:n}^{(c)}/c]$ is decreasing in the code blocklength $c$.
\end{proposition}
\begin{remark}
The fact that the row selection rule $\chi$ is column-invariant is leveraged in both of the proofs of the above two propositions.  The key difference between them is that when the field size is infinite the matrix $\Xbf^{(c)}$ has iid columns, whereas when the field size is finite the matrix $\Xbf^{(c)}$ has independent columns with a column-dependent distribution.  In fact, in the former case, the proof holds regardless of the distribution (hence our use of the generic RV $Z$ in the proof), whereas in the latter case the proof requires properties of the assumed distribution for $\Xbf^{(c)}$.  
\end{remark}

The following corollary illustrates a consequence of the above monotonicity propositions.  Given a fixed number of $M$ packets, any block coding based strategy must select a {\em block partition} $(m,\cbf)$ consisting of $i)$ the number of blocks  $m \in \Nbb$, and  $ii)$ block sizes $\cbf = (c_1,\ldots,c_m) \in \Nbb^m$ satisfying $c_1 + \cdots + c_m = M$.  
\begin{corollary}
\label{cor:blockpartitionopt}
Assume A1: State-dependent receptions, heterogeneous receivers.  Given a fixed number of $M$ packets, the expected total delay to broadcast those packets is minimized over all block partitions $(m,\cbf)$ by using a single block of $M$ packets. 
\end{corollary}
\begin{IEEEproof}
Following the notation in previous proofs, let $\Ebb[Y_{n:n}^{(c_{i})}]$ be the expected time slots to broadcast a block of $c_{i}$ packets under RLC.  By the monotonicity propositions:
\begin{equation}
\sum_{i=1}^{m} \Ebb[Y_{n:n}^{(c_{i})}]   \! =  \! \sum_{i=1}^{m} \frac{\Ebb[Y_{n:n}^{(c_{i})}]}{c_{i}} c_{i} \geq \sum_{i=1}^{m} \frac{\Ebb[Y_{n:n}^{(M)}]}{M} c_{i} = \Ebb[Y_{n:n}^{(M)}],
\end{equation}
which establishes the desired inequailty for all feasible block partitions $(m,\cbf)$.
\end{IEEEproof}
Having established $\Ebb[Y^{(c)}_{n:n}]/c > \Ebb[Y^{(c+1)}_{n:n}]/(c+1)$ for all $c$, it is natural to wonder if in fact a stochastic ordering holds.  The following proposition is a partial negative answer for the case of infinite field size.  
\begin{proposition}
\label{prop:nostochasticordering}
Assume A2: State-independent receptions, heterogeneous receivers.  The RVs $Y^{(c)}_{j}/c$ and $Y^{(c+1)}_{j}/(c+1)$ are not stochastically ordered for any $c$ and $j$.  
\end{proposition}
\begin{IEEEproof}
We use the shorthand notation $\hat{Y}^{(c)}_j \equiv Y^{(c)}_j/c$.  We first show a stochastic ordering equivalence among RVs $\hat{Y}^{(c)}_j$, $\hat{Y}^{(c+1)}_j$, and $X_j$:
\begin{equation}
\label{eq:nx00}
\left( \hat{Y}^{(c)}_j \gtrless_{\rm st} \hat{Y}^{(c+1)}_j \right) ~ \Leftrightarrow ~ \left( \hat{Y}^{(c)}_j \gtrless_{\rm st} X_j \right), ~ j \in [n].
\end{equation}
We prove the equivalence $\leq_{\rm st}$; the proof for $\geq_{\rm st}$ is similar.    Decompose $\hat{Y}^{(c+1)}_j$ as
\begin{IEEEeqnarray}{rCl}
\IEEEeqnarraymulticol{3}{l}
{\hat{Y}^{(c+1)}_j = \frac{1}{c+1} \left( X_{j,1}^{(c+1)} + \sum_{k=2}^{c+1} X_{j,k}^{(c+1)} \right)
}\nonumber \\* \quad
& \stackrel{\rm d}{=} & \frac{1}{c+1} \left( X_j + \sum_{k=1}^c X_{j,k}^{(c)} \right) = \frac{1}{c+1} X_j + \frac{c}{c+1} \hat{Y}_j^{(c)},\IEEEeqnarraynumspace
\label{eq:yplwqwe}
\end{IEEEeqnarray}
where the equality in distribution holds due to the assumption that the field size is infinite.  Next, decompose $\hat{Y}^{(c)}_j = \frac{1}{c+1} \hat{Y}^{(c)}_j + \frac{c}{c+1} \hat{Y}^{(c)}_j$.  Suppose $\hat{Y}^{(c)}_j \leq_{\rm st} \hat{Y}^{(c+1)}_j$.  Substitution of these two decompositions under the assumed stochastic ordering yields
\begin{equation}
\frac{1}{c+1} \hat{Y}^{(c)}_j + \frac{c}{c+1} \hat{Y}^{(c)}_j \leq_{\rm st} \frac{1}{c+1} X_j + \frac{c}{c+1} \hat{Y}_j^{(c)},
\end{equation}
which is equivalent to $\hat{Y}^{(c)}_j \leq_{\rm st} X_j$.  Next, suppose $\hat{Y}^{(c)}_j \leq_{\rm st} X_j$.  Then reversing the equivalences used in the proof of the forward direction yields $\hat{Y}^{(c)}_j \leq_{\rm st} \hat{Y}^{(c+1)}_j$.

We now show how \eqref{eq:nx00} leads to a contradiction. Due to Prop.\ \ref{prop:monotonicityHomogeneous}, it has to hold that $\hat{Y}_{j}^{(c)} \geq_{\rm st} \hat{Y}_{j}^{(c+1)}$ (note by definition $\Delta^{(c)} = \Ebb\left[\hat{Y}_{n:n}^{(c)}\right] - \Ebb\left[\hat{Y}_{n:n}^{(c+1)}\right]$ and recall properties of stochastic ordering discussed in \S \ref{ssec:contdisbn}). According to \eqref{eq:nx00} we have $\hat{Y}_{j}^{(c)} \geq_{\rm st} X_{j}$ $\forall j$, which gives $\Ebb\left[\hat{Y}_{n:n}^{(c)}\right] - \Ebb\left[X_{n:n}\right] \geq 0$. But this is a contradiction by repeated application of Prop.\ \ref{prop:monotonicityHomogeneous} (also recall Prop.\ \ref{prop:comparingAsymptotic}). Therefore we conclude that $Y_{j}^{(c)} / c$ and $Y_{j}^{(c+1)} / (c+1)$ are not stochastically ordered for any $c$ and $j$.
\end{IEEEproof}

\begin{remark}
Our proof techniques for both Prop.\ \ref{prop:monotonicityHomogeneous} and \ref{prop:monotonicityHeterogeneous} are based on the inequality that the maximum (over row indices $j \in [n]$) of the sum of the first $c$ entries, $\max_{j \in [n]} (X_{j,1} + \cdots + X_{j,c})$, is lower bounded by $X_{\hat{J},1}+\cdots+X_{\hat{J},c}$, where $\hat{J}$ is a (random) row index that maximizes the sum of the first $c+1$ entries.  As illustrated in the proofs of these propositions, application of this inequality requires careful manipulation of the distribution of this random index $\hat{J}$.  From the definition of $\Delta^{(c)}$, one might be led to hope that a simpler inequality may suffice to establish $\Delta^{(c)} > 0$, possibly leading to a simpler proof.  The purpose of this remark is to show that at least one such simpler inequality is insufficient.  Suppose the field size is infinite.  Applying the decomposition in \eqref{eq:yplwqwe} to the definition of $\Delta^{(c)}$ gives 
\begin{IEEEeqnarray}{rCl}
\IEEEeqnarraymulticol{3}{l}
{\Delta^{(c)}  \stackrel{}{=} \Ebb \left[ \frac{1}{c} Y_{n:n}^{(c)} - \frac{1}{c+1} \max_{j \in [n]} \left( X_j  + Y_j^{(c)} \right) \right]
}\nonumber \\* \qquad 
& \geq &  \Ebb \left[ \frac{1}{c} Y_{n:n}^{(c)} - \frac{1}{c+1} \max_{j \in [n]} X_j  - \frac{1}{c+1} \max_{j \in [n]} Y_j^{(c)} \right] \label{eq:rtrtyuyuioio} \IEEEeqnarraynumspace \\
& = & \frac{1}{c+1} \Ebb \left[ \frac{1}{c} Y_{n:n}^{(c)} - X_{n:n} \right] \label{eqref:tmqwxcpot} 
\end{IEEEeqnarray}
where the inequality used in \eqref{eq:rtrtyuyuioio} is $\max_j (a_j + b_j) \leq \max_j a_j + \max_j b_j$.  But, repeated application of Prop.\ \ref{prop:monotonicityHomogeneous} shows \eqref{eqref:tmqwxcpot} is negative, and therefore the inequality \eqref{eq:rtrtyuyuioio} is too weak to establish $\Delta^{(c)} > 0$. 
\end{remark}

We now turn our focus to the sample path relationship of these quantities.  Consider the total delay to transmit $m$ packets using blocklengths $c,c'$, where without loss of generality we suppose $c' > c$.   Using a blocklength $c$ ($c'$) requires $\lceil m/c \rceil$ ($\lceil m/c' \rceil$) blocks, respectively, with the last block possibly being a partial block.  We consider two cases (Prop.\ \ref{prop:monoSamplePathCombined}): $i)$ $m \leq c'$ and $ii)$ $m > c'$, and assume the field size $d$ to be infinite for both.  Under this assumption the only randomness in the delay is from the erasure or nonerasure of each transmission to each receiver, and does not include the randomness from the random linear combinations.  We introduce some notation.  Let $T_{n:n}^{(c,m)}(\omega)$, $T_{n:n}^{(c',m)}(\omega)$ be the delay to broadcast a workload of $m$ packets using a blocklength $c$, $c'$, respectively, for realization $\omega$. 

\begin{proposition}
\label{prop:monoSamplePathCombined}
Assume A2: State-independent receptions, heterogeneous receivers.  Consider any sample path (realization) $\omega$ of erasures and nonerasures to each receiver over the sequence of transmissions.  The total time to complete the transmission of a workload of $m$ packets under blocklengths $c,c'$ with $c < c'$ obeys
\begin{itemize}
\item[$i)$] $T_{n:n}^{(c',m)}(\omega) \leq T_{n:n}^{(c,m)}(\omega)$, if $m \leq c'$,
\item[$ii)$] $T_{n:n}^{(c',m)}(\omega) \leq T_{n:n}^{(c,m)}(\omega)$ iff $c' = k c$ for some integer $k \geq 2$, if $m > c'$,
\end{itemize}
for all realizations $\omega \in \Omega$.
\end{proposition}

\begin{remark}
It follows that the total delay under RLC (for any blocklength $c' \geq 2$) is no worse than that of UT for all sample paths, provided $d = \infty$.
\end{remark}
In summary, Props.\ \ref{prop:monotonicityHomogeneous} and \ref{prop:monotonicityHeterogeneous} establish that the expected delay per packet is decreasing in the blocklength $c$, Prop.\ \ref{prop:nostochasticordering} establishes the random delay per packet sequence $\{Y^{(c)}_{j}/c\}_c$ is not stochastically ordered, and Prop.\ \ref{prop:monoSamplePathCombined} establishes that the delay per packet is not necessarily nonincreasing in $c$ on a sample path basis.  In particular, when the workload exceeds the larger blocklength, sample path ordering of delay per packet is only guaranteed when the blocklength is increased by some integer multiple.

\subsection{Bounds on expected delay per packet}
\label{ssec:rossdlctightinc}

In the previous subsection we established that the expected delay per packet, $\Ebb[Y_{n:n}^{(c)}]/c$, is decreasing in the blocklength $c$ for any finite $n$.  In this subsection we supply lower and upper bounds on $\Ebb[Y_{n:n}^{(c)}]/c$ that provide a more explicit characterization of the dependence of the delay per packet on the blocklength.  These bounds will be shown to be (almost) asymptotically tight in $c$, in that the asymptotic difference between the lower and upper bounds is one.  

In this subsection we restrict our attention to Assumption A3 (state-independent receptions, homogeneous receivers). Recall in Prop.\ \ref{prop:stochorder} we established the stochastic ordering $\frac{1}{\phi(q)} \tilde{Y}^{(c)} \leq_{\rm st} Y^{(c)} \leq_{\rm st} \frac{1}{\phi(q)} \tilde{Y}^{(c)}+c$.  When $d = \infty$ the RVs $Y^{(c)},\tilde{Y}^{(c)}$ are $\mathrm{NegBin}(c,q)$ and $\mathrm{Gamma}(c,1)$, respectively. It follows
\begin{equation}
\frac{1}{\phi(q)} \frac{\Ebb[\tilde{Y}_{n:n}^{(c)}]}{c} \leq \frac{\Ebb[Y_{n:n}^{(c)}]}{c} \leq \frac{1}{\phi(q)} \frac{\Ebb[\tilde{Y}_{n:n}^{(c)}]}{c} + 1.
\end{equation}
First, consider the asymptotic regime. Specializing Prop.\ \ref{prop:convergence} to the homogeneous case yields $\lim_{c \to \infty} \Ebb[Y_{n:n}^{(c)}]/c = 1/q$ and, as the proof of Prop.\ \ref{prop:LBandUBinCHomoRxInftyFieldsize} will show, $\lim_{c \to \infty} \Ebb[\tilde{Y}_{n:n}^{(c)}]/c = 1$. These give the asymptotic ordering 
\begin{IEEEeqnarray}{rCl}
\lim_{c \to \infty} \frac{\Ebb[Y_{n:n}^{(c)}]}{c} = \frac{1}{q}  & \geq & \lim_{c \to \infty} \frac{1}{\phi(q)} \frac{\Ebb[\tilde{Y}_{n:n}^{(c)}]}{c} = \frac{1}{\phi(q)} \nonumber \\ 
\lim_{c \to \infty} \frac{\Ebb[Y_{n:n}^{(c)}]}{c} = \frac{1}{q}  &  \leq & \lim_{c \to \infty} \frac{1}{\phi(q)} \frac{\Ebb[\tilde{Y}_{n:n}^{(c)}]}{c} + 1 = \frac{1}{\phi(q)}+1. \IEEEeqnarraynumspace
\end{IEEEeqnarray}
The fact that $\frac{1}{\phi(q)} \leq \frac{1}{q} \leq \frac{1}{\phi(q)} + 1, ~ \forall q \in (0,1)$ may also be verified directly. Next, consider the finite parameter regime. We have the following proposition.
\begin{proposition}
\label{prop:LBandUBinCHomoRxInftyFieldsize}
Assume A3: State-independent receptions, homogeneous receivers.  For any positive integers $n,c$ we have bounds on the expected delay per packet
\begin{IEEEeqnarray}{rCl}
\frac{1}{\phi(q)} \tilde{l}(n,c) \leq  \! \frac{\Ebb\left[Y_{n:n}^{(c)}\right]}{c}  \! & \leq & \frac{1}{\phi(q)} n Q\left(c+1,Q^{-1}\left(c,\frac{1}{n}\right)\right)  \! +  \! 1 \nonumber \\ 
& \leq & \frac{1}{\phi(q)} \tilde{u}(n,c) + 1,
\end{IEEEeqnarray}
where
\begin{IEEEeqnarray}{rCl}
\tilde{l}(n,c) &  \equiv  & 1+\sqrt{\frac{\log n}{c}}\left(1 - \left(1 - Q\left(c,c + \sqrt{c \log n} \right) \right)^{n-1} \right) \nonumber \\
\tilde{u}(n,c)  & \equiv  & 1 + n/\sqrt{2 \pi c}.
\label{eq:LBandUBinCHomoRxInftyFieldsize}
\end{IEEEeqnarray}
These bounds are asymptotically (in $c$) tight in the sense that:
\begin{equation}
\label{eq:bothConvToOne}
\lim_{c \to \infty} \tilde{l}(n,c) = \lim_{c \to \infty} \tilde{u}(n,c) = 1.
\end{equation}
\end{proposition}
The proof can be found in Appendix \ref{app:rossdlctightincPf}. It is by specializing Prop.\ \ref{prop:boundsOnEY3} to the $r = 1$ case and further bounding the bounds on the normalized (by $c$) expected maximum of iid $\mathrm{Gamma}(c,1)$ RVs: in particular, the lower and upper bounds on $\Ebb[\tilde{Y}_{n:n}^{(c)}]/c$ can both be shown to converge to $1$ (c.f., \eqref{eq:bothConvToOne}).

Fig.\ \ref{fig:rossdlcvsc} shows the exact expected delay per packet and the lower and upper bounds from Prop.\ \ref{prop:LBandUBinCHomoRxInftyFieldsize} vs.\ the blocklength $c$.   Recall the asymptotic delay per packet is $1/q$, which equals $12$ (left) and $5$ (right), respectively.  The lower bound appears to reach its asymptotic value $(1/\phi(q))$ too quickly, while the upper bound appears to track the actual value better.

\begin{figure}[!ht]
\centering
\includegraphics[width=0.49\textwidth]{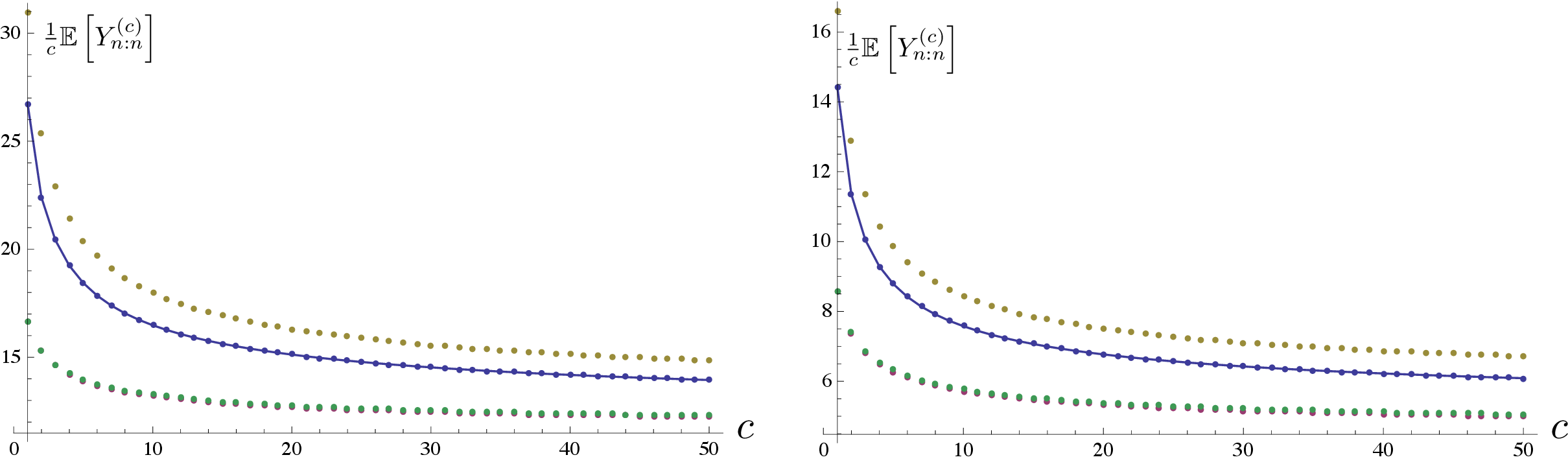}
\caption{The exact expected delay per packet $\frac{1}{c} \Ebb\left[Y_{n:n}^{(c)}\right]$ and the lower and upper bounds from Prop.\ \ref{prop:LBandUBinCHomoRxInftyFieldsize} vs.\ the blocklength $c$ for $n=5$ and $q=1/12$ (left), and $n=12$ and $q = 1/5$ (right). Shown are the exact delay per packet (blue, joined), optimal Ross's upper bound (yellow), de la Cal's lower bound with heuristically chosen free parameter $t(n,c) = c + \sqrt{c \log n}$ (red) and (numerically) optimal free parameter $t(n,c)^{*}$ (green). The exact expected delay per packet, the upper bound, and the lower bounds tend to $1/q$, $1 + 1/\phi(q)$, and $1/\phi(q)$ respectively.}
\label{fig:rossdlcvsc}
\end{figure}

\section{Dependence of RLC delay on the number of receivers $n$}
\label{sec:deponnumrx}

The previous section investigated the dependence of RLC delay on the blocklength $c$, holding the number of receivers $n$ fixed.  In this section we investigate the dependence on $n$, the number of receivers, holding the blocklength $c$ fixed.  Throughout this section we make Assumption A3 (state-independent receptions, homogeneous receivers). The key analytical tool in this section is extreme value theory (EVT) \cite{Res1987,LeaLin1983}, a closely related field of order statistics \cite{DavNag2003}, which studies the convergence of minima and maxima of collections of random variables.  A difficulty in applying EVT to our framework is the fact that for many common discrete distributions (including geometric, Poisson, negative binomial, etc.)  there does not exist a linear normalization such that the normalized maximum order statistic converges in distribution (e.g., \cite[Thm.\ 1.7.13]{LeaLin1983}). This difficulty is circumvented through the use of the stochastic ordering relationship between our discrete delay RV and a continuous analog, as leveraged in \S\ref{sec:delayRLNC}. 

Specifically we will again use the stochastic ordering $\frac{1}{\phi(q)} \tilde{Y}_{n:n}^{(c)} \leq_{\rm st} Y_{n:n}^{(c)} \leq_{\rm st} \frac{1}{\phi(q)} \tilde{Y}_{n:n}^{(c)}+c$ (Prop.\ \ref{prop:stochorder}), relating the discrete $Y_{n:n}^{(c)}$ to its continuous Gamma-distributed analog $\tilde{Y}_{n:n}^{(c)}$.  Thus, the focus in much of this section is the application of EVT (scaling $n \to \infty$) to the continuous RV $\tilde{Y}_{n:n}^{(c)} = \max(\tilde{Y}_1^{(c)},\ldots,\tilde{Y}_n^{(c)})$, with iid $\tilde{Y}_j^{(c)} \sim \mathrm{Gamma}(c,1)$.  The EVT framework requires identification of sequences $(a_n,b_n)$ such that the linearly normalized sequence $(\tilde{Y}_{n:n}^{(c)} - b_n)/a_n$ converges in distribution in $n$ to one of three possible extreme value distributions (Gumbel, Fr\'{e}chet, or Reversed Weibull).   

We note that asymptotic (in $n$) expressions for the first two moments of $Y_{n:n}^{(c)}$ are derived in \cite[Prop.\ 4]{EryOzd2008}, with proof techniques adapted from \cite{GraPro1997}, which relies on tools from complex analysis (e.g., Mellin transform).  Our contribution is two-fold. First, we provide an alternate proof technique to that of \cite{EryOzd2008,GraPro1997}, i.e., EVT applied to the continuous distribution $\tilde{Y}_{n:n}^{(c)}$ stochastically ordered with $Y_{n:n}^{(c)}$, which can be used to demonstrate the same dependence upon $n$ of the (first) moment of $Y_{n:n}^{(c)}$. Although the EVT framework naturally gives convergence in distribution of the normalized sequence of RVs, under mild technical conditions, this convergence also implies convergence in $r^{\rm th}$ mean.  Using this relationship, we are able to derive asymptotic bounds for the first moment of $Y_{n:n}^{(c)}$ and show that they match reasonably well with those obtained in \cite{EryOzd2008}.  Second, we establish that the lower and upper bounds on $\Ebb[\tilde{Y}_{n:n}^{(c)}]$ from \S\ref{sec:delayRLNC} are asymptotically tight in $n$ as $n \to \infty$.  Proving this requires selecting the free parameters (i.e., $s = s_n$ and $t = t_n$) such that the limits can be established.  Our choice of $(s_n,t_n)$ is informed by the normalizing sequences $(a_n,b_n)$ identified in the EVT analysis.  

\subsection{Asymptotic delay as $n \to \infty$}
\label{ssec:evtstuff}

The next two propositions are well-known: the first gives the normalizing sequence $(a_n,b_n)$ for the Gamma distribution to be attracted to the standard Gumbel distribution, and the second relates this convergence to convergence of the normalized moments.  The subsequent corollary follows immediately from these two propositions and the stochastic ordering between $Y_{n:n}^{(c)},\tilde{Y}_{n:n}^{(c)}$.  
\begin{proposition}[\cite{Res1987} \S 1.5 Example 3]
\label{prop:normalizingConstantsCunchangingGamma}
Let $\tilde{Y}_{n:n}^{(c)}$ be the maximum of $n$ iid $(\tilde{Y}_1^{(c)},\ldots,\tilde{Y}_n^{(c)})$ Gamma RVs, with $\tilde{Y}_j^{(c)} \sim \mathrm{Gamma}(c,1)$.  Then
\begin{equation}
\lim_{n \to \infty}
F_{\tilde{Y}_{n:n}^{(c)}}(a_{n} \tilde{y} + b_{n}) 
= \Lambda(\tilde{y}) \equiv \erm^{-\erm^{-\tilde{y}}}, ~ \tilde{y} \in \Rbb,
\end{equation}
for normalizing sequences $(a_n,b_n)$ with
\begin{equation}
\label{eq:bnsequencegamma}
a_{n} = 1, ~~b_{n} = \log n - \log \Gamma(c) + (c-1) \log \log n.
\end{equation}
\end{proposition}
In other words, with the above choice of $(a_{n},b_{n})$, the normalized maximum order statistic $\left(\tilde{Y}_{n:n}^{(c)} - b_{n}\right)/a_n$ converges in distribution to the standard Gumbel, with CDF $\Lambda(y)$.  A random variable $\tilde{Z}$ with distribution $F_{\tilde{Z}}$ is said to belong to the domain of attraction of an extreme value distribution (i.e., Gumbel, Fr\'{e}chet, or Reversed Weibull) if there is a choice of $(a_n,b_n)$ such that the linear scaling $(\tilde{Z}_{n:n} - b_n)/a_n$ converges in distribution to that extreme value distribution, where $\tilde{Z}_{n:n} = \max(\tilde{Z}_1,\ldots,\tilde{Z}_n)$.  The next proposition states that, subject to a mild technical condition, if $\tilde{Z} \sim F_{\tilde{Z}}$ belongs to the Gumbel domain of attraction, then its normalized moments converge to a moment-specific constant.  
\begin{proposition}[\cite{Res1987} Prop.\ 2.1]
\label{prop:momentConvergenceResnick1987}
If $\tilde{Z} \sim F_{\tilde{Z}}$ belongs to the Gumbel domain of attraction under scaling $(a_n,b_n)$, and if $\int_{ - \infty}^{0} \vert z \vert^r \drm F_{\tilde{Z}}(z) < \infty$ for some $r \in \Nbb$, then
\begin{equation} 
\lim_{n \to \infty} \Ebb\left[ \left( \frac{\tilde{Z}_{n:n} - b_{n}}{a_{n}} \right)^r \right] = (-1)^r \Gamma^{(r)}(1),
\end{equation}
where $\tilde{Z}_{n:n} = \max(\tilde{Z}_1,\ldots,\tilde{Z}_n)$, and $\Gamma^{(r)}(1)$ is the $r^{\rm th}$ derivative of the Gamma function evaluated at $1$.
\end{proposition}
Note in particular $(-1)^1\Gamma^{(1)}(1) = \gamma$ for $\gamma \approx 0.5772$ the Euler-Mascheroni constant, and $(-1)^2 \Gamma^{(2)}(1)= \gamma^2 + \pi^2/6 \approx 1.9781$.  Since $\tilde{Y}_j^{(c)}$ is a nonnegative RV, the technical condition is satisfied. 
\begin{corollary}
\label{cor:lowupboundsnegbinevt}
Assume A3: State-independent receptions, homogeneous receivers.  The following lower and upper bounds hold for the asymptotic in $n$ scaled $r^{\rm th}$ power (for $r$ odd) of $Y_{n:n}^{(c)} = \max(Y_1^{(c)},\ldots,Y_n^{(c)})$, for iid $(Y_1^{(c)},\ldots,Y_n^{(c)})$ with $Y_j^{(c)} \sim \mathrm{NegBin}(c,q)$:
\begin{IEEEeqnarray}{rCl}
\IEEEeqnarraymulticol{3}{l}
{(-1)^r \Gamma^{(r)}(1) \leq \lim_{n \to \infty}  \Ebb \left[ \left(\phi(q) Y_{n:n}^{(c)} - b_n\right)^r \right]
}\nonumber \\* \qquad\qquad\qquad \!\!
& \leq & \sum_{s=0}^r \binom{r}{s} (-1)^s \Gamma^{(s)}(1) (\phi(q) c)^{r-s},
\end{IEEEeqnarray}
for $b_n$ in \eqref{eq:bnsequencegamma} and $\Gamma^{(r)}(1)$ the $r^{\rm th}$ derivative of the Gamma function evaluated at $1$.
\end{corollary}
\begin{IEEEproof}
Multiplying the stochastic ordering $\frac{1}{\phi(q)} \tilde{Y}_{n:n}^{(c)} \leq_{\rm st} Y_{n:n}^{(c)} \leq_{\rm st} \frac{1}{\phi(q)} \tilde{Y}_{n:n}^{(c)}+c$ by $\phi(q)$, subtracting $b_n$, raising to the $r^{\rm th}$ power, and applying the binomial theorem to the upper bound gives
\begin{IEEEeqnarray}{rCl}
\IEEEeqnarraymulticol{3}{l}
{\left(\tilde{Y}_{n:n}^{(c)} - b_n \right)^r \leq_{\rm st} \left( \phi(q) Y_{n:n}^{(c)} - b_n\right)^r
}\nonumber \\* \qquad\qquad\qquad \!\!
&  \leq_{\rm st} & \left( \left(\tilde{Y}_{n:n}^{(c)} - b_n \right) + \phi(q)c \right)^r \nonumber \\
& = & \sum_{s=0}^r \binom{r}{s} \left(\tilde{Y}_{n:n}^{(c)} - b_n \right)^s (\phi(q) c)^{r-s}. \IEEEeqnarraynumspace
\end{IEEEeqnarray}
Taking expectations preserves stochastic order, and we apply linearity of expectation to the upper bound:
\begin{IEEEeqnarray}{rCl}
\IEEEeqnarraymulticol{3}{l}
{\Ebb \left[ \left(\tilde{Y}_{n:n}^{(c)} - b_n \right)^r \right] \leq \Ebb \left[ \left( \phi(q) Y_{n:n}^{(c)} - b_n\right)^r \right]
}\nonumber \\* \quad
& \leq & \sum_{s=0}^r \binom{r}{s} \Ebb \left[ \left(\tilde{Y}_{n:n}^{(c)} - b_n \right)^s \right] (\phi(q) c)^{r-s}. 
\end{IEEEeqnarray}
Taking limits and applying Prop.\ \ref{prop:momentConvergenceResnick1987} gives the corollary.
\end{IEEEproof}
The corollary is only established for $r$ odd on account of the fact that the function $z^r$ is increasing in its argument for all $z \in \Rbb$ only for $r$ odd, i.e., it is decreasing in $z$ for $z < 0$ for $r$ even.
For $r=1$, Corollary \ref{cor:lowupboundsnegbinevt} gives
\begin{equation}
0 \leq \lim_{n \to \infty} \left( \Ebb[Y_{n:n}^{(c)}] - \frac{b_n+\gamma}{\phi(q)} \right) \leq c,
\end{equation}
which captures the same dependence upon $n$ as given in the expression for $m_1^{\rm RBC}$ (i.e., $\Ebb[Y_{n:n}^{(c)}]$) in Proposition 4 of \cite{EryOzd2008}. To see this, recall\begin{equation}
m_1^{\rm RBC} = {\rm lq} \left(T\right)+ \frac{1}{2} + \frac{\gamma}{\phi(q)} + h\left({\rm lq } \left(T\right) \right) + o \left(1\right),
\end{equation}
where ${\rm lq} (\cdot) \equiv \log_{\frac{1}{1-q}}\left( \cdot \right) $, 
\begin{equation}
T = n \left( \frac{q}{1-q} \right)^{c-1} \left(\log_{\frac{1}{1-q}}n \right)^{c-1}\left(c-1\right)!^{-1},
\end{equation}
and $h$ is a periodic $C^{\infty}$-function of period $1$ and mean value $0$, whose Fourier coefficients are $\hat{h}(k) = \frac{1}{\log (1-q)} \Gamma\left(\frac{2 \irm k \pi}{\log(1-q)}\right)$.
Observe ${\rm lq}\left(T\right)$ can be rewritten as:
\begin{eqnarray}
{\rm lq} \left(T\right)  =  \frac{1}{\phi(q)} \left( b_{n} + \left(c-1\right) \log \left( \frac{q}{(1-q) \phi(q)}\right) \right).
\end{eqnarray}

We now give a lemma which $i)$ implies the scaling of $\Ebb\left[ \left( Y_{n:n}^{(c)}\right)^{r} \right]$ w.r.t. $n$ is $\left( \frac{1}{\phi(q)} \log n \right)^{r}$ (Prop.\ \ref{prop:scalingOfYnnWRTn}), and $ii)$ is central to proving the subsequent two propositions showing the asymptotic tightness as $n \to \infty$ of the lower and upper bounds on the $r^{\rm th}$ moment of $\tilde{Y}_{n:n}^{(c)}$.
\begin{lemma}
\label{lem:momentConvergence}
Suppose a sequence of random variables $(Z_n)$ is such that 
\begin{equation}
\label{eq:tbbnmnmjk}
\lim_{n \to \infty} \Ebb \left[ \left( \frac{Z_n - b_n}{a_n} \right)^r \right] = c_r
\end{equation}
for sequences $(a_n,b_n)$ independent of $r$ and $(c_r)$ independent of $n$, where $a_n = o(b_n)$.  Then
\begin{equation}
\label{eq:tererqwsdfff}
\lim_{n \to \infty} \frac{\Ebb[Z_n^r]}{b_n^r} = 1, ~ r \in \Nbb.
\end{equation}
\end{lemma}
\begin{IEEEproof}
Using \eqref{eq:tbbnmnmjk} and $a_n = o(b_n)$ gives
\begin{eqnarray}
\label{eq:tbnmklopioyree}
0 & = & \lim_{n \to \infty} \left( \frac{a_n}{b_n} \right)^r \left( \Ebb \left[ \left( \frac{Z_n - b_n}{a_n} \right)^r \right] - c_r \right) \nonumber \\
& = & \lim_{n \to \infty} \Ebb \left[ \left( \frac{Z_n}{b_n} - 1 \right)^r \right].
\end{eqnarray}
Our proof is by induction on $r$.  The above equation immediately gives the base case for $r=1$.  Suppose that \eqref{eq:tererqwsdfff} (with $r$ replaced by $s$) is true for all $s = 1,\ldots,r-1$.  We show this is sufficient to establish the same is true for $s=r$.  Using \eqref{eq:tbnmklopioyree}, the binomial theorem, and the induction hypothesis gives the conclusion:
\begin{eqnarray}
0 
&=& \lim_{n \to \infty} \Ebb \left[ \sum_{s=0}^r \binom{r}{s} \left( \frac{Z_n}{b_n} \right)^s (-1)^{r-s} \right] \nonumber \\
& = & \lim_{n \to \infty}  \Ebb \left[ \left( \frac{Z_n}{b_n} \right)^r \right]  \nonumber \\
& & \negmedspace{} + \sum_{s=0}^{r-1} \binom{r}{s} \left( \lim_{n \to \infty} \Ebb \left[ \left( \frac{Z_n}{b_n} \right)^s \right] \right) (-1)^{r-s}  \nonumber \\
&=& \lim_{n \to \infty}  \Ebb \left[ \left( \frac{Z_n}{b_n} \right)^r \right] + 
\sum_{s=0}^{r-1} \binom{r}{s} (-1)^{r-s}  \nonumber \\
&=& \lim_{n \to \infty}  \Ebb \left[ \left( \frac{Z_n}{b_n} \right)^r \right] - 1.
\end{eqnarray}
\end{IEEEproof}

\begin{proposition}
\label{prop:scalingOfYnnWRTn}
Assume A3: State-independent receptions, homogeneous receivers.  As $n \to \infty$, the scaling of the $r^{\rm th}$ moment of RLC delay $Y_{n:n}^{(c)}$ is: $\Ebb\left[ \left( Y_{n:n}^{(c)}\right)^{r} \right] \sim \left( \frac{1}{\phi(q)} \log n \right)^{r}$.
\end{proposition}
\begin{IEEEproof}
Raising the stochastic ordering $\frac{1}{\phi(q)} \tilde{Y}_{n:n}^{(c)} \leq_{\rm st} Y_{n:n}^{(c)} \leq_{\rm st} \frac{1}{\phi(q)} \tilde{Y}_{n:n}^{(c)}+c$  to the $r^{\rm th}$ power, applying binomial theorem and taking expectations gives
$\frac{1}{\phi(q)^{r}} \Ebb\left[ \left( \tilde{Y}^{(c)}_{n:n}\right)^{r}\right] \leq \Ebb\left[ \left( {Y}^{(c)}_{n:n}\right)^{r}\right] \leq \sum_{p=0}^{r} \binom{r}{p} \frac{1}{\phi(q)^{p}} \Ebb\left[ \left( \tilde{Y}^{(c)}_{n:n}\right)^{p}\right] c^{r-p}$. Now dividing through by $b_{n}^{r}$ for $b_n$ in \eqref{eq:bnsequencegamma} and taking the limit as $n \to \infty$, we have on the left side of the inequality: $\frac{1}{\phi(q)^{r}} \lim_{n \to \infty} \frac{\Ebb\left[ \left( \tilde{Y}^{(c)}_{n:n}\right)^{r}\right]}{b_{n}^{r}} =  \frac{1}{\phi(q)^{r}}$ by Lem.\ \ref{lem:momentConvergence}, and furthermore since $b_{n} \to \infty$ the only term that survives on the right side of the inequality is the one corresponding to $p=r$. Applying Lem.\ \ref{lem:momentConvergence} again, we have
\begin{equation}
\lim_{n \to \infty} \frac{\Ebb\left[ \left(Y^{(c)}_{n:n}\right)^{r}\right]}{b_{n}^{r}} = \frac{1}{\phi(q)^{r}}.
\end{equation}
Finally the observation $b_{n} \sim \log n$ concludes the proof.
\end{IEEEproof}

\subsection{Bounds on the moments of delay as $n \to \infty$}
\label{ssec:rossdlctightinn}

Prop.\ \ref{prop:dlcasymptoticinn} and Prop.\ \ref{prop:rossasymptoticinn}  establish that the lower and upper bounds on $\Ebb\left[\left(Y_{n:n}^{(c)}\right)^r\right]$, resulting from application of the inequalities in Prop.\ \ref{prop:delacal} and Prop.\ \ref{prop:ross} to the Gamma RV stochastic ordering from Prop.\ \ref{prop:stochorder}, can be made (almost) asymptotically tight as $n \to \infty$ for fixed $c$.  
\begin{proposition}
\label{prop:dlcasymptoticinn}
Assume A3: State-independent receptions, homogeneous receivers.  There exists a parameter sequence $(t_n)$ such that the lower bound on $\Ebb\left[\left(\tilde{Y}_{n:n}^{(c)}\right)^r\right]$ from \eqref{eq:delaCallb} can be made asymptotically tight in $n$ for every integer $r \geq 1$.
\end{proposition}

\begin{proposition} 
\label{prop:rossasymptoticinn}
Assume A3: State-independent receptions, homogeneous receivers.  There exists a parameter sequence $(s_n)$ such that the upper bound on $\Ebb\left[\left(\tilde{Y}_{n:n}^{(c)}\right)^r\right]$ from \eqref{eq:rossub} can be made asymptotically tight in $n$ for every integer $r \geq 1$.  In particular $s_n = b_n^r$ is sufficient to establish the bound is asymptotically tight, as well as asymptotically optimal, in the asymptotic sense of \eqref{eq:rossuboptcond}. 
\end{proposition}

The proofs of Props.\ \ref{prop:dlcasymptoticinn} and \ref{prop:rossasymptoticinn} can be found in Appendix \ref{app:deponnumrx}.

Fig.\ \ref{fig:scaling} shows several of the bounds discussed in this section vs.\ the number of receivers, $n$.  The plots attest to the fact that our lower and upper bounds do enclose the exact expected delay for all $n$.  Note that our lower bound is in fact better than the approximation\footnote{By ignoring the $o(1)$ term in Prop.\ 4 of \cite{EryOzd2008}.} of $m_1^{\rm RBC}$ from \cite[Prop.\ 4]{EryOzd2008} for larger $c$, but not for smaller $c$.  

\begin{figure}[!ht]
\centering
\includegraphics[width=0.49\textwidth]{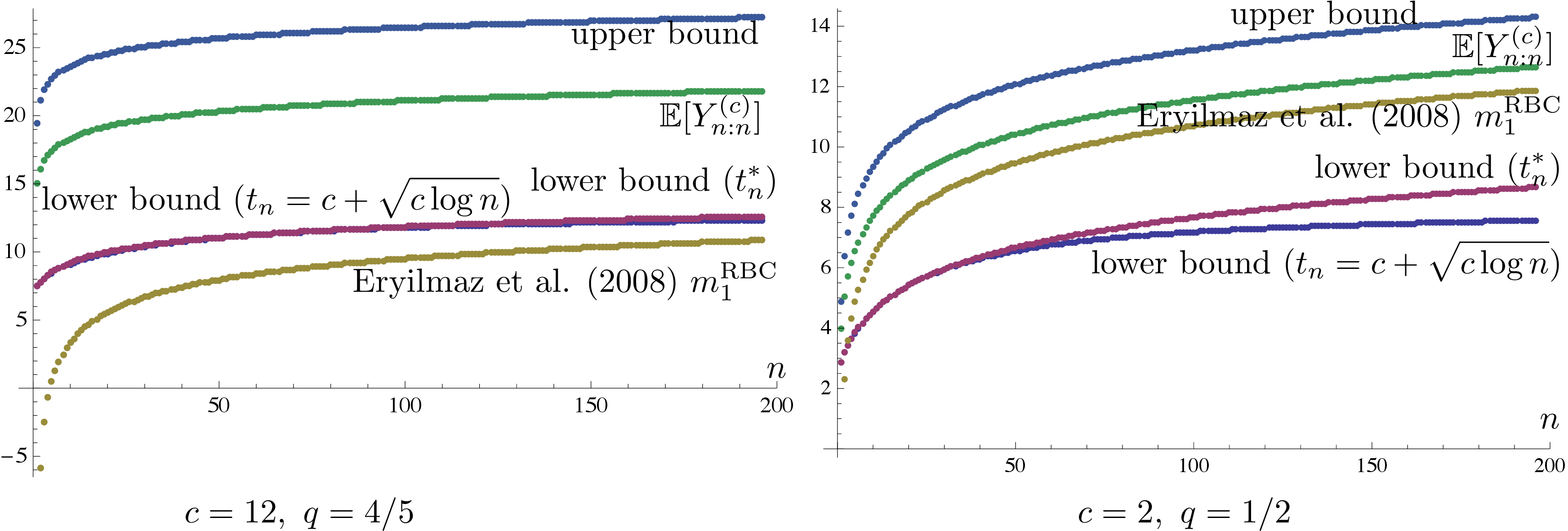}
\caption{The exact block delay ($\Ebb[Y_{n:n}^{(c)}]$), the optimized upper bound, the lower bound with $t_n = c + \sqrt{c \log n}$, the numerically optimized lower bound (with $t_n^*$), and the approximation of $m_1^{\rm RBC}$ from \cite{EryOzd2008}, for $c=12$ and $q=4/5$ (left) and $c=2$ and $q=1/2$ (right).}
\label{fig:scaling}
\end{figure}

\section{Conclusion} 
\label{sec:conclusion}

We have investigated the random block delay, $Y_{n:n}^{(c)}$, when transmitting $c$ packets over a broadcast erasure channel to $n$ receivers with reception probabilities $\qbf$ using both uncoded transmission (UT) and random linear combinations/coding (RLC).  Our key contributions involve bounds, exact expressions, and asymptotic properties (in $c$, or in $n$) for the $r^{\rm th}$ moment of $Y_{n:n}^{(c)}$.  

Several extensions of this work seem natural to us.  First and foremost, our results in \S\ref{sec:deponnumrx} study the delay as $n$ grows large for fixed $c$, while the results in \cite{SwaEry2013} have established that $c = c(n)$ should grow with $n$ according to $\Omega(\log n)$ in order to have a non-diminishing throughput.  One possible approach is to apply our framework of order statistic inequalities, stochastic ordering, and extreme value theory to address this case.  

Second, there are additional results that seem possible.  In particular, $i)$ in \S\ref{ssec:mono} we conjecture the expected delay per packet is not only decreasing in $c$, but in fact convex decreasing in $c$, although we have not been able to prove this; and $ii)$ since the optimal free parameter in de la Cal's bound may be hard to be expressed in closed-form, it is desirable to find a systematic approach for proposing tractable approximations of the optimizer while still keeping the quality of the bound good.

\bibliographystyle{IEEEtran}
\bibliography{BroadcastDelayNetCodBib}

\appendices


\section{Proofs from \S\ref{sec:delayUT}}
\label{app:delayUT}

This appendix contains the proof of Prop.\ \ref{prop:kthMomentGeo} in \S\ref{sec:delayUT}.

\begin{IEEEproof}[Proof of Prop.\ \ref{prop:kthMomentGeo}]
Let $M_X(t) = \frac{q \erm^{t}}{1-(1-q)\erm^{t}}$ for $t < -\log\left(1-q\right)$ denote the moment generating function (MGF) for the geometric distribution. It is well-known that the $r^{\rm th}$ moment of a RV can be obtained by evaluating the $r^{\rm th}$ derivative of its MGF at $0$. Here to retain tractability, we will employ a change of variable, so that the $r^{\rm th}$ derivative of the original MGF (in terms of $t$) is expressed as a weighted sum of all lower-order derivatives of a related function ($\tilde{M}_{X}(s)$) with respect to a new variable ($s$), for which these lower-order derivatives admit a simple general formula and their weights (coefficients) obey a recurrence that can be further solved.

Define a monotone function $s(t) \equiv 1- (1-q) \erm^{t}$, and $\tilde{M}_X(s) \equiv \frac{q}{1-q}\left( s^{-1} - 1 \right)$. Observe $\frac{\drm s(t)}{\drm t} = s(t)-1$ and $\tilde{M}_X(s)$ and $M_X(t)$ are related by $\tilde{M}_X(s(t)) = M_X(t)$. We first prove\footnote{The superscripts of the MGF $M_{X}(t)$ are indicating its (higher) derivatives, whereas the superscripts of the coefficients $a_{l}$ are for indexing (together with their subscripts) these coefficients.} 
\begin{equation} 
\label{eq:rthDerivativeGeoMGF}
M_X^{(r)}(t) =  \left. f_{}(s; r) \right\vert_{s = s(t)},
\end{equation}
where we define a function of $s$ (parameterized by $r$)
\begin{equation}
f_{}(s; r) \equiv  \sum_{l=1}^{r} a_{l}^{(r)} \left( \frac{\drm^{l}}{\drm s^{l}} \tilde{M}_X(s) \right) \left( s -1 \right)^{l},
\end{equation}
and the coefficients $a_{l}^{(r)}$ satisfy the following recurrence:
\begin{equation}
\label{eq:recurrenceGeorthMoment}
a_{l}^{(r)} = l \cdot a_{l}^{(r-1)} + a_{l-1}^{(r-1)}, ~ l \in [r], 
\end{equation}
with boundary conditions\footnote{The support of $r$ in the statement does not contain $0$, yet for the recurrence to work, $r$ is allowed to take $0$, so is the subscript of the $a_{l}^{(r)}$'s.} $a_{l}^{(r)} = 0$ if $l > r$, and $a_{0}^{(r)} = \mathbf{1}_{\left\{ r=0 \right\}}$.

Observe $a_{r}^{(r)} = 1$ for all $r \geq 0$, and $a_{1}^{(r)} = 1$ for all $r > 0$. We prove \eqref{eq:rthDerivativeGeoMGF} by induction on $r$. The base case $r = 1, 2$ can be verified to be true, with $M_X^{(1)}(t) =\frac{q}{1-q} \left( -s(t)^{-1} + s(t)^{-2} \right), M_X^{(2)}(t) = \frac{q}{1-q} \left( s(t)^{-1} - 3 s(t)^{-2} + 2 s(t)^{-3} \right)$.

Assuming \eqref{eq:rthDerivativeGeoMGF} holds for $r > 0$ and applying the chain rule, we can compute the $(r+1)^{\rm th}$ derivative
\begin{IEEEeqnarray} {rCl}
& & M_X^{(r+1)}(t)  \nonumber \\
& = &   \frac{\drm }{\drm t} \left( \left. f_{}(s; r) \right\vert_{s = s(t)} \right)
=  \left. \left[ \left( \frac{\drm }{\drm s} f_{}(s; r) \right) \cdot \frac{\drm s(t)}{\drm t} \right] \right\vert_{s = s(t)}. \IEEEeqnarraynumspace
\end{IEEEeqnarray}
We need to show the terms in the above brackets, viewed as a whole and as a function of $s$, equals $f(s; r+1)$. Applying the rules of differentiation to $\frac{\drm }{\drm s} f_{}(s; r)$ and recalling $\frac{\drm s(t)}{\drm t} = s(t)-1$, we have
\begin{eqnarray}
\label{eq:rthDerivativeGeoMGFInductionStep}
& & \sum_{l=1}^{r} a_{l}^{(r)} \left( \frac{\drm^{l+1}}{\drm s^{l+1}} \tilde{M}_X(s) \right) (s-1)^{l+1} \nonumber \\
 & & \qquad \qquad \quad \! \negmedspace{} + \sum_{l=1}^{r} l \cdot a_{l}^{(r)} \left( \frac{\drm^{l}}{\drm s^{l}} \tilde{M}_X(s) \right) (s-1)^{l}.
\end{eqnarray}
By changing the summing variable from $l$ to $l' \equiv l +1$,  the first summand in \eqref{eq:rthDerivativeGeoMGFInductionStep} can be written as
\begin{eqnarray}
& & \sum_{l'=2}^{r+1} a_{l' - 1}^{(r)} \left( \frac{\drm^{l'}}{\drm s^{l'}} \tilde{M}_X(s) \right) (s-1)^{l'} \nonumber \\
& = & \sum_{l'=1}^{r} a_{l' - 1}^{(r)} \left( \frac{\drm^{l'}}{\drm s^{l'}} \tilde{M}_X(s) \right) (s-1)^{l'} \nonumber \\
& & \qquad \quad \negmedspace{} + a_{r+1}^{(r+1)} \left( \frac{\drm^{r+1}}{\drm s^{r+1}} \tilde{M}_X(s) \right) (s-1)^{r+1},
\end{eqnarray}
due to the fact that $a_{0}^{(r)} = 0$ for $r > 0$, and $a_{r}^{(r)} = 1 = a_{r+1}^{(r+1)}$. Now substituting it back to \eqref{eq:rthDerivativeGeoMGFInductionStep}, we have
\begin{IEEEeqnarray}{rCl}
& & \sum_{l=1}^{r} \left( a_{l - 1}^{(r)} + l \cdot a_{l}^{(r)} \right) \left( \frac{\drm^{l}}{\drm s^{l}} \tilde{M}_X(s) \right) (s-1)^{l}  \nonumber \\
& & \qquad \quad \negmedspace{} + a_{r+1}^{(r+1)} \left( \frac{\drm^{r+1}}{\drm s^{r+1}} \tilde{M}_X(s) \right) (s-1)^{r+1} \nonumber \\
& = & \sum_{l=1}^{r+1} a_{l}^{(r+1)} \left( \frac{\drm^{l}}{\drm s^{l}} \tilde{M}_X(s) \right) (s-1)^{l} = f(s; r+1),
\end{IEEEeqnarray}
where the first equality follows from the recurrence \eqref{eq:recurrenceGeorthMoment}.

Now that we have established \eqref{eq:rthDerivativeGeoMGF} (through \eqref{eq:recurrenceGeorthMoment}), we then need to solve for the $a_{l}^{(r)}$'s and $\frac{\drm^{l}}{\drm s^{l}} \tilde{M}_X(s)$. First, note the recurrence \eqref{eq:recurrenceGeorthMoment} (together with the boundary conditions) exactly coincides with that of the Stirling number of the second kind for which this recurrence can be solved using generating functions (\cite[\S 1.6]{Wil1994}):
\begin{equation}
a_{l}^{(r)} = {r \brace l} = \sum_{m=1}^{l} \left( -1 \right)^{l-m} \frac{m^{r}}{m! \left( l - m \right)!}, ~ r, l \geq 0.
\end{equation}
Second, it can be verified that $\frac{\drm^{l}}{\drm s^{l}} \tilde{M}_X(s) = \frac{q}{1-q} \left( -1 \right)^{l} l! \cdot s^{- \left( l + 1 \right)}$. 
Finally, substitution and evaluation of \eqref{eq:rthDerivativeGeoMGF} at $t = 0$ (equivalently, at $\left. s(t) \right\vert_{t = 0} = q$) yields the desired formula.
\end{IEEEproof}


\section{Proofs from \S \ref{sec:delayRLNC}}
\label{app:delayRLNC}

We first give two lemmas regarding some properties of the \textit{regularized incomplete beta function} $I_{x}(a,b)$ that will be used in the proofs of Prop.\ \ref{prop:boundsOnEY2} (\S\ref{ssec:rlncbounds}) and Prop.\ \ref{prop:extremaldelayRLNC} (\S\ref{ssec:extremalchannels}) respectively. 

The following definitions apply for $a>0$, $b>0$ and $x \in (0,1)$:
\begin{equation} 
\label{eq:defofIncompleteBetaFunc}
I_x(a,b)= \frac{B(x; a, b)}{B(a,b)}, ~B(x; a,b) = \int_{0}^{x} t^{a-1} (1-t)^{b-1} \drm t,
\end{equation}
where $B(x; a,b)$ is the incomplete beta function and $B(a,b) = B(1;a,b)$ is the beta function.

\begin{lemma} 
\label{lem:regbetamono}
$I_x(a,b)$ is increasing in $b$.
\end{lemma}

\begin{IEEEproof}
As 
\begin{equation}  
\frac{\partial}{\partial b} I_{x}(a,b) = \frac{B(a,b) \frac{\partial}{\partial b} B(x; a,b)  - B(x; a,b) \frac{\partial}{\partial b} B(a,b) }{B^{2}(a,b)},
\end{equation}
it suffices to show the positivity of the numerator in the above expression, which becomes:
\begin{eqnarray}
\int_{0}^{x} t^{a-1} (1-t)^{b-1} \log (1-t) \drm t \cdot \int_{0}^{1} t^{a-1} (1-t)^{b-1} \drm t  \nonumber \\
\negmedspace{} -  \int_{0}^{x} t^{a-1} (1-t)^{b-1} \drm t \cdot \int_{0}^{1} t^{a-1} (1-t)^{b-1} \log (1-t) \drm t.
\end{eqnarray}
Observe that due to the negativity of $ \log (1-t)$ for $t \in (0,1)$, the signs of the above four integrals are respectively: $-, +, +, -$, and hence we need to show:
\begin{IEEEeqnarray}  {rCl}
f(x;a,b) & \equiv & \frac{\int_{0}^{x} t^{a-1} (1-t)^{b-1} \log (1-t) \drm t}{\int_{0}^{x} t^{a-1} (1-t)^{b-1} \drm t} \nonumber \\
& > & \frac{\int_{0}^{1} t^{a-1} (1-t)^{b-1} \log (1-t) \drm t}{\int_{0}^{1} t^{a-1} (1-t)^{b-1} \drm t} = f(1;a,b).\IEEEeqnarraynumspace
\end{IEEEeqnarray}
It suffices to show $f(x;a,b)$ is decreasing in $x$ for $x \in (0,1)$, i.e., $\frac{\partial}{\partial x} f(x; a,b) < 0$.
Applying Leibniz's rule, we can compute this partial derivative, which would be negative if
\begin{IEEEeqnarray} {rCl}
& & \int_{0}^{x} t^{a-1} (1-t)^{b-1} \log (1-x)  \drm t \nonumber \\
& < &  \int_{0}^{x} t^{a-1} (1-t)^{b-1} \log (1-t) \drm t. 
\end{IEEEeqnarray}
It can be seen that $0 \leq \log (1-x) \leq \log (1-t)$ when $0 \leq t \leq x < 1$, which means the above inequality holds. Hence we have shown $I_x(a,b)$ is increasing in $b$ when $a > 0, b > 0 \text{ and }x \in (0,1)$.
\end{IEEEproof}

\begin{lemma} 
\label{lem:regbetalogconcave}
The function $\log I_{x}(a,b)$ is concave on $x \in (0,1)$, for integers $a \geq 1, b \geq 1$.
\end{lemma}

\begin{IEEEproof}
We shall verify the second derivative is non-positive. Toward this, 
\begin{IEEEeqnarray}  {rCl}
\frac{\partial }{\partial x} \log I_{x}(a,b)  &=& \frac{1}{I_{x}(a,b)} \frac{\partial }{\partial x} I_{x}(a,b) \nonumber \\
\frac{\partial^{2} }{\partial x^{2}} \log I_{x}(a,b) &=&  \frac{I_{x}(a,b) \frac{\partial^{2} }{\partial x^{2}} I_{x}(a,b) - \left( \frac{\partial }{\partial x} I_{x}(a,b) \right)^{2} }{\left( I_{x}(a,b) \right)^{2}}.
\end{IEEEeqnarray} 
Hence $\frac{\partial^{2} }{\partial x^{2}} \log I_{x}(a,b) \leq 0 \Leftrightarrow I_{x}(a,b) \frac{\partial^{2} }{\partial x^{2}} I_{x}(a,b)  \leq \left( \frac{\partial }{\partial x} I_{x}(a,b) \right)^{2}$. 
As
\begin{eqnarray} 
& & \frac{\partial }{\partial x} I_{x}(a,b)  =  \frac{1}{B(a,b)} \frac{\partial }{\partial x} B(x; a,b), \nonumber \\
& & \frac{\partial^{2} }{\partial x^{2}} I_{x}(a,b)  =  \frac{1}{B(a,b)} \frac{\partial^{2} }{\partial x^{2}} B(x; a,b),
\end{eqnarray}
we need to show:
\begin{equation} 
\label{eq:ththth}
B(x;a,b) \frac{\partial^{2} }{\partial x^{2}} B(x; a,b) \leq \left( \frac{\partial }{\partial x} B(x; a,b) \right)^{2}.
\end{equation}
Recalling the definition of the incomplete beta function given in \eqref{eq:defofIncompleteBetaFunc}, we have 
\begin{IEEEeqnarray}   {rCl}
\frac{\partial }{\partial x} B(x; a,b) &=& x^{a-1} (1-x)^{b-1} \nonumber \\
\frac{\partial^{2} }{\partial x^{2}} B(x; a,b) &=& (a-1)x^{a-2} (1-x)^{b-1} \nonumber \\
& & \negmedspace{} - (b-1)x^{a-1} (1-x)^{b-2}.
\end{IEEEeqnarray}
After canceling $x^{a-2} (1-x)^{b-2}$, \eqref{eq:ththth} becomes
\begin{equation} 
((a-1)(1-x) - (b-1)x) \int_{0}^{x} t^{a-1} (1-t)^{b-1} \drm t \leq x^{a} (1-x)^{b}.
\end{equation}
Define $x^{*} \equiv \frac{a-1}{a-1 + b-1}$ to satisfy $(a-1)(1-x^{*}) - (b-1)x^{*} = 0$ so that $(a-1)(1-x) - (b-1)x \gtrless 0$ iff $x \lessgtr x^*$ (recall $a \geq 1, b \geq 1$). When $x \geq x^{*}$, the above inequality holds since the LHS is nonpositive.

It remains to consider $x < x^*$.  When $a=1$ and/or $b = 1$, the above inequality can be verified to hold for all $x$. Therefore, from now on assume $a > 1, b > 1$ and $0 < x < x^{*}$.  Note $x^* < 1$ under these assumptions. We seek to show 
\begin{IEEEeqnarray}{rCl}
f(x) & \equiv & ((a-1)(1-x) - (b-1)x) \nonumber \\
& & \cdot \int_{0}^{x} t^{a-1} (1-t)^{b-1} \drm t - x^{a} (1-x)^{b} \leq 0. 
\end{IEEEeqnarray}
Note $f(0) = 0$, thus it suffices to show $f'(x) \leq 0$, where 
\begin{IEEEeqnarray}{rCl}
\label{eq:yjyjyj}
f'(x) & = & (2-a-b) \int_{0}^{x} t^{a-1} (1-t)^{b-1} \drm t \nonumber \\
& & \negmedspace{} + x^{a-1} (1-x)^{b-1} (2x-1). 
\end{IEEEeqnarray}
The assumption $a>1$ and $b>1$ ensures the first term in the sum is nonpositive. Observe when $x \in (0, 1/2]$ the second term in \eqref{eq:yjyjyj} is nonpositive, making $f'(x) \leq 0$. Thus in the following we only need to discuss the regime $x \in (1/2, x^{*})$. We again consider two cases: $a \leq b$ and $a >b$.  When $a \leq b$ (equivalently, $x^* \leq 1/2$), there is no $x$ satisfying $x \in (1/2, x^{*})$.  When $a > b$ (equivalently $x^* > 1/2$), and in particular $1/2 < x < x^* < 1$.  Observe $f'(1/2)<0$, so it further suffices to show $f^{''}(x) < 0$ for the above regime of interest.  After some algebra, one can show that
\begin{IEEEeqnarray}{rCl}
f^{''}(x) & = & x^{a-2} (1-x)^{b-2} \nonumber \\
& & \negmedspace{} \cdot \left[ -(a+b) \left( x - \frac{a}{a+b} \right)^{2} + 1 -\frac{ab}{a+b} \right],\IEEEeqnarraynumspace
\end{IEEEeqnarray}
where the sign of $f^{''}(x)$ is determined by the expression in the brackets, which is a concave quadratic taking maximum value of $1-ab/(a+b)$ at $x = a/(a+b)$.  For integers $a,b$ satisfying $a>b>1$, observe $b > 1+b/a$ and thus $ab > a+b > 0$, and so the value of the quadratic at its maximum is negative, establishing $f^{''}(x) < 0$ for $x \in (1/2,x^*)$.  
\end{IEEEproof}

The following three propositions are used in the proof of Prop.\ \ref{prop:boundsOnEY3} (\S \ref{ssec:rlncbounds}). Specifically, Prop.\ \ref{prop:LogConcaveSuffConditions} gives some sufficient conditions for the CCDF of the maximum order statistic of $n-1$ iid continuous non-negative RVs to be logarithmically concave. Prop.\ \ref{prop:GammaLogConcave} verifies for $\mathrm{Gamma}(c,1)$ ($c \in \Nbb$), one of the sufficient conditions in Prop.\ \ref{prop:LogConcaveSuffConditions} is satisfied and hence it can be shown in Prop.\ \ref{prop:delaCalOptimizationGammaRV} that de la Cal's lower bound on $\Ebb \left[ \left( \tilde{Y}^{(c)}_{n:n}\right)^{r} \right]$ can be optimized by finding its unique stationary point.

\begin{proposition}
\label{prop:LogConcaveSuffConditions}
Let an integer $n \geq 2$ be given, for a continuous non-negative random variable $Z$ with distribution and density functions denoted by $F_{Z}$ and $f_{Z}$ respectively, the function $1 - F_{Z}^{n-1}(t)$ (which can be viewed as the CCDF of the maximum order statistic of $n-1$ such iid $Z$'s) is logarithmically concave in $t$ if any of the following conditions holds, for $t$ in a convex subset of $\Rbb$ such that $F_{Z}(t) \in (0, 1)$:
\begin{eqnarray}
\label{eq:SuffCondLogConcavity}
(n-2) \frac{f_{Z}^{2}(t)}{F_{Z}(t)} + \frac{\drm}{\drm t} f_{Z}(t) & \geq & 0 \nonumber \\
\frac{\drm}{\drm t} f_{Z}(t) & \geq & 0  \nonumber \\
\frac{f_{Z}^{2}(t)}{1 - F_{Z}(t)} + \frac{\drm}{\drm t} f_{Z}(t) & \geq & 0.
\end{eqnarray}
Furthermore, if any of the above inequalities is strict, then strict logarithmic concavity holds.
\end{proposition}

\begin{IEEEproof}
By definition, logarithmic concavity means the logarithm of the (positive) function is concave, which in the case of a twice differentiable function $g(t)$ defined on a convex domain, is equivalent to verifying \cite[\S 3.5.2 pp.\ 105]{BoyVan2004} $g(t)  g''(t) - (g'(t))^{2} \leq 0$, and for strictly log-concavity, it suffices to verify this holds with strict inequality. Substituting $1 - F_{Z}^{n-1}(t)$ for $g(t)$, we have
\begin{equation}
\label{eq:logconcaveCCDFMaxos}
\left( 1 - F_{Z}^{n-1}(t) \right) \frac{\drm^{2}}{\drm t^{2}} F_{Z}^{n-1}(t) + \left( \frac{\drm}{\drm t} F_{Z}^{n-1}(t) \right)^{2} \geq 0.
\end{equation}
First, observe that if $\frac{\drm^{2}}{\drm t^{2}} F_{Z}^{n-1}(t) \geq 0$, then \eqref{eq:logconcaveCCDFMaxos} evidently holds. Since
\begin{equation}
\frac{\drm^{2}}{\drm t^{2}} F_{Z}^{n-1}(t) = (n-1) F_{Z}^{n-2}(t) \left[ (n-2) \frac{f_{Z}^{2}(t)}{F_{Z}(t)} + \frac{\drm}{\drm t} f_{Z}(t)\right],
\end{equation}
this yields the first condition in the proposition statement (i.e., the above bracketed terms is non-negative).

Next, we show the second or the third condition in the proposition statement will make \eqref{eq:logconcaveCCDFMaxos} hold. For notational simplicity, we may suppress the dependence on $t$ and also write $f_{Z}'$ for $\frac{\drm}{\drm t} f_{Z}(t)$. Substituting the expressions of the derivatives of $F_{Z}^{n-1}$, \eqref{eq:logconcaveCCDFMaxos} becomes
\begin{IEEEeqnarray}{rCl}
& & (n-1)F_{Z}^{n-2} \left[ (n-2) \frac{f_{Z}^{2}}{F_{Z}} + f_{Z}' \right] (1 - F_{Z}^{n-1}) \nonumber \\
& & \qquad \qquad \qquad \qquad \qquad \!\! \negmedspace{} + (n-1)^{2} F_{Z}^{2 (n-2)} f_{Z}^{2} \geq 0.\IEEEeqnarraynumspace
\end{IEEEeqnarray}
Multiplying through by $F_{Z}$, canceling $(n-1)F_{Z}^{n-2}$ and simplifying, we get
\begin{equation}
\label{eq:logconcaveCCDFMaxosTransf}
\left( 1 - F_{Z}^{n-1} \right) \left[ \left( \frac{n-1}{1 - F_{Z}^{n-1}}  -  1 \right) f_{Z}^{2} + F_{Z} f_{Z}' \right] \geq 0.
\end{equation}
We now focus on conditions ensuring the terms in the brackets in \eqref{eq:logconcaveCCDFMaxosTransf} to be non-negative.

Define a function of $n$ (parameterized by $t$)
\begin{equation}
h(n; t) \equiv \frac{n-1}{1 - F_{Z}^{n-1}}  -  1.
\end{equation}
Observe for the regimes of interest ($n \geq 2$, $F_{Z} \in (0,1)$), $h(n; t) \geq 0$, which implies the second condition in the proposition statement (i.e., $f_{Z}' \geq 0$) will make \eqref{eq:logconcaveCCDFMaxosTransf} (hence \eqref{eq:logconcaveCCDFMaxos}) hold.

The derivative of $h(n; t)$ (w.r.t.\ n) is
\begin{IEEEeqnarray}{rCl}
\label{eq:dhdn}
\frac{\partial}{\partial n} h(n; t) 
& = & \frac{ \left(1 - F_{Z}^{n-1}\right) - (n-1) \left( - F_{Z}^{n-1} \log F_{Z} \right) }{\left( 1 - F_{Z}^{n-1} \right)^{2}}  \nonumber \\
& = & \frac{1 - F_{Z}^{n-1} + F_{Z}^{n-1} \log F_{Z}^{n-1} }{\left( 1 - F_{Z}^{n-1} \right)^{2}},
\end{IEEEeqnarray}
whose numerator can be verified to be decreasing in $F_{Z}^{n-1}$ for $F_{Z}^{n-1} \in (0, 1)$. Since $F_{Z}^{n-1}$ is itself decreasing in $n$ (for fixed $t$), this means the numerator of \eqref{eq:dhdn} is increasing in $n$ for $n \geq 1$ (and fixed $t$). Evaluation of \eqref{eq:dhdn}'s numerator at $n = 1$ yields $0$: this allows us to conclude that $\frac{\partial}{\partial n} h(n; t) \geq 0$ for all $n \geq 2$.

Now that we know $h(n; t)$ is non-negative and is increasing in $n$, in order for \eqref{eq:logconcaveCCDFMaxosTransf} to hold, we only need to verify the terms in its brackets when $n = 2$ is non-negative i.e., 
\begin{equation}
\frac{F_{Z}}{1 - F_{Z}} f_{Z}^{2} + F_{Z} f_{Z}' \geq 0,
\end{equation}
which after canceling $F_{Z}$ is the same as the third condition in the proposition statement.
\end{IEEEproof}

\begin{proposition}
\label{prop:GammaLogConcave}
The CCDF of the maximum order statistic of $n-1$ ($n \in \Nbb$) iid $\mathrm{Gamma}(c,1)$ (for $c \in \Nbb$) RVs is log-concave on $(0, \infty)$.
\end{proposition}
\begin{IEEEproof}
We shall show, for $\mathrm{Gamma}(c,1)$, the third condition in Prop.\ \ref{prop:LogConcaveSuffConditions} is satisfied. 
Toward this, we need to know the CDF, PDF, and the derivative of the PDF of $\tilde{Y}^{(c)} \sim \mathrm{Gamma}(c,1)$. The CDF is usually given as $F_{\tilde{Y}^{(c)}}(\tilde{t}_{c}) = 1 - Q(c, \tilde{t}_{c})$ for $Q(c, \tilde{t}_{c})$ the CCDF of $\mathrm{Gamma}(c,1)$ RV's which is called the \textit{regularized Gamma function}. The PDF and its derivative can then be computed as
\begin{IEEEeqnarray}{rCl}
\label{eq:PDFAndDerivativeTildeYc}
f_{\tilde{Y}^{(c)}} (\tilde{t}_{c})  & =  & \frac{1}{\Gamma(c)}   \frac{\tilde{t}_{c}^{c-1}}{\erm^{\tilde{t}_{c}}}, \nonumber \\
\frac{\drm}{\drm \tilde{t}_{c}} f_{\tilde{Y}^{(c)}} (\tilde{t}_{c})   & =   & \frac{1}{\Gamma(c)}   \frac{\tilde{t}_{c}^{c-2} \left( c - 1 - \tilde{t}_{c} \right)}{\erm^{\tilde{t}_{c}}}.
\end{IEEEeqnarray}
Recall $Q(c, x) \equiv \frac{\Gamma(c, x)}{\Gamma(c)}$ where $\Gamma(c, x) \equiv \int_{x}^{\infty} t^{c-1} \erm^{-t} \drm t$ is the upper incomplete Gamma function. For our problem since $c$ is a natural number, we leverage a connection between the CDF of Poisson and Gamma distributions. Specifically, for $c \in \Nbb$, the CCDF of $\mathrm{Gamma}(c,1)$ evaluated at $\lambda > 0$ namely $Q(c, \lambda)$, is equal to the CDF of $\mathrm{Poi}(\lambda)$ evaluated at $c-1$ namely $\erm^{-\lambda} \sum_{i=0}^{c-1}\frac{\lambda^{i}}{i!}$. This equivalence can be verified by working with $Q(c,\lambda)$ using integration by parts (and mathematical induction). Therefore we have $F_{\tilde{Y}^{(c)}}(\tilde{t}_{c})$ re-expressed as
\begin{equation}
\label{eq:CDFTildeYc}
F_{\tilde{Y}^{(c)}}(\tilde{t}_{c}) = 1 - \erm^{-\tilde{t}_{c}} \sum_{i=0}^{c-1} \frac{\tilde{t}_{c}^{i}}{i!}, \text{ for } c \in \Nbb, \tilde{t}_{c} > 0.
\end{equation}
Denote by $h_{Z, 3}(t)$ the function given in the third equation in \eqref{eq:SuffCondLogConcavity} in the statement of Prop.\ \ref{prop:LogConcaveSuffConditions}. Specializing $h_{Z, 3}(t)$ to our problem using the expressions from \eqref{eq:PDFAndDerivativeTildeYc} and \eqref{eq:CDFTildeYc} and observing $\Gamma(c) = (c-1)!$ for $c \in \Nbb$ yields (after some algebra)
\begin{equation}
\label{eq:ThirdSuffCondTransf}
(c-1)!  \frac{\erm^{\tilde{t}_{c}}}{\tilde{t}_{c}^{c-2}} \sum_{i=0}^{c-1} \frac{\tilde{t}_{c}^{i}}{i!} \cdot h_{\tilde{Y}^{(c)}, 3}(\tilde{t}_{c}) 
 \!=\!   \frac{\tilde{t}_{c}^{c}}{(c-1)!  } + \left( c - 1 - \tilde{t}_{c} \right) \sum_{i=0}^{c-1} \frac{\tilde{t}_{c}^{i}}{i!}.
\end{equation}
Our goal is to show $h_{\tilde{Y}^{(c)}, 3}(\tilde{t}_{c})$ is non-negative. Equivalently, we need to show the RHS of \eqref{eq:ThirdSuffCondTransf} is so. We have

\begin{eqnarray}
\label{eq:polyWithPositiveCoeff}
&  &  \frac{\tilde{t}_{c}^{c}}{(c-1)!  } + \left( c - 1 - \tilde{t}_{c} \right) \sum_{i=0}^{c-1} \frac{\tilde{t}_{c}^{i}}{i!} \nonumber \\
& = &  \sum_{i=0}^{c-1} \frac{c-1}{i!} \tilde{t}_{c}^{i}   -  \sum_{i=0}^{c-2} \frac{\tilde{t}_{c}^{i+1}}{i!} \nonumber \\
& \stackrel{i' \equiv i + 1}{=} & \sum_{i=0}^{c-1} \frac{c-1}{i!} \tilde{t}_{c}^{i}   -  \sum_{i' = 1}^{c-1} \frac{\tilde{t}_{c}^{i'}}{(i'-1)!} \nonumber \\
& = &  (c - 1) + \sum_{i=1}^{c-1} \left( \frac{c-1}{i!}  -  \frac{1}{(i-1)!} \right) \tilde{t}_{c}^{i}  \nonumber \\
& = & \sum_{i=0}^{c-2} \frac{c - 1 - i}{i!} \tilde{t}_{c}^{i},
\end{eqnarray}
which is a polynomial in $\tilde{t}_{c}$ of order $c-2$ with all the coefficients being positive: this verifies, when $c \geq 2$, the non-negativeness of $h_{\tilde{Y}^{(c)}, 3}(\tilde{t}_{c})$ for the domain of interest. When $c = 1$, the RHS of \eqref{eq:ThirdSuffCondTransf} can be verified to equal $0$. Alternatively, we might address the $c = 1$ case by recognizing that $\mathrm{Gamma}(1,1)$ is $\mathrm{Exp}(1)$ and computing $h_{Z, 3}(t)$ using the functions associated with unit rate exponential RV's (instead of \eqref{eq:PDFAndDerivativeTildeYc} and \eqref{eq:CDFTildeYc}), which can be easily shown to be $0$, hence fulfilling the third sufficient condition in Prop.\ \ref{prop:LogConcaveSuffConditions} as well.
\end{IEEEproof}

\begin{proposition}
\label{prop:delaCalOptimizationGammaRV}
de la Cal's lower bound on the $r^{\rm th}$ moment of the maximum order statistic $\tilde{Y}^{(c)}_{n:n}$ for $\tilde{Y}^{(c)}_{j \in [n]} \sim \mathrm{Gamma}(c,1)$ (for $c \in \Nbb$) may be maximized (i.e., made the tightest) by finding its unique stationary point for $\tilde{t}_{c} \geq \left( \Ebb \left[ \left( \tilde{Y}^{(c)} \right)^{r} \right] \right)^{\frac{1}{r}}$. More precisely, this $\tilde{t}_{c}^{*}$ is solved from the following equation on $[\prod_{i=1}^{r} \left(c - 1 + i \right)^{\frac{1}{r}}, \infty)$, where $F_{\tilde{Y}^{(c)}}$ and $f_{\tilde{Y}^{(c)}}$ denote the CDF and PDF of $\mathrm{Gamma}(c,1)$.
\begin{IEEEeqnarray}{rCl}
\label{eq:delaCalLBGammac1Derivative}
& & \left(1 - F_{\tilde{Y}^{(c)}}(\tilde{t}_{c})^{n-1} \right) r \tilde{t}_{c}^{r-1} \nonumber \\
& = & (n-1) \left( \tilde{t}_{c}^{r} - \prod_{i=1}^{r}(c-1+i) \right) F_{\tilde{Y}^{(c)}}(\tilde{t}_{c})^{n-2} f_{\tilde{Y}^{(c)}}(\tilde{t}_{c}).\IEEEeqnarraynumspace
\end{IEEEeqnarray}
\end{proposition}

\begin{IEEEproof}
Recall de la Cal's lower bound for $Z_{n:n}$, the maximum order statistics of $n$ ($n \in \Nbb$) iid random variables $Z_{j \in [n]}$, is given by 
\begin{equation}
\label{eq:delaCalLB}
\Ebb[Z_{n:n}] \geq t_{Z} -  (t_{ Z}- \Ebb[Z]) F_{Z}^{n-1} (t_{Z}) \equiv l_{Z}(t_{Z}), \forall t_{Z} \geq \Ebb[Z],
\end{equation}
where $Z \stackrel{\rm d}{=} Z_{j \in [n]}$, and thus $\Ebb[Z]$ and $F_{Z}$ denote the common mean and CDF of iid $Z_{j}$'s, respectively.

{\bf First}, we start by extending the bounding to the $r^{\rm th}$ ($ r \in \Nbb$) moment of a RV. Denote by $l_{Z}(t_{Z})$ to be the RHS of \eqref{eq:delaCalLB}. If each $Z_{j}$ is itself the $r^{\rm th}$ power of some \textit{non-negative} RV, $V_{j}^{r}$, applying \eqref{eq:delaCalLB} to $\left(V^{r}\right)_{j}, j \in [n]$ gives
\begin{equation}
\label{eq:delaCalLBrthmoment}
\Ebb[\left(V^{r}\right)_{n:n}] \geq t_{V^{r}} -  \left(t_{V^{r}} - \Ebb[\left(V^{r}\right)] \right) F_{\left(V^{r}\right)}^{n-1} (t_{V^{r}}) = l_{V^{r}}(t_{V^{r}}),
\end{equation}
for all $ t_{V^{r}} \geq \Ebb[\left(V^{r}\right)]$,
where similarly $\left(V^{r}\right) \stackrel{\rm d}{=} \left(V^{r}\right)_{j \in [n]}$, and $\Ebb[\left(V^{r}\right)]$ and $F_{\left(V^{r}\right)}$ denote the common mean and CDF of iid $\left(V^{r}\right)_{j}$'s, respectively. Note here and henceforth we may sometimes write $\left(V^{r}\right)$ to highlight it is in itself the RV of interest, whereas $V^{r}$ (i.e., with no parenthesis) denotes the $r^{\rm th}$ power of the random variable $V$ (namely a function of the RV of interest). Now observing 
\begin{equation}
F_{\left(V^{r}\right)} (t_{V^{r}}) 
= \Pbb\left( \left(V^{r}\right) \leq t_{V^{r}} \right) 
= \Pbb\left( V \leq t_{V^{r}}^{\frac{1}{r}}\right)
= F_{V} (t_{V^{r}}^{\frac{1}{r}}),
\end{equation}
where the second equality follows from the non-negativeness of the RV, and further relabeling $t_{V^{r}}^{\frac{1}{r}}$ as $\tilde{t}_{V} \geq 0$, we can re-express $l_{V^{r}}(t_{V^{r}})$ (defined as the RHS of \eqref{eq:delaCalLBrthmoment}) by a new function
\begin{equation}
\tilde{l}_{V} (\tilde{t}_{V}) \equiv \tilde{t}_{V}^{r} - \left( \tilde{t}_{V}^{r} - \Ebb\left[ V^{r} \right] \right) F_{V}^{n-1} (\tilde{t}_{V}), \forall \tilde{t}_{V} \geq \left( \Ebb\left[ V^{r} \right] \right)^{\frac{1}{r}},
\end{equation}
where $V \stackrel{\rm d}{=} V_{j \in [n]}$, and $\Ebb\left[ V^{r} \right]$ and $F_{V}$ denote the common $r^{\rm th}$ moment and CDF of iid $V_{j}$'s, respectively.

For non-negative RVs, it always holds that $\left(V^{r}\right)_{n:n} = V_{n:n}^{r}$, and thus we have obtained the de la Cal's bound applied to the $r^{\rm th}$ moment of $V_{n:n}$
\begin{equation}
\label{eq:delaCalLBrthmomentReexp}
\Ebb\left[V_{n:n}^{r} \right] \geq \tilde{l}_{V} (\tilde{t}_{V}), \text{ for all } \tilde{t}_{V} \geq \left( \Ebb\left[ V^{r} \right] \right)^{\frac{1}{r}}.
\end{equation}
Note, when $r = 1$, $\tilde{l}_{V} (\tilde{t}_{V})$ (\eqref{eq:delaCalLBrthmomentReexp}) recovers $l_{Z}(t_{Z})$ (\eqref{eq:delaCalLB}).

{\bf Second}, we shall show de la Cal's bound for the $r^{\rm th}$ moment of $\tilde{Y}^{(c)}_{n:n}$ (recall each $\tilde{Y}_{j}^{(c)} \sim \mathrm{Gamma}(c,1)$) after a translation is strictly logarithmically concave, which then implies there exists a unique stationary point $\tilde{t}_{c}^{*}$ maximizing de la Cal's bound.

Specializing \eqref{eq:delaCalLBrthmomentReexp} to the $r^{\rm th}$ moment of the maximum order statistic of iid $\tilde{Y}_{j \in [n]}^{(c)} \sim \mathrm{Gamma}(c,1)$ (for $c \in \Nbb$) with free variable denoted $\tilde{t}_{c}$, we have
\begin{IEEEeqnarray}{rCl}
\label{eq:delaCalLBGammac1}
& & \Ebb \left[ \left( \tilde{Y}^{(c)}_{n:n}\right)^{r} \right] 
= \Ebb \left[ \left( \tilde{Y}^{(c)}\right)^{r}_{n:n} \right] \nonumber \\
& \geq & \tilde{t}_{c}^{r} - \left( \tilde{t}_{c}^{r} - \Ebb \left[ \left( \tilde{Y}^{(c)} \right)^{r} \right] \right) F_{\tilde{Y}^{(c)}}^{n-1} (\tilde{t}_{c}) 
\equiv \tilde{l}_{\tilde{Y}^{(c)}} (\tilde{t}_{c}),\IEEEeqnarraynumspace
\end{IEEEeqnarray}
for all $ \tilde{t}_{c} \geq \left( \Ebb \left[ \left( \tilde{Y}^{(c)} \right)^{r} \right] \right)^{\frac{1}{r}}$, where $\tilde{Y}^{(c)} \stackrel{\drm}{=}\tilde{Y}^{(c)}_{j \in [n]}$, and $\Ebb \left[ \left( \tilde{Y}^{(c)} \right)^{r} \right]$ and $F_{\tilde{Y}^{(c)}}$ denote the common $r^{\rm th}$ moment and CDF of iid $\tilde{Y}^{(c)}_{j}$'s, respectively. Note the $r^{\rm th}$ moment of $\mathrm{Gamma}(\alpha,\beta)$ (using the convention (``shape'', ``rate'') for the parameters $(\alpha, \beta)$) has a closed-form expression $\prod_{i=1}^{r} (\alpha - 1 + i)/\beta^{r}$, which gives $\Ebb \left[ \left( \tilde{Y}^{(c)} \right)^{r} \right] = \prod_{i=1}^{r} (c-1+i)$, and thus $\tilde{t}_{c} \in [\prod_{i=1}^{r} \left(c - 1 + i \right)^{\frac{1}{r}}, \infty)$. Also, recognize the lower bound in \eqref{eq:delaCalandRossSecFourA} is just \eqref{eq:delaCalLBGammac1} by expressing $F_{\tilde{Y}^{(c)}}(\tilde{t}_{c})$ as $1 - Q(c,\tilde{t}_{c})$ and relabeling $\tilde{t}_{c}$ as $t$.

Applying a constant (in $\tilde{t}_{c}$) translation, we get
\begin{IEEEeqnarray}{rCl}
\label{eq:delaCalLBTranslated}
& & \tilde{l}_{\tilde{Y}^{(c)}} (\tilde{t}_{c}) - \Ebb \left[ \left( \tilde{Y}^{(c)} \right)^{r} \right] \nonumber \\
& = & \left( \tilde{t}_{c}^{r} - \Ebb \left[ \left( \tilde{Y}^{(c)} \right)^{r} \right] \right)    \left( 1 - F_{\tilde{Y}^{(c)}}^{n-1} (\tilde{t}_{c}) \right).
\end{IEEEeqnarray}
We shall show \eqref{eq:delaCalLBTranslated} is log-concave. Recall log-concavity is preserved under the product of two functions (\cite[3.5.2 pp.\ 105]{BoyVan2004}) and it is not hard to verify that strict log-concavity is preserved if furthermore at least one of the functions is strictly log-concave. We use the following two steps to show the two multiplicative terms in the RHS of \eqref{eq:delaCalLBTranslated} are (strictly) log-concave: one for each term.

\underline{Step 1}: $\tilde{t}_{c}^{r} - \Ebb \left[ \left( \tilde{Y}^{(c)} \right)^{r} \right]$  is strictly log-concave in $\tilde{t}_{c}$ for $\tilde{t}_{c} \in [\prod_{i=1}^{r} \left(c - 1 + i \right)^{\frac{1}{r}}, \infty)$.

To verify a twice differentiable univariate function $g(t)$ defined on a convex domain to be (strictly) log-concave, we need to verify (strict inequality suffices [but is not necessary] for verifying strict log-concavity) (\cite[\S 3.5.2 pp.\ 105]{BoyVan2004})
\begin{equation}
\label{eq:CondForLogConcavity}
g(t) g'' (t) - (g'(t))^{2} \leq 0.
\end{equation}
Substituting $\tilde{t}_{c}$ for $t$ and $\tilde{t}_{c}^{r} - \Ebb \left[ \left( \tilde{Y}^{(c)} \right)^{r} \right]$ for $g(t)$, the LHS of \eqref{eq:CondForLogConcavity} becomes
\begin{equation}
- r \tilde{t}_{c}^{r-2} \left( \tilde{t}_{c}^{r} + (r-1) \Ebb \left[ \left( \tilde{Y}^{(c)} \right)^{r} \right] \right) < 0,
\end{equation}
which verifies the strict log-concavity of $\tilde{t}_{c}^{r} - \Ebb \left[ \left( \tilde{Y}^{(c)} \right)^{r} \right]$ for the domain of interest.

\underline{Step 2}: $1 - F_{\tilde{Y}^{(c)}}^{n-1} (\tilde{t}_{c})$ is log-concave in $\tilde{t}_{c}$ for $\tilde{t}_{c} \in  [\prod_{i=1}^{r} \left(c - 1 + i \right)^{\frac{1}{r}}, \infty)$.

This follows from Prop.\ \ref{prop:GammaLogConcave}.

Finally, we can verify the RHS of \eqref{eq:delaCalLBGammac1} evaluated at $\tilde{t}_{c} = \prod_{i=1}^{r} \left(c - 1 + i \right)^{\frac{1}{r}}$ and $\tilde{t}_{c} \to \infty$ (applying the L'H\^opital's rule) both give $\Ebb \left[ \left( \tilde{Y}^{(c)} \right)^{r} \right]$, the trivial lower bound on $\Ebb \left[ \left( \tilde{Y}^{(c)}_{n:n}\right)^{r} \right]$, this means the de la Cal's bound can indeed be made the tightest over the (interior of the) domain of interest, where the optimizer $\tilde{t}_{c}^{*}$ is the unique stationary point of \eqref{eq:delaCalLBGammac1}. Differentiation of \eqref{eq:delaCalLBGammac1} yields \eqref{eq:delaCalLBGammac1Derivative} and the proof is completed.
\end{IEEEproof}

\begin{IEEEproof}[Proof of Prop.\ \ref{prop:boundsOnEY3}]
The ordering $\frac{1}{\phi(q)^{r}} \Ebb\left[ \left( \tilde{Y}^{(c)}_{n:n}\right)^{r}\right] \leq \Ebb\left[ \left( {Y}^{(c)}_{n:n}\right)^{r}\right] \leq \sum_{p=0}^{r} \binom{r}{p} \frac{1}{\phi(q)^{p}} \Ebb\left[ \left( \tilde{Y}^{(c)}_{n:n}\right)^{p}\right] c^{r-p}$ follows from Prop.\ \ref{prop:stochorder}. It remains to show the bounds as given in \eqref{eq:delaCalandRossSecFourA} are valid. The bounds we employ are introduced in Appendix \ref{app:delaCalAndRossBounds}. The lower bound is an application of \eqref{eq:delaCallb}. The derivation of the bound as well as its optimization can be found in Prop.\ \ref{prop:delaCalOptimizationGammaRV}.

The upper bound is a direct application of \eqref{eq:rossub} to $\Ebb\left[ \left( \tilde{Y}^{(c)}_{n:n}\right)^{r}\right]$, followed by an exchange of the order of integration, and repeated use of the recursion $Q(c+1,x) = Q(c,x) + x^c \erm^{-x}/\Gamma(c+1)$.  In particular, applying the upper bound and exchanging the order of integration gives:
\begin{eqnarray}
\label{eq:RossHighermomentClosedForm}
\Ebb\left[\left(\tilde{Y}_{n:n}^{(c)}\right)^{r}\right] 
\!\!& \leq &\!\! s + n \int_s^{\infty} \Pbb \left(\left(\tilde{Y}^{(c)}\right)^{r} > t \right) \drm t \nonumber \\
&=& s + n \int_s^{\infty} Q(c, t^{\frac{1}{r}}) \drm t \nonumber \\
&=& s + \frac{n}{\Gamma(c)} \int_s^{\infty} \left( \int_{t^{\frac{1}{r}}}^{\infty} u^{c-1} \erm^{-u} \drm u \right) \drm t \nonumber \\
&=& s + \frac{n}{\Gamma(c)} \int_{s^{\frac{1}{r}}}^{\infty} \left( \int_s^{u^{r}} \drm t \right) u^{c-1} \erm^{-u} \drm u \nonumber \\
&=& s + \frac{n}{\Gamma(c)} \int_{s^{\frac{1}{r}}}^{\infty} (u^{r}-s) u^{c-1} \erm^{-u} \drm u.
\end{eqnarray}
Now after the last step of the following equation, repeated application of the recursion mentioned above gives the upper bound $u_{\tilde{Y}}(s,n,c,r)$ in \eqref{eq:delaCalandRossSecFourA}.
\begin{IEEEeqnarray}{rCl}
\IEEEeqnarraymulticol{3}{l}
{\Ebb\left[\left(\tilde{Y}_{n:n}^{(c)}\right)^{r}\right]
}\nonumber \\* \quad
& \leq & s + \frac{n}{\Gamma(c)} \left( \int_{s^{\frac{1}{r}}}^{\infty} u^{r+c-1} \erm^{-u} \drm u - s \int_{s^{\frac{1}{r}}}^{\infty} u^{c-1} \erm^{-u} \drm u \right) \nonumber \\
&=& s + \frac{n}{\Gamma(c)} \left( \Gamma(c+r,s^{\frac{1}{r}}) - s \Gamma(c,s^{\frac{1}{r}}) \right) \nonumber \\
&=& s + n \left( \frac{\Gamma(c+r)}{\Gamma(c)} Q(c+r,s^{\frac{1}{r}}) - s Q(c,s^{\frac{1}{r}}) \right).
\end{IEEEeqnarray}
The optimal $s^*$ may be found as (see e.g., \eqref{eq:rossuboptcond}):
\begin{IEEEeqnarray}{rCl}
n \Pbb\left(\left(\tilde{Y}^{(c)}\right)^{r} > s^*\right) = 1  & \Rightarrow &  n Q\left(c, \left(s^*\right)^{\frac{1}{r}}\right) = 1 \nonumber \\
& \Rightarrow &  s^* = \left( Q^{-1}\left(c,\frac{1}{n}\right) \right)^{r}. \quad
\end{IEEEeqnarray}
Here, $Q^{-1}(c,u) = x$ is the inverse with respect to $x$ of $Q(c,x)$ so that $Q(c,x) = u$, for $u \in [0,1]$.  Substitution of the optimal $s^{*}$ into the upper bound $u_{\tilde{Y}}(s,n,c,r)$ in \eqref{eq:delaCalandRossSecFourA} gives
\begin{IEEEeqnarray}{rCl}
& & \frac{\Gamma(c+r)}{\Gamma(c)} \nonumber \\
& &  \negmedspace{} + n  \frac{\Gamma(c+r)}{\Gamma(c)}  \left( s^{*}\right)^{\frac{c+r-1}{r}} \erm^{-\left(s^{*}\right)^{\frac{1}{r}}} \sum_{m=0}^{r-1} \frac{1}{\left( s^{*}\right)^{\frac{m}{r}} \Gamma(c+r-m)}.\IEEEeqnarraynumspace
\end{IEEEeqnarray}
For $r=1$, the above expression simplifies to $c + n c \frac{(s^*)^c \erm^{-s^*}}{\Gamma(c+1)}$, which after applying the recursion and taking into account the stochastic ordering between $Y_{n:n}^{(c)}$ and $\tilde{Y}_{n:n}^{(c)}$ yields \eqref{eq:rossuboptfornegbin}.  
\end{IEEEproof}

\begin{IEEEproof}[Proof of Prop.\ \ref{prop:exacty}]
We first derive \eqref{eq:tu}. We start by using an expression for the expectation of a nonnegative RV with support $\{c^{r}, (c+1)^{r}, \ldots \}$ in terms of its complementary CDF:
\begin{IEEEeqnarray} {rCl}
\Ebb\left[\left(Y_{n:n}^{(c)}\right)^r \right] & = & c^{r} + \sum_{m = c^{r}}^{\infty} \Pbb\left(\left(Y_{n:n}^{(c)}\right)^r > m\right) \nonumber \\ 
& = & c^{r} + \sum_{m = c^{r}}^{\infty} \Pbb\left(Y_{n:n}^{(c)} > \lfloor m^{\frac{1}{r}} \rfloor \right) \nonumber \\ 
& = & c^{r} + \sum_{m = c^{r}}^{\infty} \left(  1 -  \Pbb\left(Y_{n:n}^{(c)} \leq \lfloor m^{\frac{1}{r}} \rfloor \right) \right) \nonumber \\
& = & c^{r} + \sum_{m = c^{r}}^{\infty} \left(  1 -   \prod_{j=1}^{n} \Pbb\left(Y_j^{(c)} \leq \lfloor m^{\frac{1}{r}} \rfloor \right) \right).\IEEEeqnarraynumspace
\label{eq:tz}
\end{IEEEeqnarray}
Next, fix $j \in [n]$ and observe $\{Y_j^{(c)} \leq \lfloor m^{\frac{1}{r}} \rfloor \}$ means that there are $c$ successes in the first $t$ trials, for some $t \in \{c,c+1,\ldots,\lfloor m^{\frac{1}{r}} \rfloor\}$.  Conditioning on $t$, the $c$ successes occurred over trials $\{1,\ldots,t\}$ with the number of trials between successive successes given by a $c$-vector $\boldsymbol\alpha = (\alpha_1,\ldots,\alpha_c)$ such that $\alpha_k \in \Nbb$ and $\alpha_1 + \cdots + \alpha_c = t$.  Fixing the particular sequence of successes and failures over trials $\{1,\ldots,t\}$, that sequence of outcomes has probability $\prod_{k \in [c]} (1-q_{j,k})^{\alpha_k-1} q_{j,k}$.  This yields \eqref{eq:tu}.

To obtain \eqref{eq:tv} from \eqref{eq:tu} set $q_{j,k} = q_j$ and observe
\begin{IEEEeqnarray}{rCl}
\sum_{\boldsymbol\alpha \in \Amc_t} \prod_{k=1}^{c}(1-q_{j})^{\alpha_{k} - 1}q_j 
& = & \sum_{\boldsymbol\alpha \in \Amc_t} q_j^{c} (1-q_j)^{t-c} \nonumber \\
& = & q_j^{c} (1-q_j)^{t-c} |\Amc_t|,
\end{IEEEeqnarray}
where $|\Amc_t|$ is the cardinality of $\Amc_t$, i.e., the number of $c$-vectors of positive integers that sum to $t$.  It can be shown that $|\Amc_t| = \binom{t-1}{c-1}$; this yields \eqref{eq:tv}.

To derive \eqref{eq:tw} which holds for the state-independent case, observe $Y_j^{(c)}$ is a negative binomial RV, i.e., the number of trials required to obtain $c$ receptions where each trial yields a reception with probability $q_j$.  Let $W_j \sim \mathrm{Bin}(\lfloor m^{\frac{1}{r}} \rfloor,q_j)$ be a binomial RV counting the number of receptions in $\lfloor m^{\frac{1}{r}} \rfloor$ trials, each trial having reception probability $q_j$.  Observe the equivalence of the events $\{Y_j^{(c)} > \lfloor m^{\frac{1}{r}} \rfloor\} = \{ W_j < c\}$.  Taking probabilities of both sides, and applying \eqref{eq:tx} yields:
\begin{eqnarray}
\Pbb\left(Y_j^{(c)} > \lfloor m^{\frac{1}{r}} \rfloor \right) &=& \Pbb(W_j < c) = 1 - \Pbb(W_j \geq c) \nonumber \\
&=& 1 - I_{q_j}(c,\lfloor m^{\frac{1}{r}} \rfloor-c+1).
\label{eq:taa}
\end{eqnarray}
Starting from \eqref{eq:tz} and substituting \eqref{eq:taa} yields \eqref{eq:tw}.
\end{IEEEproof}

\section{Proofs from \S \ref{sec:deponbl}}
\label{app:deponbl}

Proofs from \S\ref{ssec:convergence}, \S\ref{ssec:mono}, and \S\ref{ssec:rossdlctightinc} are given in Appendix \ref{app:convergencePf}, Appendix \ref{app:monoPf}, and Appendix \ref{app:rossdlctightincPf}, respectively.

\subsection{Proofs from \S\ref{ssec:convergence}}
\label{app:convergencePf}

\begin{IEEEproof}[Proof of Prop.\ \ref{prop:convergence}]
We first provide the outline of the proof.  To establish convergence in probability of $\hat{Y}_{n:n}^{(c)}$ we first establish convergence in probability of each component of the vector $(\hat{Y}^{(c)}_1,\ldots,\hat{Y}^{(c)}_n)$ via the weak law of large numbers (WLLN) for independent but not necessarily identically distributed RVs.  Convergence in probability of each component of a sequence of vector-valued random variables to a limit ensures convergence in probability of the vector as a whole.  Next, use the fact that convergence in probability is preserved under continuous functions, and in particular the function $\max(y_1,\ldots,y_n)$.  This establishes convergence in probability of $\hat{Y}_{n:n}^{(c)}$.  

Furthermore, to establish convergene in $r^{\rm th}$ mean of $\hat{Y}_{n:n}^{(c)}$, we shall show the sequence $\left\{ \right\vert \hat{Y}_{n:n}^{(c)} \vert^{r}\}$ ($c = 1, 2, \ldots$) is \textit{uniformly integrable} (UI), we can then employ the fact that a sequence of RVs that converges in probability and is UI with parameter $r$ must converge in $r^{\rm th}$ mean.

We now establish convergence in probability.  Prop.\ \ref{prop:asymcondsat} shows \eqref{eq:convcond} is satisfied, then according to the WLLN for independent but not necessarily identically distributed RVs (e.g., \cite[Thm.\ 1.1 in \S 7]{Gut2009}), we have:
\begin{equation}
\hat{Y}_j^{(c)} \stackrel{\Pbb}{\longrightarrow} \Ebb\left[ \lim_{c \to \infty} \frac{1}{c} \sum_{k=1}^c X_{jk}^{(c)}\right] = \frac{1}{q_j}, ~ j \in [n].
\end{equation}
Convergence in probability of each component $j \in [n]$ ensures convergence in probability of the vector as a whole:
\begin{equation}
\left(\hat{Y}_1^{(c)}, \ldots,\hat{Y}_n^{(c)}\right) \stackrel{\Pbb}{\longrightarrow} \left(\frac{1}{q_1},\ldots,\frac{1}{q_n} \right).
\end{equation}
Since convergence in probability is preserved under continuous functions such as $\max(y_1,\ldots,y_n)$, it follows that:
\begin{equation}
\hat{Y}_{n:n}^{(c)} = \max\left(\hat{Y}_1^{(c)}, \ldots,\hat{Y}_n^{(c)}\right) \stackrel{\Pbb}{\longrightarrow} \max \left(\frac{1}{q_1},\ldots,\frac{1}{q_n} \right).
\end{equation}
This establishes convergence in probability of $\hat{Y}_{n:n}^{(c)}$ in $c$ to $1/\min_j q_j$.  

We next establish convergence in $r^{\rm th}$ mean. We use a theorem (\cite[Thm.\ 4.1 in \S7]{Gut2009}) which states that convergence in probability along with uniform integrability with parameter $r$ ensures convergence in $L^{r}$. Namely, we need to show the sequence $\left\{ \right\vert \hat{Y}_{n:n}^{(c)} \vert^{r}\}_{c}$ is UI. Toward this, we apply Thm.\ 4.2 (also see Remark 4.3 on pp.\ 198) in \cite[\S 7]{Gut2009} which states to show a sequence $\{Z_{c}\}$ is UI, it \textit{suffices} to find $\eta > 1$ so that $\sup_{c} \Ebb[\vert Z_{c} \vert^{\eta}] < \infty$. Here we choose $\eta = 2$, that is, we wish to show
\begin{equation}
\label{eq:UIobjective}
\sup_{c} \Ebb \left[ \left(\hat{Y}_{n:n}^{(c)}\right)^{2r}\right] < \infty.
\end{equation}
Note the RVs are non-negative, $r \in \Nbb$ is given and fixed. We observe 
\begin{IEEEeqnarray}{rCl}
\IEEEeqnarraymulticol{3}{l}
{\Ebb \left[ \left(\hat{Y}_{n:n}^{(c)}\right)^{2r}\right] 
 =  \Ebb \left[ \left(\max_{j \in [n]}\hat{Y}_{j}^{(c)}\right)^{2r}\right] 
}\nonumber \\* \quad
& = & \Ebb \left[ \max_{j \in [n]} \left( \hat{Y}_{j}^{(c)}\right)^{2r} \right]
\leq \sum_{j=1}^{n} \Ebb \left[ \left( \hat{Y}_{j}^{(c)}\right)^{2r} \right].\IEEEeqnarraynumspace
\end{IEEEeqnarray}
Since $n$ is fixed, to show \eqref{eq:UIobjective}, we need to show for any $j \in [n]$, $\Ebb \left[ \left( \hat{Y}_{j}^{(c)}\right)^{2r} \right] < \infty$ uniformly in $c$. By definition
\begin{equation}
\label{eq:substituteX}
\Ebb \left[ \left( \hat{Y}_{j}^{(c)}\right)^{2r} \right]
 = \Ebb \left[ \left( \frac{1}{c} \sum_{k = 1}^{c} X_{j, k}^{(c)} \right)^{2r} \right].
\end{equation}
As a consequence of the convexity of the function $z^{\alpha}$ for $\alpha = 2r > 0, z > 0$, applying Jensen's inequality (with a discrete uniform random variable $V$ with PMF $\Pbb(V = X_{j,k}^{(c)})  = \frac{1}{c}, \forall k \in [c]$) gives
\begin{equation}
\label{eq:AppJensen}
\left( \frac{1}{c} \sum_{k = 1}^{c} X_{j, k}^{(c)} \right)^{2r} \leq  \frac{1}{c} \sum_{k = 1}^{c} \left( X_{j, k}^{(c)} \right)^{2r}.
\end{equation}
Now taking expectation, we have from \eqref{eq:substituteX} and \eqref{eq:AppJensen}
\begin{equation}
\label{eq:UpToOrder2r}
\Ebb \left[ \left( \hat{Y}_{j}^{(c)}\right)^{2r} \right]
\leq \Ebb \left[ \frac{1}{c} \sum_{k = 1}^{c} \left( X_{j, k}^{(c)} \right)^{2r} \right]
\! = \! \frac{1}{c} \sum_{k=1}^{c} \Ebb \left[ \left( X_{j, k}^{(c)} \right)^{2r} \right].
\end{equation}
Therefore, if each $X_{j,k}^{(c)}$ has uniformly in $c$ bounded moments up to order $2r$, i.e., $\Ebb\left[ \left( X_{j, k}^{(c)} \right)^{2r} \right] \leq R_{X} < \infty, \forall c$, then we conclude from \eqref{eq:UpToOrder2r} that $\Ebb \left[ \left( \hat{Y}_{j}^{(c)}\right)^{2r} \right]$ is also uniformly bounded in $c$ by $R_{X} < \infty$: this shows \eqref{eq:UIobjective} and hence the proof is completed.
\end{IEEEproof}

\begin{IEEEproof}[Proof of Prop.\ \ref{prop:asymcondsat}]
Assume A1: State-dependent receptions, heterogeneous receivers.  In order to establish \eqref{eq:convcond} holds for $q_{j,k} = (1-d^{k-1-c})q_j$ it suffices to establish $a)$ $\lim_{c \to \infty} \frac{1}{c}  \sum_{k=1}^{c} (1 - d^{k-1-c})^{-1} = 1$ and $b)$ $\lim_{c \to \infty} \frac{1}{c}  \sum_{k=1}^{c} (1 - d^{k-1-c})^{-2} = 1$.  The first limit $a)$ establishes the convergence of the means, and the two limits together, $a)$ and $b)$, establish the convergence of the variances.  We establish both $a)$ and $b)$ by sandwiching the sum between lower and upper bounds which are integrals of the function being summed.  In particular, if $\{f(k)\}_{k=0}^c$ is a nondecreasing sequence in $k$ then 
\begin{equation}
\label{eq:ineqforboundsum}
\int_0^c f(x) \drm x \leq \sum_{k=1}^c f(k) \leq \int_1^c f(x) \drm x + f(c),
\end{equation}
where the lower (upper) bound follows by viewing $\{f(k)\}$ as a Riemann sum of $c$ terms with unit intervals and heights at the right (left) endpoints, respectively. 
 Both $f_a(k) \equiv (1 - d^{k-1-c})^{-1}$ and $f_b(k) \equiv ( 1 - d^{k-1-c})^{-2}$ are increasing in $k \in [c]$, and so \eqref{eq:ineqforboundsum} holds for $f(k) = f_a(k)$ and $f(k) = f_b(k)$.  We now proceed to evaluate the integrals.  We first find the indefinite integral of $g_a(x) \equiv \int f_a(x) \drm x$ and $g_b(x) \equiv \int f_b(x) \drm x$ for $d > 1$ and $c \geq 1$.  Use the change of variables $z = d^{x-1-c}$ so that $\drm x = \drm z / (z \log d)$, yielding:
\begin{IEEEeqnarray}{rCl}
g_a(x(z))  & =  & \frac{1}{\log d} \int \frac{1}{z(1-z)} \drm z, \nonumber \\
g_b(x(z))  & =  & \frac{1}{\log d} \int \frac{1}{z(1-z)^2} \drm z
\end{IEEEeqnarray}
for $x(z) \equiv c+1+ \log z / \log d$.  Use partial fraction expansion to get
\begin{IEEEeqnarray}{rCl}
g_a(x(z))  & =  & \frac{1}{\log d} \int \left( \frac{1}{z} + \frac{1}{1-z} \right)\drm z,\nonumber \\
g_b(x(z))  & =  & \frac{1}{\log d} \int \left( \frac{1}{z} + \frac{1}{1-z} + \frac{1}{(1-z)^2} \right)\drm z
\end{IEEEeqnarray}
which integrates to
\begin{IEEEeqnarray}{rCl}
g_a(x(z))  & =  & \frac{\log z - \log (1-z)}{\log d}, \nonumber \\
g_b(x(z))  & =  & \frac{\log z - \log(1-z) + \frac{1}{1-z}}{\log d}.
\end{IEEEeqnarray}
Substituting back $z = d^{x-1-c}$ yields:
\begin{IEEEeqnarray}{rCl}
g_a(x)  & =  & \frac{\log d^{x-1-c} - \log (1-d^{x-1-c})}{\log d}, \nonumber \\
g_b(x)  & =  &  \frac{\log d^{x-1-c} - \log(1-d^{x-1-c}) + \frac{1}{1-d^{x-1-c}}}{\log d}.\IEEEeqnarraynumspace
\end{IEEEeqnarray}
We next integrate from $0$ to $c$ and from $1$ to $c$ respectively, and simplify:
\begin{IEEEeqnarray}{rCl}
\left. g_a(x) \right|_0^c &=& \frac{\log (d^{c+1}-1) - \log (d-1)}{\log d}  \nonumber \\
\left. g_a(x) \right|_1^c &=& \frac{\log (d^c-1) - \log(d-1)}{\log d} \nonumber \\
\left. g_b(x) \right|_0^c&=& \frac{\log (d^{c+1}-1) - \log (d-1) + \frac{d(d^c-1)}{(d-1)(d^{c+1}-1)}}{\log d} \nonumber \\
\left. g_b(x) \right|_1^c &=& \frac{\log (d^c-1) - \log(d-1) + \frac{d(d^{c-1}-1)}{(d-1)(d^c-1)}}{\log d}.\IEEEeqnarraynumspace
\end{IEEEeqnarray}
We are interested in the limits of these integrals divided by $c$ as $c \to \infty$.  Observe the last term in the numerators of the latter two expressions is dominated by the first term as $c \to \infty$ and may be ignored.  It is clear both the numerator and denominator diverge, therefore we differentiate both and apply L'H\^opital's rule.
This yields 
\begin{IEEEeqnarray}{rCl}
\lim_{c \to \infty} \left. \frac{1}{c} g_a(x) \right|_0^c = \lim_{c \to \infty} \left. \frac{1}{c} g_a(x) \right|_1^c = 1,  \nonumber \\
\lim_{c \to \infty} \left. \frac{1}{c} g_b(x) \right|_0^c = \lim_{c \to \infty} \left. \frac{1}{c} g_b(x) \right|_1^c = 1. 
\end{IEEEeqnarray}

\end{IEEEproof}

\subsection{Proofs from \S\ref{ssec:mono}}
\label{app:monoPf}
Assume A2: State-independent receptions, heterogeneous receivers.  The proof below holds for an arbitrary $n \times c$ matrix of discrete RVs $\Zbf = (Z_{j,k})$ such that each row $(Z_{j,k}, k \in [c])$ holds entries iid in $k$, and the rows are independent of one another.  We begin with a lemma that holds for a column-invariant row selection rule $\chi$.

The following lemma will be used in the proof of Prop.\ \ref{prop:monotonicityHomogeneous}.

\begin{lemma}
\label{lem:randind}
Let $\Zbf$ be an $n \times c$ random matrix with entries $Z_{j,k}$ iid in $k$ and independent in $j$.  Let $J = \chi(\mathbf{Z}) \in [n]$ be the index selected by any column invariant row-index selection rule $\chi$.  Then the RVs $(Z_{J,k}, k \in [c])$ in row $J$ are identically distributed.
\end{lemma}
\begin{IEEEproof} 
Fix column indices $k,k' \in [c]$ and some $z \in \Zmc$, the support of the RVs comprising $\Zbf$.  Then:
\begin{IEEEeqnarray}{rCl}
\Pbb(Z_{J,k} = z)
& = & \sum_{j=1}^{n} \Pbb(Z_{J, k} = z, J=j)  \nonumber \\
& = & \sum_{j=1}^{n} \Pbb(Z_{j, k} = z, J=j)  \nonumber \\ 
& = & \sum_{j=1}^{n} \Pbb(J = j | Z_{j, k}=z) \Pbb(Z_{j, k} = z)  \nonumber \\ 
& = & \sum_{j=1}^{n} \Pbb(J = j | Z_{j, k'}=z) \Pbb(Z_{j, k'} = z)  \nonumber \\ 
&=& \Pbb(Z_{J,k'}=z),
\end{IEEEeqnarray}
where we have used the facts that $i)$ $\chi$ is column invariant, i.e., $\Pbb(J = j | Z_{j, k}=z) = \Pbb(J = j | Z_{j, k'}=z)$, and $ii)$ each row of $\Zbf$ has entries iid in $k$, i.e., $\Pbb(Z_{j, k} = z) = \Pbb(Z_{j, k'} = z)$.  
\end{IEEEproof}

\begin{IEEEproof} [Proof of Prop.\ \ref{prop:monotonicityHomogeneous} (monotonicity for infinite field size)]
Let $\Zbf^{(c)}$ and $\Zbf^{(c+1)}$ be $n \times c$ and $n \times (c+1)$ matrices, respectively.  The entries $Z_{j,k}^{(c)}$ in $\Zbf^{(c)}$ are assumed iid in $k$ and independent in $j$, as are the entries $Z_{j,k}^{(c+1)}$ in $\Zbf^{(c+1)}$.  Furthermore, for each $j,j' \in [n]$ and each $k \in [c]$, $k' \in [c+1]$, the RVs $Z_{j,k}^{(c)}$ and $Z_{j',k'}^{(c+1)}$ are assumed independent.  Let $\Zbf^{(c+1)}_{1:c}$ be the $n \times c$ submatrix of $\Zbf^{(c+1)}$ obtained by removing the entries in column $c+1$.  The above assumptions assert $\Zbf^{(c)} \stackrel{\drm}{=} \Zbf^{(c+1)}_{1:c}$.  As equality in distribution is preserved under any measurable function, say $f$, it follows that $f(\Zbf^{(c)}) \stackrel{\drm}{=} f(\Zbf^{(c+1)}_{1:c})$ for all such functions.  Finally, equality in distribution ensures equality in expectation, thus $\Ebb[f(\Zbf^{(c)})] = \Ebb[f(\Zbf^{(c+1)}_{1:c})]$.  This will be used in step $(b)$ of \eqref{eq:monproofInfFieldsize} below.

Let $\hat{\chi}$ be the column invariant row-index selection rule that selects the minimum row-index among the row sum maximizing indices, and let $\hat{J}^{(c+1)} = \hat{\chi}(\Zbf^{(c+1)})$ be the index selected under $\hat{\chi}$ for $\Zbf^{(c+1)}$.  Then:
\begin{IEEEeqnarray}{rCl}
\label{eq:monproofInfFieldsize}
\Delta^{(c)} 
& \stackrel{(a)}{=} & \Ebb\left[ \frac{1}{c} \max_{j \in [n]} \sum_{k=1}^{c} Z_{j,k}^{(c)} \right] - \Ebb\left[ \frac{1}{c+1} \max_{j \in [n]} \sum_{k=1}^{c+1} Z_{j,k}^{(c+1)} \right] \nonumber \\
& \stackrel{(b)}{=} & \Ebb\left[ \frac{1}{c} \max_{j \in [n]} \sum_{k=1}^{c} Z_{j,k}^{(c+1)} \right] - \Ebb\left[ \frac{1}{c+1} \max_{j \in [n]} \sum_{k=1}^{c+1} Z_{j,k}^{(c+1)} \right] \nonumber \\
& \stackrel{(c)}{=} & \Ebb\left[ \frac{1}{c} \max_{j \in [n]} \sum_{k=1}^{c} Z_{j,k}^{(c+1)} \right] - \Ebb \left[ \frac{1}{c+1} \sum_{k=1}^{c+1} Z_{\hat{J}^{(c+1)},k}^{(c+1)} \right] \nonumber \\
& \stackrel{(d)}{\geq} & \Ebb\left[ \frac{1}{c} \sum_{k=1}^{c} Z_{\hat{J}^{(c+1)},k}^{(c+1)} \right] - \Ebb \left[ \frac{1}{c+1} \sum_{k=1}^{c+1} Z_{\hat{J}^{(c+1)},k}^{(c+1)} \right]  \nonumber \\
& \stackrel{(e)}{=} & 0 
\end{IEEEeqnarray}
where $(a)$ follows by linearity of expectation, $(b)$ follows since $\Ebb[f(\Zbf^{(c)})] = \Ebb[f(\Zbf^{(c+1)}_{1:c})]$ for any measurable $f$, $(c)$ is by definition of the index $\hat{J}^{(c+1)}$, $(d)$ follows since the maximum of a set is no smaller than any of its elements, and $(e)$ follows by applying linearity of expectation and Lemma \ref{lem:randind}.  This establishes $\Delta^{(c)} \geq 0$.  Lemma \ref{lem:tuwerty} below establishes that the inequality in \eqref{eq:monproofInfFieldsize} is in fact strict, and thus $\Delta^{(c)} > 0$.  
\end{IEEEproof}

\begin{lemma}
\label{lem:tuwerty}
Let $\Zbf^{(c+1)}$ and $\hat{J}^{(c+1)}$ be as in the proof of Prop.\ \ref{prop:monotonicityHomogeneous}.  Then:
\begin{equation}
\Ebb\left[ \max_{j \in [n]} \sum_{k=1}^{c} Z_{j,k}^{(c+1)} \right] > \Ebb\left[ \sum_{k=1}^{c} Z_{\hat{J}^{(c+1)},k}^{(c+1)} \right].
\end{equation}
\end{lemma}

\begin{IEEEproof}
Observe
\begin{IEEEeqnarray}{rCl}
\IEEEeqnarraymulticol{3}{l}
{\max_{j \in [n]} \sum_{k=1}^c Z_{j,k}^{(c+1)} 
}\nonumber \\* \quad
&  = & \max \left( \max_{j \in [n] \setminus \{\hat{J}^{(c+1)}\}} \sum_{k=1}^{c}  Z_{j,k}^{(c+1)}, ~\sum_{k=1}^c Z_{\hat{J}^{(c+1)},k}^{(c+1)} \right),\IEEEeqnarraynumspace
\end{IEEEeqnarray}
which we henceforth abbreviate as $Z_1 = \max(Z_2,Z_3)$. We further denote by $A_{3\geq}$ the event $\{Z_{2} \leq Z_{3} \}$, and $A_{2>}$ the complement event $\{Z_{2} > Z_{3} \}$. The lemma is equivalent to the assertion $\Ebb[Z_1] > \Ebb[Z_3]$, for which we need to show the following strict inequality, where the first (last) equality follows from (reversely) applying the total expectation theorem.
\begin{IEEEeqnarray}{rCl}
\Ebb[Z_1] 
& = & \Ebb[Z_1 \vert A_{3\geq}] \Pbb(A_{3\geq}) + \Ebb[Z_1 \vert A_{2>}] \Pbb(A_{2>}) \nonumber \\
& = &  \Ebb[Z_3 \vert A_{3\geq}] \Pbb(A_{3\geq}) + \Ebb[Z_2 \vert A_{2>}] \Pbb(A_{2>}) \nonumber \\
& > & \Ebb[Z_3 \vert A_{3\geq}] \Pbb(A_{3\geq}) +  \Ebb[Z_3 \vert A_{2>}] \Pbb(A_{2>}) \nonumber \\
& = & \Ebb[Z_{3}].
\end{IEEEeqnarray}
The inequality in the above equation is equivalent to $\Ebb[Z_{2} - Z_{3} \vert Z_{2} > Z_{3}] \Pbb(Z_{2} > Z_{3}) > 0$, for which we need to show $\Pbb(Z_{2} > Z_{3}) > 0$.

Clearly $\Pbb(Z_{2} > Z_{3}) > 0$ has to hold, for otherwise $Z_{1} = \max(Z_{2}, Z_{3}) = Z_{3}$ always holds, that is
\begin{eqnarray}
\max_{j \in [n]} \sum_{k=1}^c Z_{j,k}^{(c+1)}  = \sum_{k=1}^c Z_{\hat{J}^{(c+1)},k}^{(c+1)}.
\end{eqnarray}
But this can not be true. Recall $\hat{J}^{(c+1)}$ is defined to be the minimum row index for which the sum of all the $c+1$ columns is the maximum. This index can not be \textit{always} a row index for which the sum of the first $c$ columns is the maximum.

\end{IEEEproof}

\begin{IEEEproof}[Proof of Prop.\ \ref{prop:monotonicityHeterogeneous} (monotonicity for finite field size)]
Assume A1: State-dependent receptions, heterogeneous receivers.  Note $\Delta^{(c)} > 0 \Leftrightarrow c(c+1) \Delta^{(c)} > 0$.  Fix the field size $d$ to be a finite integer $\geq 2$, fix the blocklength $c \in \Nbb$, and fix the number of receivers $n \in \Nbb$.  Let $\Xbf^{(c)},\Xbf^{(c+1)}$ be $n \times c$ and $n \times (c+1)$ random matrices with entries $X_{j,k}^{(c)} \sim \mathrm{Geo}(q_{j,k}^{(c)})$ and $X_{j,k}^{(c+1)}\sim \mathrm{Geo}(q_{j,k}^{(c+1)})$, respectively.  Then:
\begin{equation} 
c(c+1) \Delta^{(c)} \!=\! \Ebb\left[ \left(c+1\right) \max_{j \in [n]} \sum_{k=1}^c X_{j,k}^{(c)} - c \max_{j \in [n]} \sum_{k=1}^{c+1} X_{j,k}^{(c+1)} \right].
\end{equation}
By the standing model assumptions, the entries $X_{j,k}^{(c)}$ in $\Xbf^{(c)}$ are independent in $j$ and $k$, as are the entries $X_{j,k}^{(c+1)}$ in $\Xbf^{(c+1)}$.  Furthermore, for each $j,j' \in [n]$ and each $k \in [c]$, $k' \in [c+1]$, the RVs $X_{j,k}^{(c)}$ and $X_{j',k'}^{(c+1)}$ are assumed independent.  Recall $q_{j,k}^{(c)} = (1-d^{k-1-c})q_j$.  It is immediate that $X_{j,k}^{(c)} \stackrel{\drm}{=} X_{j,k+1}^{(c+1)}$.  Let $\Xbf^{(c+1)}_{2:c+1}$ be the submatrix obtained by removing the entries in column $1$ of $\Xbf^{(c+1)}$.  The above equality in distribution then gives $\Xbf^{(c)} \stackrel{\drm}{=} \Xbf^{(c+1)}_{2:c+1}$.  Further, for any measurable function $f$, $\Ebb[f(\Xbf^{(c)})] = \Ebb[f(\Xbf^{(c+1)}_{2:c+1})]$.  It follows that
\begin{IEEEeqnarray}{rCl}
\IEEEeqnarraymulticol{3}{l}
{c(c+1) \Delta^{(c)} 
}\nonumber \\* \quad
& = & \Ebb\left[ \left(c+1\right) \max_{j \in [n]} \sum_{k=1}^c X_{j,k+1}^{(c+1)} - c \max_{j \in [n]} \sum_{k=1}^{c+1} X_{j,k}^{(c+1)} \right]. \IEEEeqnarraynumspace
\end{IEEEeqnarray}
Now that $c(c+1) \Delta^{(c)}$ has been expressed solely in terms of $\Xbf^{(c+1)}$, we henceforth simplify our notation, writing $\Xbf \equiv \Xbf^{(c+1)}$.  Let $Y_j \equiv Y_j^{(c+1)}$ for $j \in [n]$ be the sums of each of the rows of $\Xbf$.  Let $\hat{J} \in [n]$ be the random row index $\hat{J} = \hat{\chi}(\Xbf)$, where $\hat{\chi}$ is the column invariant row selection rule that selects the smallest index in the (random) set of indices in $[n]$ that maximize $Y_j$ over $j \in [n]$.  By this definition:
\begin{equation}
c(c+1) \Delta^{(c)} = \Ebb\left[ \left(c+1\right) \max_{j \in [n]} \sum_{k=1}^{c} X_{j,k+1} - c \sum_{k=1}^{c+1} X_{\hat{J},k} \right].  
\end{equation}
Further, since the maximum of any set is no smaller than any of its elements, we have the lower bound
\begin{equation}
c(c+1) \Delta^{(c)} \geq \Ebb\left[ \left(c+1\right) \sum_{k=1}^{c} X_{\hat{J},k+1} - c \sum_{k=1}^{c+1} X_{\hat{J},k} \right].
\end{equation}
The right hand side may be rearranged as
\begin{equation}
\Ebb\left[ \sum_{k=1}^{c} X_{\hat{J},k+1} - c X_{\hat{J},1} \right] 
= \sum_{k=1 }^{c} \Ebb\left[ X_{\hat{J},k+1} - X_{\hat{J},1} \right].
\end{equation}
Suppose the following stochastic ordering condition is true:
\begin{equation}
\label{eq:stochordercondition}
X_{\hat{J},k+1} >_{\rm st} X_{\hat{J},1}, ~~  k \in [c].
\end{equation}
This immediately implies $\Delta^{(c)} > 0$, completing the proof.  

It remains to prove \eqref{eq:stochordercondition}.  It suffices to show 
\begin{equation}
\label{eq:ryryry}
\Pbb(X_{\hat{J},k} > x) > \Pbb(X_{\hat{J},k'} > x)
\end{equation}
for $k,k' \in [c+1]$ with $k > k'$ and $x \in \Nbb$.  We start by using the total probability theorem:
\begin{IEEEeqnarray}{rCl}
\Pbb(X_{\hat{J},k} \!>\! x) \! & = & \sum_{a} \sum_{j} \Pbb(X_{\hat{J},k} > x, X_{\hat{J},k} + X_{\hat{J},k'} = a, \hat{J} = j) \nonumber \\
& = & \sum_{a} \sum_{j} \Pbb(X_{j,k} > x, X_{j,k} + X_{j,k'} = a, \hat{J} = j) \nonumber \\
& = & \sum_{a} \sum_{j} \Pbb(\hat{J} = j \vert X_{j,k} > x, X_{j,k} + X_{j,k'}= a)  \nonumber \\
& & \qquad \quad \negmedspace{} \cdot \Pbb(X_{j,k} > x, X_{j,k} + X_{j,k'} = a).
\end{IEEEeqnarray}
Similarly, 
\begin{IEEEeqnarray}{rCl}
\Pbb(X_{\hat{J},k'} \!>\! x) \! & = & \sum_{a} \sum_{j} \Pbb(\hat{J} = j \vert X_{j,k'}> x, X_{j,k} + X_{j,k'} = a) \nonumber \\
& & \qquad \quad \negmedspace{} \cdot \Pbb(X_{j,k'} > x, X_{j,k} + X_{j,k'} = a). 
\end{IEEEeqnarray}
We will establish \eqref{eq:ryryry} by showing 
\begin{IEEEeqnarray}{rCl}
& & \Pbb(\hat{J} = j \vert X_{j,k} > x, X_{j,k} + X_{j,k'} = a)  \nonumber \\
&=& \Pbb(\hat{J} = j \vert X_{j,k'} > x, X_{j,k} + X_{j,k'} = a) \label{eq:uououo}
\end{IEEEeqnarray}
and
\begin{IEEEeqnarray}{rCl}
& & \Pbb(X_{j,k} > x, X_{j,k} + X_{j,k'} = a)   \nonumber \\
&>& \Pbb(X_{j,k'} > x, X_{j,k} + X_{j,k'} = a) \label{eq:wrwrwr}
\end{IEEEeqnarray}
for all $x$, $a$, and $j$.   Lemma \ref{lem:markovchainselectionrule} easily yields \eqref{eq:uououo}, while Lemma \ref{lem:StochasticOrderingOfGeoRVs} establishes  \eqref{eq:wrwrwr}.  Both Lemmas are stated and proved below.
\end{IEEEproof}

\begin{lemma}
\label{lem:markovchainselectionrule}
For all $j \in [n]$ and distinct $k,k' \in [c+1]$, the RVs $(\hat{J},X_{j,k}+X_{j,k'},X_{j,k},X_{j,k'})$ form Markov chains
\begin{equation}
\hat{J} - (X_{j,k}+X_{j,k'}) - X_{j,k}, ~~ \hat{J} - (X_{j,k}+X_{j,k'}) - X_{j,k'}.
\end{equation}
\end{lemma}

\begin{IEEEproof}
We establish the first chain; the proof of the second is the same.  Fix a row index $j \in [n]$, and two distinct column indices $k,k' \in [c+1]$.  By the definition of a Markov chain, we must show
\begin{IEEEeqnarray}{rCl}
\label{eq:markcondtest}
& & \Pbb\left( \hat{J} = i ~ \vert X_{j,k} + X_{j,k'} = a, X_{j,k} = x \right) \nonumber \\
& = & \Pbb\left( \hat{J} = i ~ \vert X_{j,k} + X_{j,k'} = a \right)
\end{IEEEeqnarray}
for all $i \in [n]$ (including $i = j$), integer $a \geq 2$, and integer $x \geq 1$. To make our expressions more compact, we introduce the notation for events $A^{(a)}_1 = \{X_{j,k} + X_{j,k'} = a\}$, $A_2 = \{X_{j,k} = x\}$, and $A^{(a)}_{1,2} = A^{(a)}_1 \cap A_2$.  Then, showing \eqref{eq:markcondtest} is equivalent to showing $\Pbb\left( \hat{J} = i ~ \vert A^{(a)} \right)$ is the same for $A^{(a)}$ equal to either $A^{(a)}_{1,2}$ or $A^{(a)}_1$.  Let $\Nbb_{>c} \equiv \{c+1,c+2,\ldots\}$, and define the sets
\begin{equation}
\label{eq:rowsumcond}
\Ymc_i = \left\{ \ybf \! \in \! \Nbb_{>c}^n : 
\begin{cases}
    y_i \geq y_l \!&\!\! \text{if } l \in \{i+1,\ldots,n\},
    \\
    y_i > y_l \!&\!\! \text{if } l \in \{1,\ldots,i-1\}.
  \end{cases}
\right\}, ~ i \in [n].
\end{equation}
Observe that $\{\hat{J} = i\} = \{ \Ybf \in \Ymc_i\}$ for each $i \in [n]$, where $\Ybf = (Y_1,\ldots,Y_n)$ are the row-sums of $\Xbf$.  In words, row $i$ is selected under the row selection rule $\hat{\chi}$ as the smallest index among the row-sum maximizing indices precisely when the row-sums $\Ybf$ lie in $\Ymc_i$ \eqref{eq:rowsumcond}.  Recall we use the phrase ``row-sum'' to indicate the sum of all the elements in a given row, rather than the sum of the rows.  Conditioning on all possible row-sums using the total probability theorem gives:
\begin{IEEEeqnarray}{rCl}
\IEEEeqnarraymulticol{3}{l}
{\Pbb(\hat{J}=i|A^{(a)}) 
}\nonumber \\* \quad \!
&=& \sum_{\ybf \in \Nbb_{>c}^n} \Pbb(\hat{J}=i|\Ybf=\ybf,A^{(a)})\Pbb(\Ybf=\ybf|A^{(a)}) \nonumber \\
&=& \sum_{\ybf \in \Ymc_i \cap \{\ybf : y_j \geq a + c - 1 \}} \Pbb(\hat{J}=i|\Ybf=\ybf,A^{(a)})\Pbb(\Ybf=\ybf|A^{(a)}) \nonumber \\
&=& \sum_{\ybf \in \Ymc_i \cap \{\ybf : y_j \geq a + c - 1 \}} \Pbb(\Ybf=\ybf|A^{(a)})
\end{IEEEeqnarray}
The penultimate equality holds because $\Pbb(\hat{J}=i|\Ybf=\ybf) = 0$ for $\ybf \not\in\Ymc_i$ and $\Pbb(\Ybf = \ybf|A^{(a)}) = 0$ for $\ybf$ with $y_j < a+c-1$, and the last equality holds because the function $\hat{J}$ of $\Ybf$ takes value $i$ over all $\ybf \in \Ymc_i$.  We now focus on $\Pbb(\Ybf=\ybf|A^{(a)})$:
\begin{IEEEeqnarray}{rCl}
\IEEEeqnarraymulticol{3}{l}
{\Pbb(\Ybf=\ybf|A^{(a)})
}\nonumber \\* \quad
&=& \Pbb(Y_j = y_j | A^{(a)}) \prod_{l \neq j} \Pbb(Y_l = y_l | A^{(a)}) \nonumber \\
&=& \Pbb(Y_j = y_j | A^{(a)}) \prod_{l \neq j} \Pbb(Y_l = y_l) \nonumber \\
&=& \Pbb\left( \left. (X_{j,k}+X_{j,k'}) + \sum_{s \in [c+1] \setminus \{k,k'\}} X_{j,s} = y_j \right| A^{(a)} \right) \nonumber \\
& & \negmedspace{} \cdot \prod_{l \neq j} \Pbb(Y_l = y_l).  
\end{IEEEeqnarray}
The first two equalities hold because of the assumed independence of the rows (note $A^{(a)}$ is specific to row $j$), and the last equality is by definition of $Y_j$.  Regardless of whether $A^{(a)} = A^{(a)}_1$ or $A^{(a)} = A^{(a)}_{1,2}$, we may write
\begin{IEEEeqnarray}{rCl}
\IEEEeqnarraymulticol{3}{l}
{\Pbb(\Ybf=\ybf|A^{(a)}) 
}\nonumber \\* \quad
& = & \Pbb\left(\sum_{s \in [c+1] \setminus \{k,k'\}} X_{j,s}  = y_j -a \right) \prod_{l \neq j} \Pbb(Y_l = y_l)\IEEEeqnarraynumspace
\end{IEEEeqnarray}
due to the assumed independence of the columns.  Specifically, $(X_{j,s}, s \in [c+1] \setminus \{k,k'\})$ is independent of $(X_{j,k},X_{j,k'})$.  As the above RHS clearly depends upon $a$ but not $x$, it follows that $\Pbb(\hat{J}=i|A^{(a)}_1) = \Pbb(\hat{J}=i|A^{(a)}_{1,2})$ for each $i \in [n]$, i.e., \eqref{eq:markcondtest} holds.  
\end{IEEEproof}

\begin{lemma} 
\label{lem:StochasticOrderingOfGeoRVs}
Consider independent geometric RVs $X_{1} \sim \mathrm{Geo}(q_{1})$, $X_{2} \sim \mathrm{Geo}(q_{2})$, where $q_{1} \in (0, 1], q_{2} \in [0, 1), q_{1} > q_{2}$ (so $\Ebb[X_{1}] < \Ebb[X_{2}]$).  Then $\delta(x,a) > 0$ for all integers $x>0$ and $a > x+1$, where  
\begin{equation}
\delta(x,a) \equiv \Pbb(X_{2} > x, X_{1} + X_{2} = a) - \Pbb(X_{1} > x, X_{1} + X_{2} = a).
\end{equation}
\end{lemma}

\begin{IEEEproof}
Define $\bar{q}_1 = 1 - q_1$ and $\bar{q}_2 = 1 - q_2$.  Then:
\begin{IEEEeqnarray} {rCl}
& & \delta(x,a) \nonumber \\
& = & \sum_{y=x+1}^{a-1}\Pbb(X_{2}=y) \Pbb(X_{1}= a-y) \nonumber \\
& & \negmedspace{} - \sum_{y=x+1}^{a-1}\Pbb(X_{1}=y) \Pbb(X_{2}= a-y) \nonumber \\ 
& = & \sum_{y=x+1}^{a-1} \bar{q}_2^{y-1} q_{2} \bar{q}_1^{a-y-1} q_{1} - \sum_{y=x+1}^{a-1} \bar{q}_1^{y-1} q_{1} \bar{q}_2^{a-y-1} q_{2} \nonumber \\ 
& = & \frac{q_{1}q_{2}\bar{q}_2^{a}}{\bar{q}_1 \bar{q}_2} \sum_{y=x+1}^{a-1} \left(\frac{\bar{q}_1}{\bar{q}_2} \right)^{a-y} - \frac{q_{1}q_{2}\bar{q_{2}}^{a}}{\bar{q}_1 \bar{q}_2} \sum_{y=x+1}^{a-1} \left(\frac{\bar{q}_1}{\bar{q}_2} \right)^{y} \nonumber \\ 
& = & \frac{q_{1}q_{2}\bar{q}_2^{a}}{\bar{q}_1 \bar{q}_2} \sum_{y'=1}^{a-x-1} \left(\frac{\bar{q}_1}{\bar{q}_2} \right)^{y'} - \frac{q_{1}q_{2}\bar{q}_2^{a}}{\bar{q}_1 \bar{q}_2} \sum_{y'=1}^{a-x-1} \left( \frac{\bar{q}_1}{\bar{q}_2} \right)^{y'} \left(\frac{\bar{q}_1}{\bar{q}_2} \right)^{x} \nonumber \\ 
& = & \left(1 - \left(\frac{\bar{q}_1}{\bar{q}_2} \right)^{x} \right) \frac{q_{1}q_{2}\bar{q}_2^{a}}{\bar{q}_1 \bar{q}_2} \sum_{y'=1}^{a-x-1} \left(\frac{\bar{q}_1}{\bar{q}_2} \right)^{y'} > 0.
\end{IEEEeqnarray}
\end{IEEEproof}

\begin{IEEEproof}[Proof of Prop.\ \ref{prop:monoSamplePathCombined} (sample path monotonicity)]
{\bf Case $i)$} when $m \leq c'$ meaning the workload is no larger than the larger blocklength $c'$. We consider two scenarios: $1)$ $m \leq c$, and $2)$ $c < m \leq c'$.  In the first scenario, the increase in the blocklength from $c$ to $c'$ will have no effect on the completion time of the workload, since a single block suffices for both blocklengths, thus $T_{n:n}^{(c',m)}(\omega) = T_{n:n}^{(c,m)}(\omega)$ for all $\omega \in \Omega$.  In the second scenario, $c < m \leq c'$ means under blocklength $c$ there exists a positive integer $k$ and a nonnegative integer $l$ such that $m = k c + l$, i.e., $k+1$ (or $k$, if $l = 0$) blocks are required, with the first $k$ blocks having length $c$ and the last i.e., $(k+1)^{\rm th}$ block having length $l$.  At most one block is required to handle the workload of $m$ packets under blocklength $c'$.  Let $y^{(c)}_{n:n,i}$ be the duration of the $i^{\rm th}$ block for $i \in [k]$ and $y^{(l)}_{n:n,k+1}$ be the durations of the $(k+1)^{\rm th}$ (partial) block under blocklength $c$, and 
\begin{equation}
t^{(c,m)}_{n:n} = T_{n:n}^{(c,m)}(\omega) = \sum_{i=1}^k y^{(c)}_{n:n,i} + y^{(l)}_{n:n,k+1}
\end{equation}
the time at which the workload is completed under blocklength $c$.  Analogously, $t^{(c',m)}_{n:n} = T_{n:n}^{(c',m)}(\omega) = y^{(c')}_{n:n,1}$, is the completion time under $c'$.   We must show $t^{(c',m)}_{n:n} \leq t^{(c,m)}_{n:n}$ for all $\omega$.  Let $r_j(t)$ be the number of receptions by receiver $j$ by time $t$ and observe $r_j(t^{(c,m)}_{n:n}) \geq m$ for all $j \in [n]$.  Let $t^{(c',m)}_j = \min\{ t : r_j(t) = m\}$, and thus $t^{(c',m)}_j \leq t^{(c,m)}_{n:n}$, and $t^{(c',m)}_{n:n} = \max_{j \in [n]} t^{(c',m)}_j \leq t^{(c,m)}_{n:n}$.

{\bf Case $ii)$}  when $m > c'$ meaning the workload exceeds the larger blocklength $c'$. Let $T = T_{n:n}^{(c,m)}$ be the (random) completion time to transmit all $m$ packets with blocklength $c$, and let $T' = T_{n:n}^{(c',m)}$ be the time with blocklength $c'$.  Prop.\ \ref{prop:monoSamplePathCombined} (case $ii)$) may be restated: 
\begin{equation}
T'(\omega) \leq T(\omega), ~ \forall \omega \in \Omega ~~ \Leftrightarrow ~~ c' = kc, ~ k \in \{2,3,4,\ldots\}.
\end{equation}
We first prove the forward direction (necessity): $T'(\omega) \leq T(\omega), ~ \forall \omega \in \Omega \Rightarrow c' = kc$ for integer $k \geq 2$, which is equivalent to: if $c' = k c +l$ with integer $k \geq 1$ and integer $l \in \{1,\ldots,c-1\}$ then there exists $\omega : T'(\omega) > T(\omega)$.  We begin with perhaps the simplest example showing that increasing the blocklength may increase the completion time for some sample paths.  In particular, consider $n=2$ receivers and a workload of $m=4$ packets which are transmitted using one of two schemes: $i)$ a blocklength of $c=2$ for a total of $2$ full blocks, and $ii)$ a blocklength of $c'=3$ for a total of $1$ full and $1$ partial block, as shown in Fig.\ \ref{fig:samplePathBasicCounterexample}. The realization $\omega$ for the erasure channels is illustrated at the top of the figure for the first $6$ time slots, where $s$ ($f$) indicates a successful (failed) transmission for that receiver in that time slot.  For this realization, the longer blocklength construction finishes after the shorter blocklength construction.  The times $t_a,t_b,t_c,t_d,t_e$ are defined in the subsequent generalization of this example.  
\begin{figure}[!htbp]
\centering
\includegraphics[width=0.3\textwidth]{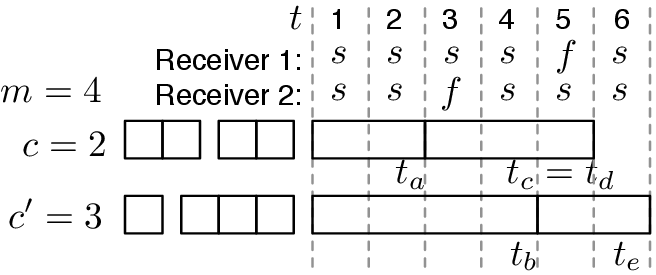}
\caption{Simple example illustrating that increasing the blocklength may increase the completion time for some sample paths.}
\label{fig:samplePathBasicCounterexample}
\end{figure}

We now generalize the above example; see Fig.\ \ref{fig:samplePathCounterExample}.  Let there be $n=2$ receivers.  By assumption, the larger blocklength is $c' = k c + l$ for integer $k \geq 1$ and integer $l \in \{1,\ldots,c-1\}$.  The blocks under both blocking constructions are shown as rectangles in the figure, and the numbers inside the blocks are the block indices.  Give each receiver a counter initialized at zero, and increment each receiver's counter upon a \textit{useful} reception at that receiver, where a reception is useful if that receiver has not yet received a blocklength of receptions within the current block.  These quantities are shown for the two receivers under the two blocking constructions for the various intervals of time shown in the figure.  The realization $\omega$ of the erasure channels for each receiver is shown at the top of the figure, with $s$ ($f$) indicating success (failure), respectively.  Under this realization, the completion time of block index $k$ under blocklength $c$ is time $t_a = k c$.  As of time $t_a$, the first block under blocklength $c'$ has received $k c$ out of the $kc+l$ receptions needed  to complete the block.  Let time $t_b = t_a + c$ be the time of completion of the first block under blocklength $c'$.  Next, let time $t_c = t_b + c-l$ be the time of completion of block $k+1$ under blocklength $c$.  Let $t_d = m + c-l$ be the completion time of the workload $m$ under blocklength $c$, and $t_e = t_d + c-l$ be the workload completion time under blocklength $c'$.  The key insight is this: the non-overlapping $c-l$ failed receptions for each of the two receivers occur in the same block (block $k+1$) under blocklength $c$ (and therefore the duration of the block is delayed by $c-l$); while under blocklength $c'$ the failed receptions for the two receivers occur in different blocks, which delays both of the first two blocks under blocklength $c'$ by $c-l$.  

\begin{figure*}[!htbp]
\centering
\includegraphics[width=0.775\textwidth]{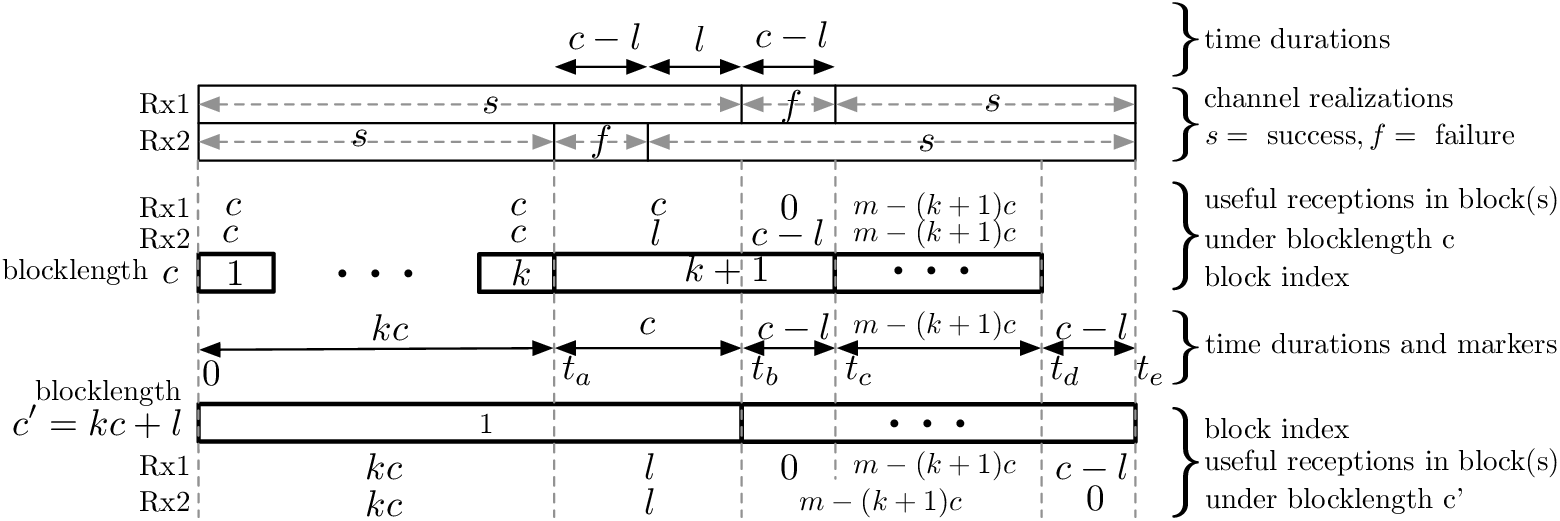}
\caption{For any blocklengths $c,c'$ with $c' > c$ and $c' = kc + l$ with $l \in \{1,\ldots,c-1\}$ and workload $m > c'$, there exists a realization under which the longer blocklength construction will have a longer completion time.}
\label{fig:samplePathCounterExample}
\end{figure*}

We next prove the reverse direction (sufficiency): if $c' = kc$ for some integer $k \geq 2$, then $T'(\omega) \leq T(\omega) ~\forall \omega \in \Omega$.  See Fig.\ \ref{fig:samplePathDoubling}.  This proof is essentially an extension of the proof of case $i)$.  Fix the realization $\omega$; in what follows we suppress the dependence upon $\omega$.  By assumption $c' = k c$ for some integer $k \geq 2$ and $m > c'$.  Write $m = k' c' + l'$ for integer $k' \geq 1$ and integer $l' \in \{0,\ldots,c'-1\}$, so that the workload $m$ requires $k'+1$ blocks under blocklength $c'$: $k$ full blocks of size $c'$ and, if $l' > 0$, a partial block of size $l'<c'$.  Let $(y^{(c')}_{n:n,1},\ldots,y^{(c')}_{n:n,k'},y^{(l')}_{n:n,k'+1})$ denote the durations of time required to complete the $k'+1$ blocks under blocklength $c'$, and $(t^{(c')}_{n:n,1},\ldots,t^{(c')}_{n:n,k'},t^{(c')}_{n:n,k'+1})$ the corresponding sequence of partial sums, so that $t^{(c')}_{n:n,\iota}$ is the completion time of block $\iota$ under blocklength $c'$, and $t^{(c')}_{n:n,k'+1}$ the completion time of the workload.  Write $l' = k^{''} c + l$ for integer $k^{''} \geq 0$ and integer $l \in \{0,\ldots,c-1\}$, so that workload $m$ requires $k'k+k^{''}+1$ blocks under blocklength $c$: $k'k+k^{''}$ blocks of size $c$ and, if $l > 0$, a partial block of size $l<c$.  Similarly, let $(y^{(c)}_{n:n,1},\ldots,y^{(c)}_{n:n,k'k+k^{''}},y^{(l)}_{n:n,k'k+k^{''}+1})$ denote the durations of time required to complete the $k'k+k^{''}+1$ blocks under blocklength $c$, and $(t^{(c)}_{n:n,1},\ldots,t^{(c)}_{n:n,k'k+k^{''}},t^{(c)}_{n:n,k'k+k^{''}+1})$ the corresponding sequence of partial sums, so that $t^{(c)}_{n:n,\iota}$ is the completion time of block $\iota$ under blocklength $c$, and $t^{(c)}_{n:n,k'k+k^{''}+1}$ the completion time of the workload.  We must show $t^{(c')}_{n:n,k'+1} \leq t^{(c)}_{n:n,k'k+k^{''}+1}$.

We first show $t^{(c')}_{n:n,\iota} \leq t^{(c)}_{n:n,k\iota}$ for each $\iota \in \{1,\ldots,k'\}$ by induction.  When $\iota = 1$,  $t^{(c')}_{n:n,1} \leq t^{(c)}_{n:n,k}$ can be verified by using the same ideas used in proving case $i)$.  Next, assuming $t^{(c')}_{n:n,\iota-1} \leq t^{(c)}_{n:n,k(\iota-1)}$ and denoting $t_{j,\iota}^{(c')} =  \min\{ t: r_{j}(t_{n:n, \iota - 1}^{(c')}, t] = c'=kc \}$ where $r_{j}(t_{1}, t_{2}]$ is the number of successful (not necessarily useful) receptions by receiver $j$ during the time interval $(t_{1}, t_{2}]$, we have:
\begin{equation}
t_{j,\iota}^{(c')} \stackrel{(a)}{\leq} \min\{ t: r_{j}(t_{n:n, k(\iota - 1)}^{(c)}, t] = c'=kc \}
\stackrel{(b)}{\leq} t_{n:n, k \iota}^{(c)},
\end{equation}
where $(a)$ is due to the induction hypothesis and $(b)$ is seen to be true by observing $r_{j} (t_{n:n, k(\iota - 1)}^{(c)}, t_{n:n, k \iota}^{(c)}] \geq kc $ for all $j \in [n]$. So $t_{n:n, \iota}^{(c')} = \max_{j} t_{j, \iota}^{(c')} \leq t_{n:n, k \iota}^{(c)}$, finishing the induction step. Finally, for the last partial block under $c'$, a similar argument (again to the one used in proving case $i)$, as effectively the remaining workload does not exceed the larger blocklength $c'$) together with the just proved result $t^{(c')}_{n:n,\iota} \leq t^{(c)}_{n:n,k\iota}$ (specialized with $\iota = k'$) gives $t_{n:n, k'+1}^{(c')} \leq t_{n:n, k'k + k'' + 1}^{(c)}$.

\begin{figure}[!htbp]
\centering
\includegraphics[width=0.5\textwidth]{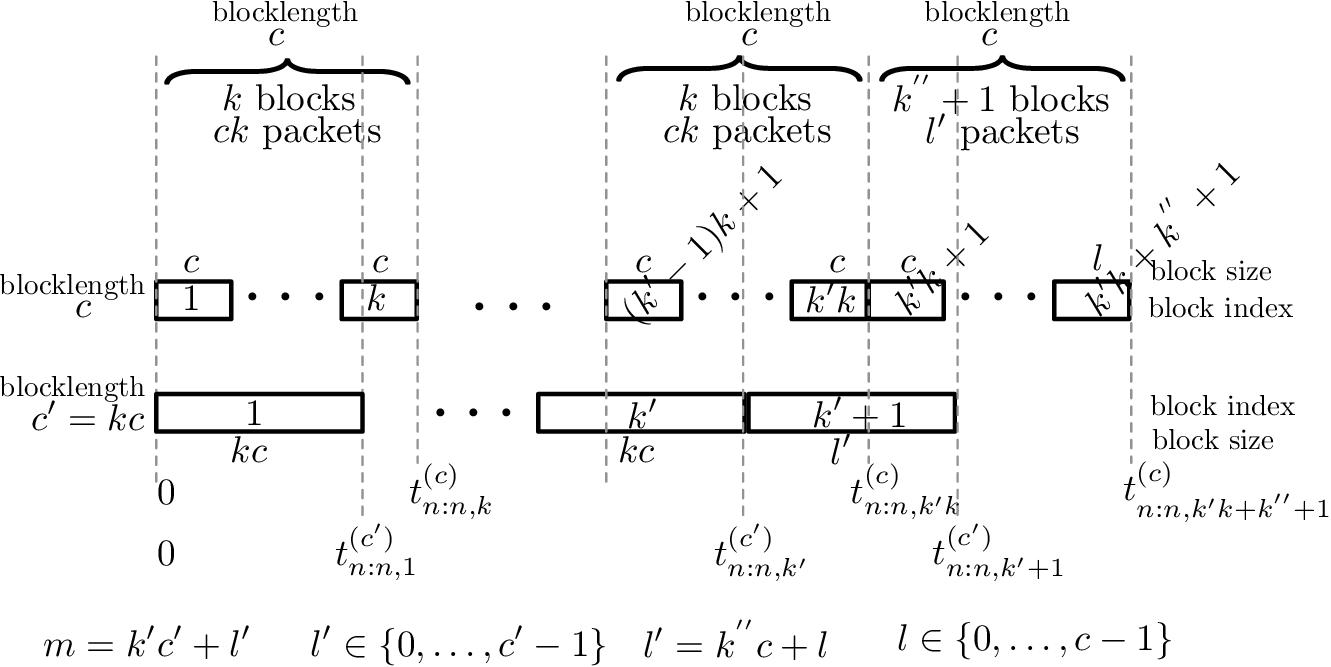}
\caption{For any blocklengths $c,c'$ with $c' > c$ and $c' = kc$ and workload $m > c'$, the longer blocklength construction will have a completion time no longer than that under the shorter blocklength construction.}
\label{fig:samplePathDoubling}
\end{figure}

\end{IEEEproof}

\subsection{Proofs from \S\ref{ssec:rossdlctightinc}}
\label{app:rossdlctightincPf}

\begin{IEEEproof}[Proof of Prop.\ \ref{prop:LBandUBinCHomoRxInftyFieldsize}]
We first show the upper bound $\tilde{u}(n,c)$ is tight. Specializing $r=1$ in the upper bound in Prop.\ \ref{prop:boundsOnEY3} (\S \ref{sec:delayRLNC}), we have
\begin{equation}
\Ebb[\tilde{Y}_{n:n}^{(c)}]  \leq s + n \left( (c-s) Q(c,s) + \frac{s^c \erm^{-s}}{\Gamma(c)} \right), ~~ s \geq 0.
\end{equation}
When $r=1$ the optimal Ross's bound is given by $c + n c \frac{(s^*)^c \erm^{-s^*}}{\Gamma(c+1)}$ where $s^{*} = Q^{-1}\left(c,\frac{1}{n}\right)$ is such that $n Q(c,s^{*}) = 1$. Yet to show Ross's bound is \textit{asymptotically} tight it suffices to work with another choice of $s$. For this we need a lower bound on $c!$ \cite{Bat2008}:
\begin{equation}
\sqrt{2 \pi c}  \left(\frac{c}{\erm} \right)^c \leq c! \leq \frac{1}{\sqrt{1-1/c}}\sqrt{2 \pi c} \left(\frac{c}{\erm} \right)^c.
\end{equation}
Then, dividing both sides of the Ross's upper bound by $c$ and applying the lower bound on the factorial gives
\begin{IEEEeqnarray}{rCl}
\IEEEeqnarraymulticol{3}{l}
{\frac{1}{c} \Ebb[\tilde{Y}_{n:n}^{(c)}] 
}\nonumber \\* \quad
& \leq & \frac{s}{c} + n \left( \left(1-\frac{s}{c}\right) Q(c,s) + \frac{s^c}{\erm^s} \frac{1}{c!} \right) \nonumber \\
& \leq & \frac{s}{c} + n \left( \left(1-\frac{s}{c}\right) Q(c,s) + \frac{s^c}{\erm^s} \frac{1}{\sqrt{2 \pi c}} \left(\frac{\erm}{c}\right)^c \right) \nonumber \\
& = & \frac{s}{c} + n \left( \left(1-\frac{s}{c}\right) Q(c,s) + \frac{1}{\sqrt{2 \pi c}} \left(\frac{s}{c}\right)^c \erm^{c\left(1-\frac{s}{c}\right)} \right).\IEEEeqnarraynumspace
\end{IEEEeqnarray}
Choosing $s = c$ achieves the desired tightness as $c \to \infty$
\begin{equation}
\frac{1}{c} \Ebb[\tilde{Y}_{n:n}^{(c)}] \leq \tilde{u}(n,c) \equiv 1 + n/\sqrt{2 \pi c} \to 1.
\end{equation}

Next, we show the lower bound $\tilde{l}(n,c)$ is tight. Specializing $r=1$ in the lower bound in Prop.\ \ref{prop:boundsOnEY3}, we have:
\begin{equation}
\label{eq:delaCalinSec4C}
\Ebb[\tilde{Y}_{n:n}^{(c)}] \geq t - (t-c) (1 - Q(c,t))^{n-1}, ~~ t \geq c 
\end{equation}
Division by $c$ and reparameterization of $t/c$ by $t$, and then by $t+1$ gives the de la Cal's bound as
\begin{IEEEeqnarray}{rCl}
\frac{1}{c} \Ebb[\tilde{Y}_{n:n}^{(c)}] 
& \geq & \frac{t}{c} - \left(\frac{t}{c}-1 \right) (1 - Q(c,t))^{n-1}, ~\frac{t}{c} \geq 1 \nonumber \\
\frac{1}{c} \Ebb[\tilde{Y}_{n:n}^{(c)}] 
& \geq & t - (t-1) (1 - Q(c,t c))^{n-1},  ~t \geq 1 \nonumber \\
\frac{1}{c} \Ebb[\tilde{Y}_{n:n}^{(c)}] 
& \geq & 1 +  t(1 - (1 - Q(c,c(1+ t)))^{n-1}),  ~t \geq 0. \IEEEeqnarraynumspace
\end{IEEEeqnarray}
Since the quantity $t(1 - (1 - Q(c,c(1+ t)))^{n-1})$ is nonnegative for $t \geq 0$, we argue that the asymptotic tightness of the Ross's bound implies the asymptotic tightness of the de la Cal's bound.  To see this, after choosing nonnegative $t$ as a function of $c$, we write de la Cal's bound as $1 + D(c)$ with $D(c) \geq 0$. Observe
\begin{equation}
1 \leq 1 + D(c) \leq \frac{1}{c} \Ebb[\tilde{Y}_{n:n}^{(c)}] \leq \tilde{u}(n,c),
\end{equation}
with $\tilde{u}(n,c) \to 1$ as $c \to \infty$, and hence it must hold $D(c) \to 0$ by the pinch lemma.

Numerical investigation suggests $t(n,c) = c + \sqrt{c \log n}$ can be used in \eqref{eq:delaCalinSec4C}, which after further normalizing by $c$ becomes:
\begin{IEEEeqnarray}{rCl}
\IEEEeqnarraymulticol{3}{l}
{\frac{1}{c} \Ebb[\tilde{Y}_{n:n}^{(c)}]
}\nonumber \\* \quad \!\!
 & & \geq 1+\sqrt{\frac{\log n}{c}}\left(1 - \left(1 - Q\left(c,c + \sqrt{c \log n} \right) \right)^{n-1} \right),\IEEEeqnarraynumspace
\end{IEEEeqnarray}
which is the lower bound $\tilde{l}(n,c)$ defined in \eqref{eq:LBandUBinCHomoRxInftyFieldsize}.
\end{IEEEproof}

\section{Proofs from \S \ref{sec:deponnumrx}}
\label{app:deponnumrx}

\begin{IEEEproof}[Proof of Prop.\ \ref{prop:dlcasymptoticinn}]
Fix integer $r \geq 1$.  The de la Cal's lower bound on $\Ebb\left[\left(\tilde{Y}_{n:n}^{(c)}\right)^r\right]$ may be written as (Prop.\ \ref{prop:boundsOnEY3})
\begin{IEEEeqnarray}{rCl}
\IEEEeqnarraymulticol{3}{l}
{\Ebb\left[\left(\tilde{Y}_{n:n}^{(c)}\right)^r\right] 
}\nonumber \\* \quad
& \geq & t - \left(t - \Ebb\left[\left(\tilde{Y}^{(c)}\right)^r \right] \right) \Pbb\left(\left(\tilde{Y}_{n-1:n-1}^{(c)} \right)^r \leq t \right) \nonumber \\
& = & t - \left(t - \Ebb\left[\left(\tilde{Y}^{(c)}\right)^r \right] \right) F_{\tilde{Y}_{n-1:n-1}^{(c)}}\left( t^{\frac{1}{r}} \right) 
\end{IEEEeqnarray}
for any $t \geq  \Ebb\left[\left(\tilde{Y}^{(c)}\right)^r \right]$.  We establish the asymptotic tightness of this lower bound as $n \to \infty$ by showing the existence of $t_n \to \infty$ such that
\begin{equation}
\label{eq:AsyminnTightdelaCal}
\lim_{n \to \infty} \frac{t_n - \left(t_n - \Ebb\left[\left(\tilde{Y}^{(c)}\right)^r \right] \right) 
F_{\tilde{Y}_{n-1:n-1}^{(c)}}\left( t_{n}^{\frac{1}{r}} \right)}
{\Ebb\left[\left(\tilde{Y}_{n:n}^{(c)}\right)^r\right]} = 1.
\end{equation}
From Prop.\ \ref{prop:normalizingConstantsCunchangingGamma}, we have (with the normalizing sequence $(a_{n}, b_{n})$ from \eqref{eq:bnsequencegamma})
\begin{eqnarray}
\label{eq:EVTindelaCal}
F_{\tilde{Y}_{n-1:n-1}^{(c)}}\left( t_{n}^{\frac{1}{r}} \right)
& \stackrel{}{\sim} & \Lambda\left( \frac{t_n^{\frac{1}{r}} - b_{n-1}}{a_{n-1}} \right),
\end{eqnarray}
as $n \to \infty$. Here the notation $\sim$ indicates asymptotic (in $n$) equivalence between functions, namely $f(n) \sim g(n)$ means $\lim_{n \to \infty} f(n)/g(n) = 1$. Since Lem.\ \ref{lem:momentConvergence} (as well as Props.\ \ref{prop:normalizingConstantsCunchangingGamma} and \ref{prop:momentConvergenceResnick1987}) implies that $\Ebb\left[\left(\tilde{Y}_{n:n}^{(c)}\right)^r\right] \sim b_{n}^{r}$,
to show \eqref{eq:AsyminnTightdelaCal}, we need to show its numerator can be asymptotically (in $n$) equal to $b_{n}^{r}$ by appropriately choosing $t_{n}$. Specifically, let $t_{n} = \left( a_{n-1} \left( - \frac{1}{2} \log\log(n-1)\right) + b_{n-1} \right)^{r}$. Note that with this choice of $t_{n}$, for all $n \in \Nbb$, $t_{n}$ grows faster than, say, $\left(\frac{1}{2} \log(n-1) \right)^{r}$, and thus there exists an $N_{r} \geq 1 + \erm^{2 \left( \Ebb\left[ \left(\tilde{Y}^{(c)}\right)^{r} \right] \right)^{\frac{1}{r}}}$ such that $t_n \geq \Ebb\left[\left(\tilde{Y}^{(c)}\right)^r \right]$ for all $n \geq N_{r}$ meaning the bound is valid. Now observe $t_{n} \sim b_{n-1}^{r} \sim b_{n}^{r}$ and $\Ebb\left[\left(\tilde{Y}_{n:n}^{(c)}\right)^r\right]$ is constant (in $n$), which gives
\begin{equation}
\label{eq:numPart1}
\lim_{n \to \infty} \frac{t_{n}}{b_{n}^{r}} = 1 = \lim_{n \to \infty} \frac{t_{n} - \Ebb\left[\left(\tilde{Y}_{n:n}^{(c)}\right)^r\right]}{b_{n}^{r}}.
\end{equation}
Substitution of $t_{n}$ into \eqref{eq:EVTindelaCal} yields
\begin{equation}
\label{eq:numPart2}
F_{\tilde{Y}_{n-1:n-1}^{(c)}}\left( t_{n}^{\frac{1}{r}} \right)
\stackrel{}{\sim} \Lambda\left( - \frac{1}{2} \log\log(n-1) \right)
\stackrel{}{\to} 0,
\end{equation}
as $n \to \infty$.  Observing $b_{n}^{r} \to \infty$ as $n \to \infty$, it follows that \eqref{eq:numPart1} and \eqref{eq:numPart2} suffices to show \eqref{eq:AsyminnTightdelaCal}.
\end{IEEEproof}

\begin{IEEEproof}[Proof of Prop.\ \ref{prop:rossasymptoticinn}]
Fix integer $r \geq 1$, and set $s_n = b_n^r$ for $b_n$ in \eqref{eq:bnsequencegamma}.  Our obligation is to show
\begin{equation}
\lim_{n \to \infty} \frac{ b_n^r + n \int_{b_n^r}^{\infty} \Pbb\left( \left(\tilde{Y}^{(c)}\right)^r  > t \right) \drm t}{\Ebb\left[\left(\tilde{Y}_{n:n}^{(c)}\right)^r\right]} = 1.
\end{equation}
First, by Lemma \ref{lem:momentConvergence}
\begin{equation}
\lim_{n \to \infty} \frac{\Ebb\left[\left(\tilde{Y}_{n:n}^{(c)}\right)^r\right]}{b_n^r} = 1.
\end{equation}
Thus, it suffices to show 
\begin{IEEEeqnarray}{rCl}
& & \lim_{n \to \infty} \frac{n}{b_n^r} \int_{b_n^r}^{\infty} \Pbb\left( \left(\tilde{Y}^{(c)}\right)^r  > t \right) \drm t \nonumber \\
& = & \lim_{n \to \infty} \frac{\int_{b_n^r}^{\infty} \Pbb\left( \left(\tilde{Y}^{(c)}\right)^r  > t \right) \drm t}{\frac{b_n^r}{n}} = 0.
\end{IEEEeqnarray}
Observe the limit of both numerator and denominator in the above expression are zero, and thus we may apply L'H\^opital's rule, and then Leibniz's rule:
\begin{IEEEeqnarray}{rCl}
\label{eq:tytywweeweerrttt}
& & \lim_{n \to \infty} \frac{n}{b_n^r} \int_{b_n^r}^{\infty} \Pbb\left( \left(\tilde{Y}^{(c)}\right)^r  > t \right) \drm t  \nonumber \\
& = & \lim_{n \to \infty} \frac{\frac{\drm}{\drm n} \int_{b_n^r}^{\infty} \Pbb\left( \left(\tilde{Y}^{(c)}\right)^r  > t \right) \drm t}{\frac{\drm}{\drm n}\left(\frac{b_n^r}{n}\right)}  \nonumber \\
& = & \lim_{n \to \infty} \frac{-\Pbb\left( \tilde{Y}^{(c)}  > b_n \right) \frac{\drm}{\drm n} b_n^r}{\frac{\drm}{\drm n}\left(\frac{b_n^r}{n}\right)}.
\end{IEEEeqnarray}
As will be shown below
\begin{equation}
\label{eq:tyhhhgbbbnnnn}
\lim_{n \to \infty} n\Pbb\left( \tilde{Y}^{(c)}  > b_n \right) = 1.
\end{equation}
Substituting \eqref{eq:tyhhhgbbbnnnn} into \eqref{eq:tytywweeweerrttt} gives
\begin{IEEEeqnarray}{rCl}
\IEEEeqnarraymulticol{3}{l}
{\lim_{n \to \infty} \frac{n}{b_n^r} \int_{b_n^r}^{\infty} \Pbb\left( \left(\tilde{Y}^{(c)}\right)^r  > t \right) \drm t 
}\nonumber \\* \quad
& = & \lim_{n \to \infty} \frac{-\frac{1}{n} \frac{\drm}{\drm n} b_n^r}{\frac{\drm}{\drm n}\left(\frac{b_n^r}{n}\right)} 
= \lim_{n \to \infty} \frac{-\frac{1}{n} \frac{\drm}{\drm n} b_n^r}{\frac{1}{n} \frac{\drm}{\drm n} b_n^r - \frac{b_n^r}{n^2}} \nonumber \\
& = &\lim_{n \to \infty} \frac{1}{\frac{b_n^r}{n\frac{\drm}{\drm n} b_n^r} - 1} 
= \lim_{n \to \infty} \frac{1}{\frac{b_n}{r n\frac{\drm}{\drm n} b_n} - 1}.
\end{IEEEeqnarray}
To establish the desired limit of $0$, it suffices to show the first term in the denominator grows to infinity in $n$:
\begin{IEEEeqnarray}{rCl}
\IEEEeqnarraymulticol{3}{l}
{\lim_{n \to \infty} \frac{b_n}{r n\frac{\drm}{\drm n} b_n} 
}\nonumber \\* \quad
& = & \lim_{n \to \infty} \frac{-\log \Gamma(c) + \log n + (c-1) \log \log n}{r \left(1 + \frac{c-1}{\log n}\right)} = \infty.\IEEEeqnarraynumspace
\end{IEEEeqnarray}
It remains to establish \eqref{eq:tyhhhgbbbnnnn}.  It is well-known that
\begin{equation}
\lim_{s \to \infty} \frac{Q(c,s)}{\frac{s^{c-1}\erm^{-s}}{\Gamma(c)}} = 1, 
\end{equation}
where $Q(c,s) = \Pbb(\tilde{Y}^{(c)} > s)$.  Substituting this into \eqref{eq:tyhhhgbbbnnnn} gives
\begin{IEEEeqnarray}{rCl}
\IEEEeqnarraymulticol{3}{l}
{\lim_{n \to \infty} n Q(c, b_n)
= \lim_{n \to \infty} n \frac{b_n^{c-1}\erm^{-b_n}}{\Gamma(c)}
= \lim_{n \to \infty}\left( \frac{b_n}{\log n} \right)^{c-1}
}\nonumber \\* \quad
= \lim_{n \to \infty} \left( - \frac{\log \Gamma(c)}{\log n} + 1 + (c-1) \frac{\log \log n}{\log n} \right)^{c-1} = 1.\IEEEeqnarraynumspace
\end{IEEEeqnarray}
We observe that for any given $n$ the optimal $s^*$ satisfies $n \Pbb(\tilde{Y}^{(c)} > s^*) = 1$; it is in this sense that $s_n = b_n^r$ is not only sufficient for the bound to be asymptotically tight, but also the asymptotically optimal choice for $s$.
\end{IEEEproof}

\section{Lower and upper bounds on the expected maximum order statistic}
\label{app:delaCalAndRossBounds}
This appendix presents two inequalities on the expected maximum order statistic, $\Ebb[Z_{n:n}]$: a lower bound due to de la Cal and C\'{a}rcamo \cite{CalCar2005} and an upper bound due to Ross and Pek\"{o}z \cite{RosPek2007}.  They have been used in several of our sections. It is important to note that $i)$ they hold for any finite parameters ($c$, $n$), and $ii)$ partially informed by the EVT results in \S\ref{sec:deponnumrx} regarding the choice of free parameters, they can both be shown to be (almost) asymptotically tight in $c$ (when $r = 1$, see Prop.\ \ref{prop:LBandUBinCHomoRxInftyFieldsize}), or in $n$ (Props.\ \ref{prop:dlcasymptoticinn} and \ref{prop:rossasymptoticinn}).

We present the inequalities, then give a pertinent example of their application, including an illustration (Prop.\ \ref{prop:rossdlcexpexample}) of how to use the bounds to establish asymptotic tightness (in $n$).

\begin{proposition}[\cite{CalCar2005} Thm.\ 13]
\label{prop:delacal}
Let $(Z_1,\ldots,Z_n)$ be independent and identically distributed (not necessarily nonnegative) RVs.  Then:
\begin{equation}
\label{eq:delaCallb}
\Ebb[Z_{n:n}] \geq t - (t-\mu_Z) F_Z^{n-1}(t), ~ \forall t \geq \mu_Z,
\end{equation}
where $\mu_Z, F_Z$ denote the mean and CDF of each $Z_j$.
\end{proposition}
Prop.\ \ref{prop:delacal} can be extended to the case of RVs that are independent but not necessarily identically distributed.  
\begin{proposition}[\cite{RosPek2007} \S4.6]
\label{prop:ross}
Let $(Z_1,\ldots,Z_n)$ be nonnegative (not necessarily independent) RVs.  Then:
\begin{equation} 
\label{eq:rossub}
\Ebb[Z_{n:n}] \leq s + \sum_{j=1}^{n} \int_s^{\infty} (1-F_{Z_j}(z)) \drm z, ~ \forall s \geq 0,
\end{equation}
where $F_{Z_j}$ is the CDF of $Z_j$.  The bound is tightest for $s^*$ satisfying 
\begin{equation}
\label{eq:rossuboptcond}
\sum_{j=1}^{n} (1 - F_{Z_j}(s^*)) = 1.
\end{equation}
\end{proposition}
Since the left side of \eqref{eq:rossuboptcond} is a decreasing function in $s$, it follows that the solution of \eqref{eq:rossuboptcond} is unique and may be easily found numerically via bisection search.
Both bounds can be (or have been) shown to have a counterpart for bounding the expected minimum order statistic $\Ebb[Z_{1:n}]$.  Observe that both inequalities have a degree of freedom in the form of a parameter that may be chosen either to make the bound as tight as possible, or put the bound in a certain form.   We illustrate the bounds for the exponential case.   Specifically, let $(\tilde{X}_1,\ldots,\tilde{X}_n)$ be iid exponential RVs with unit rate.  There is no loss in generality in setting the rate $\lambda = 1$ since $\frac{1}{\lambda} \tilde{X} \sim \mathrm{Exp}(\lambda)$, and in particular $\Ebb[ \max_{j \in [n]} \mathrm{Exp}_j(\lambda)] = \frac{1}{\lambda} \Ebb[\tilde{X}_{n:n}]$.  

\begin{proposition}
\label{prop:rossdlcexpexample}
Let $(\tilde{X}_1,\ldots,\tilde{X}_n)$ be iid unit rate exponential RVs.  Then the functions $l_{\tilde{X}}(t,n)$ and $u_{\tilde{X}}(s,n)$ below satisfy $l_{\tilde{X}}(t,n) \leq \Ebb[\tilde{X}_{n:n}] \leq u_{\tilde{X}}(s,n)$ for all $s \geq 0$, $t \geq 1$, and $n \in \Nbb$, where:
\begin{eqnarray}
l_{\tilde{X}}(t,n) & \equiv & t - (t - 1)(1 - \erm^{-t})^{n-1} \nonumber \\
u_{\tilde{X}}(s,n) & \equiv & s + n \erm^{- s}.
\label{eq:expexalb1}
\end{eqnarray}
For the choice $t_n =1 + \log n - \varepsilon_n$ where $\varepsilon_n = o(\log n)$ and $\varepsilon_n = \omega (1)$ we have
\begin{equation}
\label{eq:expexatn}
l_{\tilde{X}}(t_n,n)  \equiv  l_{\tilde{X}}(n) \! = \! 1 + (\log n -  \varepsilon_n)\left( 1 - \left(1 - \frac{1}{n \erm} \erm^{\varepsilon_n}\right)^{n-1} \right).
\end{equation}
Furthermore, the optimal value of $s$ is $s^*_n = \log n$, for which
\begin{equation}
\label{eq:expexau2}
u_{\tilde{X}}(s^*_n,n) \equiv u_{\tilde{X}}(n) = 1+\log n.
\end{equation}
Finally, the bounds $(l_{\tilde{X}}(n),u_{\tilde{X}}(n))$ are asymptotically tight in $n$ in that, as $n \to \infty$,
\begin{equation}
\frac{l_{\tilde{X}}(n)}{\log n} \to 1, ~~ 
\frac{u_{\tilde{X}}(n)}{\log n} \to 1.
\end{equation}
These together imply $\Ebb[{\tilde{X}}_{n:n}]/\log n \to 1$. 
\end{proposition}

\begin{IEEEproof}
Specializing the lower (upper) bound in Prop.\ \ref{prop:delacal} (\ref{prop:ross}) to the exponential case gives \eqref{eq:expexalb1}.  Substituting the given form for $t_n$ into the lower bound gives \eqref{eq:expexatn}.  Differentiation of the convex function $u_{\tilde{X}}(s,n)$ with respect to $s$, equating with zero, and solving for $s$ gives $s^*_n = \log n$, and $u_{\tilde{X}}(n)$ in \eqref{eq:expexau2}.  The limit of $u_{\tilde{X}}(n)/\log n$ as $n \to \infty$ is immediate.  It remains to establish that the limit of 
\begin{IEEEeqnarray}{rCl}
\IEEEeqnarraymulticol{3}{l}
{\frac{l_{\tilde{X}}(n)}{\log n}  = \frac{1 + (\log n -  \varepsilon_n)\left( 1 - \left(1 - \frac{1}{n \erm} \erm^{\varepsilon_n}\right)^{n-1} \right)}{\log n} 
}\nonumber \\* \quad
& = & \frac{1}{\log n} + \left(1 - \frac{\varepsilon_n}{\log n} \right)\left( 1 - \left(1 - \frac{1}{n \erm} \erm^{\varepsilon_n}\right)^{n-1} \right) \IEEEeqnarraynumspace
\end{IEEEeqnarray}
as $n \to \infty$ is $1$.  Since $\varepsilon_n = o(\log n)$, it suffices to show that the limit of
\begin{IEEEeqnarray}{rCl}
\label{eq:expexalim2}
\left(1 - \frac{1}{n \erm} \erm^{\varepsilon_n}\right)^{n-1} 
& = & \left( 1 - \frac{\erm^{\varepsilon_n-1}}{n}\right)^{n-1} \nonumber \\
& = & \left[ \left( 1 - \frac{\erm^{\varepsilon_n-1}}{n} \right)^{n \erm^{1-\varepsilon_n}}  \right]^{\frac{n-1}{n} \erm^{\varepsilon_n-1} } 
\end{IEEEeqnarray}
as $n \to \infty$ is zero.  First consider $\erm^{\varepsilon_n-1}/n$.  For any positive sequence $\{a_n\}$:
\begin{equation}
\lim_{n \to \infty} \frac{\log a_n}{\log n} = -1  \Rightarrow  \lim_{n \to \infty} \log a_n = -\infty ~ \Leftrightarrow ~ \lim_{n \to \infty} a_n = 0.
\end{equation}
For $a_n = \erm^{\varepsilon_n-1}/n$ with $\varepsilon_n = o(\log n)$ this gives directly that $\erm^{\varepsilon_n-1}/n \to 0$.
Recall $\lim_{b \to \infty} (1-1/b)^b = \erm^{-1}$.  Choosing $b = n \erm^{1-\varepsilon_n}$ it follows that
\begin{equation}
\label{eq:expexalim1}
\lim_{n \to \infty} \left( 1 - \frac{\erm^{\varepsilon_n-1}}{n} \right)^{n \erm^{1-\varepsilon_n}} = \erm^{-1}.
\end{equation}
Finally, the limit as $n \to \infty$ of the logarithm of the rightmost expression in \eqref{eq:expexalim2} equals (applying \eqref{eq:expexalim1}):
\begin{IEEEeqnarray}{rCl}
\IEEEeqnarraymulticol{3}{l}
{\left( \lim_{n \to \infty} \frac{n-1}{n} \erm^{\varepsilon_n-1} \right) \left( \lim_{n \to \infty} \log \left( 1 - \frac{\erm^{\varepsilon_n-1}}{n} \right)^{n \erm^{1-\varepsilon_n}} \right)
}\nonumber \\* \qquad
& = & - \lim_{n \to \infty} \frac{n-1}{n} \erm^{\varepsilon_n-1} = -\infty.
\end{IEEEeqnarray}
It follows that the limit of \eqref{eq:expexalim2} as $n \to \infty$ is zero.

\end{IEEEproof}

Using the min-max identity (Prop.\ \ref{prop:minmax}), the fact that $\min_{j \in \Amc } \tilde{X}_j \sim \mathrm{Exp}(s)$ for any $\Amc \in [n]_s$, and the binomial coefficient absorbtion identity gives
\begin{IEEEeqnarray}{rCl}
\Ebb[\tilde{X}_{n:n}] & = & \sum_{s=1}^n (-1)^{s+1} \sum_{\Amc \in [n]_s} \Ebb \left[ \min_{j \in \Amc} \tilde{X}_j \right] \nonumber \\
& = & \sum_{s=1}^n (-1)^{s+1} \binom{n}{s} \frac{1}{s} = \sum_{s=1}^n \frac{1}{s} = H_n,
\end{IEEEeqnarray}
where $H_n$ denotes the $n^{\rm th}$ harmonic number.  We select $\varepsilon_n = \log \log n$ under which the lower bound becomes
\begin{equation}
l_{\tilde{X}}(n) = 1 + (\log n - \log \log n)\left(1-\left(1 - \frac{\log n}{n} \right)^{n-1}\right).  
\end{equation}
These bounds are shown in Fig.\ \ref{fig:maxexpexample}.  

\begin{figure}[htbp]
\centering
\includegraphics[width=0.34\textwidth]{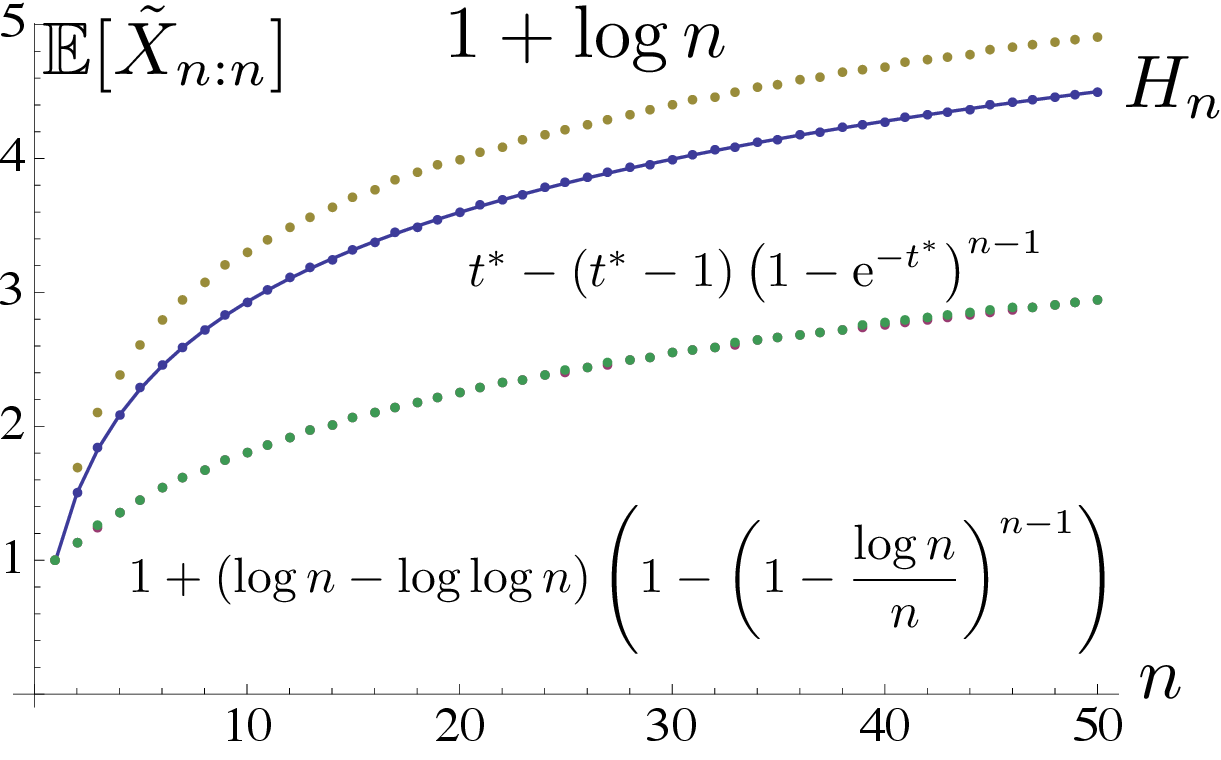}
\caption{The exact value (blue, joined) for $\Ebb[\tilde{X}_{n:n}]$, the optimal Ross's upper bound (yellow), de la Cal's lower bound with $t_{n} = 1 + \log n - \log\log n$ (red), and with (numerically) optimal free parameter $t^{*}$ (green), for $\tilde{X}_1,\ldots,\tilde{X}_n$ iid unit rate exponentials.}
\label{fig:maxexpexample}
\end{figure}

\begin{IEEEbiographynophoto}{Nan Xie}
(S'10-M'14) received his B.S. degree in communication engineering and M.S. degree in circuits and systems, both from Wuhan University, China, in 2003 and 2006, respectively. From 2006 to 2007 he was with the Department of Electrical and Computer Engineering at the University of Florida, Gainesville, FL, USA. Since 2008 he has been with the Department of Electrical and Computer Engineering at Drexel University, Philadelphia, PA, USA, where he earned his Ph.D.\ degree in electrical engineering in 2014. His research interests are focused on a better understanding of metrics such as delay, stability, throughput, and fairness, in the context of design and performance optimization of networks.
\end{IEEEbiographynophoto}

\begin{IEEEbiographynophoto}{Steven Weber}
  (S'97-M'03-SM'11) received the B.S. degree in 1996 from Marquette University, Milwaukee, WI, USA, in 1996 and the M.S. and Ph.D.\ degrees from The University of Texas at Austin, TX, USA, in 1999 and 2003, respectively.
  He joined the Department of Electrical and Computer Engineering at Drexel University, Philadelphia, PA, USA, in 2003, where he is currently a Professor.
  His research interests are centered around mathematical modeling of computer and communication networks.
\end{IEEEbiographynophoto}

\end{document}